\documentclass[a4paper,11pt,leqno]{amsart}
\usepackage{amsmath}
\usepackage{amsfonts}
\usepackage{eucal}
\usepackage{amsthm}
\numberwithin{equation}{section}

\def\d{\delta}

\def\re{\mathbb{R}}
\def\co{\mathbb{C}}
\def\ze{\mathbb{Z}}

\newcommand{\Slash}[1]{{\ooalign{\hfil#1\hfil\crcr\raise.167ex\hbox{/}}}}

\newtheorem{thm}{Theorem}
\newtheorem{lem}[thm]{Lemma}
\newtheorem{prop}[thm]{Proposition}
\newtheorem{cor}[thm]{Corollary}

\theoremstyle{definition}
\newtheorem{defn}{Definition}

\theoremstyle{remark}
\newtheorem{rem}{Remark\!}

\begin{document}

\title[Topological Current in Fractional Chern Insulators]{Topological Current\\
in Fractional Chern Insulators}
\author[T. Koma]{Tohru Koma}
\address[Koma]{Department of Physics, Gakushuin University, 
Mejiro, Toshima-ku, Tokyo 171-8588, JAPAN}
\email{tohru.koma@gakushuin.ac.jp}
\date{}

\begin{abstract}
We consider interacting fermions in a magnetic field on a two-dimensional lattice 
with the periodic boundary conditions. 
In order to measure the Hall current, we apply an electric potential with a compact support. 
Then, due to the Lorentz force, the Hall current appears along the equipotential line. 
Introducing a local current operator at the edge of the potential, we derive the Hall conductance 
as a linear response coefficient. For a wide class of the models, 
we prove that if there exists a spectral gap above the degenerate ground state, 
then the Hall conductance of the ground state is fractionally quantized 
without averaging over the fluxes. This is an extension of the topological argument for 
the integrally quantized Hall conductance in noninteracting fermion systems on lattices. 
\end{abstract}

\maketitle 

\section{Introduction}
\label{sec:intro}


The quantum Hall effect \cite{KDP,KawajiWakabayashi} is one of the most astonishing phenomena 
in condensed matter physics. 
In fact, the quantization of the Hall conductance is 
surprisingly insensitive to disorder and interactions.\footnote{A mathematical argument was 
given in \cite{Koma3}.} 
This robustness of the quantization reflects the topological nature of the Hall conductance formula. 
Namely, the topological structure never changes 
for continuously deforming disorder or interactions. 

The first discovery was that Thouless, Kohmoto, Nightingale and den Nijs \cite{TKNN} found 
that the Hall conductance is quantized to a nontrivial integer in a two-dimensional noninteracting electron 
system in a periodic potential and a magnetic field.    
Later, Kohmoto \cite{Kohmoto} realized that the integral quantization of the Hall conductance is due to 
the topological nature of the Hall conductance formula. Namely, the integer is nothing but 
the Chern number which is the winding number for the quantum mechanical $U(1)$ phases 
of the wavefunctions on the magnetic Brillouin zone which can be identified with a two-dimensional torus.  
More precisely, the $U(1)$ phases are nontrivially twisted on the torus, and the winding number is given by  
the number of times when the phases rotate on the $U(1)$ circle, and must be an integer.

Although realistic systems have disorder and interactions, they treated 
translationally invariant noninteracting systems only.  
Instead of resorting to the topological argument in usual differential geometry, 
the method of noncommutative geometry \cite{Connes} was found to be useful to show 
the integral quantization of the Hall conductance for 
noninteracting systems which do not necessarily require translation invariance \cite{BVS}.
(See also \cite{ASS,AG,ElSch,Koma2}.)  
However, its extension to interacting systems still remains an open problem. 

Another tricky approach which we will focus on the present paper was 
introduced by Niu, Thouless and Wu \cite{NTW}. They considered a generic quantum Hall system 
which is allowed to have disorder and interactions.   
They imposed the twisted boundary conditions with angles $\phi_1$ and $\phi_2$ 
at the two boundaries in a two dimensional system, and assumed 
that the system exhibits a nonvanishing uniform spectral gap above a $q$-hold degenerate ground state.   
Under the assumptions, they showed that the Hall conductance averaged over the two angles 
is fractionally quantized. (See also \cite{Koma1}.)
 
Clearly, their method is artificial, and their result does not implies 
that the Hall conductance for fixed twisted angles is fractionally quantized. 
But one can expect that the twisted boundary conditions do not affect the quantization of 
the Hall conductance. Quite recently, Hastings and Michalakis \cite{HastingsMichalakis} 
treated an interacting lattice electron system which does not necessarily require translation invariance, 
and showed that the Hall conductance is quantized to an integer irrespective of the twisted angles 
under the assumptions that the system shows a nonvanishing spectral gap above 
the unique ground state. 
 
In the present paper, we consider interacting fermions in a magnetic field on a two-dimensional lattice 
with the periodic boundary conditions. As is well known, in generic lattice systems, 
a constant electric field is known to be useless in order to measure the conductance 
because the absolutely continuous spectrum of the unperturbed Hamiltonian often changes to 
a pure point spectrum. Therefore, one cannot expect that there appears an electric flow 
due to the constant electric field. 

In the present paper, we overcome the difficulty as follows: 
In order to measure the Hall current, we apply an electric potential with a compact support, 
instead of the linear electric potential or the time dependent vector potential which 
yield the constant electric field. 
Then, due to the Lorentz force, the Hall current appears along the equipotential line. 
Introducing a local current operator at the edge of the potential, we derive the Hall conductance 
as a linear response coefficient. For a wide class of the models, 
we prove that if there exists a spectral gap above the degenerate ground state, 
then the Hall conductance of the ground state is fractionally quantized 
without averaging over the fluxes. This is an extension of the topological argument for 
the integrally quantized Hall conductance in noninteracting fermion systems on lattices.

The present paper is organized as follows. In Sec.~\ref{sec:Model}, we describe the model, 
and state our main Theorem~\ref{thm:main}. In Sec.~\ref{sec:LCO}, we define current operators.  
Using these expressions, we derive the Hall conductance formula as a linear response coefficient 
in Sec.~\ref{sec:LR}. Following the idea by Niu, Thouless and Wu, we twist the phases of 
the hopping amplitudes of the model, and we present Theorem~\ref{thm:twistDegeneGap} below 
about the stability of the spectral gap above the degenerate ground state against twisting the boundary 
conditions in Sec.~\ref{sec:TP}. In Sec.~\ref{sec:TopoInv}, 
the Hall conductance averaged over the phases is shown to be quantized to a fraction by using 
the standard topological argument. 
In Sec.~\ref{sec:topocurr}, we consider deformation of the current operators, and 
prove that the statement of our main Theorem~\ref{thm:main} is still valid 
for the deformed current operators so obtained. In other words, 
the fractional quantization of the Hall condudctance is robust against such deformation of 
the current operators. 
Sections~\ref{CurrentCurrentCorr}, \ref{sec:TPDHallcond} and \ref{sec:proofmain} are devoted to 
the proof of the main theorem. The proof of Theorem~\ref{thm:twistDegeneGap} 
about the stability of the spectral gap 
is fairly lengthy, and therefore given in Appendix~\ref{ProoftwistDegeneGap}.  
Appendices~\ref{CorrHallcon} and \ref{Proofa0tdiffinequal}--\ref{ProofTopoInvUnch} are devoted 
to technical estimates.

\section{Lattice Fermions in Two Dimensions}
\label{sec:Model}

Let us describe the model which we will treat in the present paper. 
Consider a rectangular box, 
\[
\Lambda:=\biggl[-\frac{L^{(1)}}{2},\frac{L^{(1)}}{2}\biggr]\times
\biggl[-\frac{L^{(2)}}{2},\frac{L^{(2)}}{2}\biggr],
\]
which is a finite subset of the two-dimensional square lattice $\mathbb{Z}^2$. 
Here both of $L^{(1)}$ and $L^{(2)}$ are taken to be a positive even integer.  
We consider interacting fermions on the lattice $\Lambda$ with the periodic boundary 
conditions. The Hamiltonian is given by  
\begin{equation}
H_0^{(\Lambda)}=\sum_{x,y\in\Lambda}t_{x,y}c_x^\dagger c_y
+\sum_{I\ge 1}\sum_{x_1,x_2,\ldots,x_I\in\Lambda}U_{x_1,x_2,\ldots,x_I} n_{x_1}n_{x_2}\cdots n_{x_I},
\label{ham}
\end{equation}
where $c_x^\dagger, c_x$ are, respectively, 
the creation and annihilation fermion operators at the site $x\in\Lambda$, 
the hopping amplitudes $t_{x,y}$ are complex numbers which satisfy the Hermitian conditions, 
\[
t_{y,x}=t_{x,y}^\ast, 
\]
and the coupling constants $U_{x_1,x_2,\ldots,x_I}$ of the interactions are real; 
$n_x=c_x^\dagger c_x$ is the number operator of the fermion at the site $x\in\Lambda$. 
We assume that both of the hopping amplitudes and the the interactions are 
of finite range in the sense of the graph theoretic distance.   

Clearly, the Hamiltonian $H_0^{(\Lambda)}$ of (\ref{ham}) commutes with the total number 
operator $\sum_{x\in\Lambda}n_x$ of the fermions for a finite volume $|\Lambda|<\infty$. 
We denote by $H_0^{(\Lambda,N)}$ the restriction of $H_0^{(\Lambda)}$ onto 
the eigenspace of the total number operator with the eigenvalue $N$. 

We require the existence of a ``uniform gap" above the sector of the 
ground state of the Hamiltonian $H_0^{(\Lambda,N)}$. 
Since we will take the infinite-volume limit $\Lambda\rightarrow \ze^2$, 
we keep the filling factor $\nu=N/|\Lambda|$ to be a nonzero value in the limit. 
The precise definition of the ``uniform gap" is:  

\begin{defn}
\label{definition}
We say that there is a uniform gap above the sector of the ground state
if the spectrum $\sigma(H_0^{(\Lambda,N)})$ of the Hamiltonian $H_0^{(\Lambda,N)}$ 
satisfies the following conditions: 
The ground state of the Hamiltonian $H_0^{(\Lambda,N)}$ is 
$q$-fold (quasi)degenerate in the sense that there are $q$ eigenvalues, 
$E_{0,1}^{(N)},\ldots, E_{0,q}^{(N)}$, in the sector of the ground state 
at the bottom of the spectrum of $H_0^{(\Lambda,N)}$ such that 
\[
\d E:=\max_{m,m'}\bigl\{\bigl|E_{0,m}^{(N)}-E_{0,m'}^{(N)}\bigr|\bigr\}
\rightarrow 0\quad\mbox{as }\
|\Lambda|\rightarrow\infty.
\label{deltaE} 
\]
Further the distance between the spectrum, $\{E_{0,1}^{(N)},\ldots, E_{0,q}^{(N)}\}$,   
of the ground state and the rest of the spectrum is larger than 
a positive constant $\Delta E$ which is independent of the volume $|\Lambda|$. 
Namely there is a spectral gap $\Delta E$ above the sector of the ground state. 
\end{defn}
 
We prove: 

\begin{thm}
\label{thm:main}
We assume that the ground state of the Hamiltonian $H_0^{(\Lambda,N)}$ is $q$-fold 
(quasi)degenerate, and there is a uniform gap above the sector of the ground state. 
Further, we assume that, even when changing the boundary conditions, 
there is no other infinite-volume ground state except the infinite-volume ground states 
which are derived from the $q$ ground-state vectors of the Hamiltonian $H_0^{(\Lambda,N)}$. 
Then, the Hall conductance $\sigma_{12}$ is fractionally quantized as 
\begin{equation}
\label{FQHC}
\sigma_{12}=\frac{1}{2\pi}\frac{p}{q}
\end{equation}
with some integer $p$ in the infinite-volume limit. 
\end{thm}

The precise definition of the Hall conductance $\sigma_{12}$ is given by (\ref{sigma12}) 
in Sec.~\ref{sec:LR} below. 

\begin{rem}
We can relax the assumption on the number of the ground states so that 
when twisting the phases of the hopping amplitudes at the boundaries, 
there appears no other infinite-volume ground state at energies lower than or close to 
the energies of the $q$-fold ground state,    
except the infinite-volume ground states 
which are derived from the $q$ ground-state vectors of the Hamiltonian $H_0^{(\Lambda,N)}$. 

If one includes the charge $e$ of the electron and the Planck constant $h$ in the computations, 
then the Hall conductance (\ref{FQHC}) is written in the usual form, 
$$
\sigma_{12}=\frac{e^2}{h}\frac{p}{q}. 
$$

The possibility of the spectral gap above the $q$-fold degenerate ground state 
in quantum Hall systems was treated by \cite{Koma4} in a mathematical manner. 
The numerical evidences of the existence of nontrivial fractional 
Chern insulators were shown in \cite{SGSS,NSCM,RB,WBR}.  
\end{rem}

\section{Local Current Operators}
\label{sec:LCO}

In order to measure the Hall current, we must define current operators 
by relying on the Ehrenfest's theorem in quantum mechanics. 
More precisely, the velocity of the charged particle which gives the current 
can be determined by differentiating the expectation value of the center of mass 
with respect to time $t$.

\subsection{Current operators for a single particle}

Consider the time evolution of a wavepacket of 
the single particle on the infinite-volume lattice $\ze^2$. 
The single-body Hamiltonian $\mathcal{H}_0$ corresponding to the present system is given by  
\[
(\mathcal{H}_0\psi)(x)=\sum_{y\in\ze^2}t_{x,y}\psi(y)+U_x\psi(x)
\]
for $\psi\in\ell^2(\ze^2)$. The Schr\"odinger equation which determines the time evolution 
of the wavepacket is given by 
\[
i\frac{d}{dt}\psi_t=\mathcal{H}_0\psi_t
\]
for the wavefunction $\psi_t$ at the time $t$. The expectation value of an observable $a$ is 
\[
\langle a\rangle_t:=\frac{\langle \psi_t, a\psi_t\rangle}{\langle \psi_t,\psi_t\rangle},
\]
where the inner product is defined by 
\[
\langle \varphi,\psi\rangle:=\sum_{x\in\ze^2}\varphi(x)^\ast \psi(x) 
\]
for $\varphi,\psi\in\ell^2(\ze^2)$. 

The current due to the motion of the particle is determined by the velocity of the center of mass. 
The center of mass is nothing but the expectation value of the position operator. 
The position operator $X=(X^{(1)},X^{(2)})$ is defined by 
\[
(X^{(j)}\psi)(x)=x^{(j)}\psi(x),\quad j=1,2,
\]
where we have written $x=(x^{(1)},x^{(2)})\in\ze^2$. 
The Ehrenfest's theorem for the position operator is expressed as  
\[
\frac{d}{dt}\langle X^{(j)}\rangle_t=\langle i[\mathcal{H}_0,X^{(j)}]\rangle_t, 
\]
where $[A,B]$ is the commutator of two observables, $A$ and $B$. Therefore, the natural 
definition of the current operator must be 
\begin{equation}
\mathcal{I}^{(j)}=i[\mathcal{H}_0,X^{(j)}].
\label{currentopsingle}
\end{equation}
Clearly, this is an infinite sum. We want to decompose the current operator into 
the sum of local current operators which are much more useful us 
to obtain the Hall conductance formula below. 

In order to define local current operators, we introduce 
an approximate position operator as   
\[
\bigl(X_\ell^{(j)}\psi\bigr)(x)=f_\ell^{(j)}(x)\psi(x),
\]
where the function $f_\ell^{(j)}$ on $\ze^2$ is given by  
\[
f_\ell^{(j)}(x):=\sum_{k=1}^{2\ell}\theta^{(j)}(x;k-\ell)-\ell
\]
with the step function, 
\[
\theta^{(j)}(x;k):=\begin{cases} 1, & \text{$x^{(j)}\ge k$};\\
0, & \text{$x^{(j)}<k$}.
\end{cases}
\]
Clearly, we have $\Vert (X^{(j)}-X_\ell^{(j)})\psi\Vert\rightarrow 0$ as $\ell\rightarrow\infty$ 
for $(X^{(j)}\psi)\in\ell^2(\ze^2)$. 
By replacing $X^{(j)}$ with  $X_\ell^{(j)}$ in the right-hand side of (\ref{currentopsingle}), 
we obtain the approximate current operator as 
\begin{align*}
\mathcal{I}_\ell^{(j)}&=i[\mathcal{H}_0,X_\ell^{(j)}] \\
&=\sum_{k=1}^{2\ell}i[\mathcal{H}_0,\theta^{(j)}(\cdots;k-\ell)].
\end{align*}
The summand in the right-hand side is interpreted as the local current operator. 
Namely, the local current operator $J^{(j)}(k)$ across the $k$-th site in the $j$ direction 
is given by the commutator of the Hamiltonian $\mathcal{H}_0$ and 
the step function $\theta^{(j)}(\cdots;k)$ as 
\[
J^{(j)}(k)=i[\mathcal{H}_0,\theta^{(j)}(\cdots;k)]. 
\]
\subsection{Current operators for many fermions}

The step function is written 
\begin{equation}
\theta^{(j)}(k)=\sum_{x\in\ze^2}\theta^{(j)}(x;k)\; n_x
\label{stepfunc}
\end{equation}
in terms of the number operator $n_x$ of the fermions. 
Therefore, the local current operator $J^{(j)}(k)$ is formally written 
\begin{align}
\label{localcurrent}
J^{(j)}(k)&=i[H_0^{(\ze^2)},\theta^{(j)}(k)]\\ \nonumber
&=i\sum_{x,y\in\ze^2}\sum_{\substack{z\in\ze^2\\ z^{(j)}\ge k}}t_{x,y}[c_x^\dagger c_y,n_z]\\ \nonumber
&=i\sum_{\substack{x\in\ze^2\\ x^{(j)}<k}}\;\sum_{\substack{y\in\ze^2\\ y^{(j)}\ge k}}
(t_{xy}c_x^\dagger c_y-t_{yx}c_y^\dagger c_x),
\end{align}
where we have used the expression (\ref{ham}) of the Hamiltonian, 
and we have written $x=(x^{(1)},x^{(2)})$. 
Although this is clearly an infinite sum, the expression can be justified 
for a localized wavefunction of many fermions.   
For the present finite-volume  lattice $\Lambda$, we can define 
the local current operator $J^{(j)}(k)$ as 
\[
J^{(j)}(k):=i\sum_{\substack{x\in\Lambda\\ x^{(j)}<k}}
\ \sum_{\substack{y\in\Lambda\\ y^{(j)}\ge k}}(t_{xy}c_x^\dagger c_y-t_{yx}c_y^\dagger c_x)
\]
with the periodic boundary conditions because the range of the hopping amplitudes is finite. 
In order to measure the Hall current, we introduce an approximate local current operator as 
\begin{multline}
J_\mathcal{N}^{(1)}(k,\ell)\\ 
:=i\sum_{x^{(1)}<k}\ \sum_{\ell-\mathcal{N}\le x^{(2)}\le \ell+\mathcal{N}}\  
\sum_{y^{(1)}\ge k}\ \sum_{\ell-\mathcal{N}\le y^{(2)}\le \ell+\mathcal{N}}
(t_{xy}c_x^\dagger c_y-t_{yx}c_y^\dagger c_x),
\label{approJ1}
\end{multline}
where $\mathcal{N}$ is a positive integer and $\ell$ is an integer. 
Namely, $\mathcal{N}$ is the cutoff in the second direction.

\section{Linear Response}
\label{sec:LR}

In order to measure the Hall current, we apply an electric potential to the present system. 
For this purpose, we introduce a $2\mathcal{M}\times 2\mathcal{M}$ square box as 
\[
\Gamma_\mathcal{M}(k,\ell):=\{x=(x^{(1)},x^{(2)})\; |\; k-\mathcal{M}\le x^{(1)}\le k+\mathcal{M},
\ell\le x^{(2)}\le \ell+2\mathcal{M}\}
\]
for a positive integer $\mathcal{M}$, and the sum of the number operators $n_x$ on the box as  
\begin{equation}
\label{chiGamma}
\chi(\Gamma_\mathcal{M}(k,\ell)):=\sum_{x\in\Gamma_\mathcal{M}(k,\ell)}n_x. 
\end{equation}
The latter yields the unit voltage difference from the outside of the region. 
The current operator segment $J_\mathcal{N}^{(1)}(k,\ell)$ of (\ref{approJ1})  
goes across the bottom side of the square box $\Gamma_\mathcal{M}(k,\ell)$. 
When we change the potential energy inside the box by using 
the potential operator $\chi(\Gamma_\mathcal{M}(k,\ell))$, 
the Hall current is expected to appear along the boundary of the box $\Gamma_\mathcal{M}(k,\ell)$ 
due to the Lorentz force. The Hall current is measured by the observable $J_\mathcal{N}^{(1)}(k,\ell)$. 
More precisely, we use a time-dependent electric potential.  
The Hamiltonian having the electric potential on the region $\Gamma_\mathcal{M}(k,\ell)$ is given by 
\begin{equation}
H^{(\Lambda)}(t):=H_0^{(\Lambda)}+\lambda W(t)
\end{equation}
with the perturbed Hamiltonian, 
\begin{equation}
W(t)=e^{\eta t}\chi(\Gamma_\mathcal{M}(k,\ell)). 
\label{Wt}
\end{equation}
Here the voltage difference $\lambda$ is a real parameter, 
and the adiabatic parameter $\eta$ is a small positive number.  
We switch on the electric potential at the initial time $t=-T$ 
with a large positive $T$, and measure the Hall current at the final time $t=0$. 

The time-dependent Schr\"odinger equation is given by  
\[
i\frac{d}{dt}\Psi^{(N)}(t)=H^{(\Lambda)}(t)\Psi^{(N)}(t)
\]
for the wavefunction $\Psi^{(N)}(t)$ for the $N$ fermions. 
We denote the time evolution operator for the unperturbed Hamiltonian $H_0^{(\Lambda)}$ by 
\begin{equation}
U_0^{(\Lambda)}(t,s):=\exp\bigl[-i(t-s)H_0^{(\Lambda)}\bigr]\quad
\text{for \ } 
t,s\in\re.
\label{U0}
\end{equation}
We choose the initial vector $\Psi^{(N)}(-T)$ at $t=-T$ as 
\[
\Psi^{(N)}(-T)=U_0^{(\Lambda)}(-T,0)\Phi^{(N)}
\] 
with a vector $\Phi^{(N)}$. Then, the final vector $\Psi^{(N)}=\Psi^{(N)}(t=0)$ 
is obtained as 
\begin{equation}
\Psi^{(N)}=\Phi^{(N)}-i\lambda\int_{-T}^0ds\; U_0^{(\Lambda)}(0,s)
W(s)U_0^{(\Lambda)}(s,0)\Phi^{(N)}+o(\lambda)
\label{PsiN}
\end{equation}
by using a perturbation theory \cite{Koma1}, where $o(\lambda)$ denotes a vector $\Psi^{(N)}_R$ 
with the norm $\Vert\Psi_R^{(N)}\Vert$ satisfying $\Vert\Psi_R^{(N)}\Vert/\lambda\rightarrow 0$ 
as $\lambda\rightarrow 0$.  

We denote the $q$ ground-state vectors of the unperturbed Hamiltonian $H_0^{(\Lambda,N)}$ 
by $\Phi_{0,m}^{(N)}$ with the energy eigenvalue $E_{0,m}^{(N)}$, $\ m=1,2,\ldots,q$. 
We choose the initial vector $\Phi^{(N)}=\Phi_{0,m}^{(N)}$ with the norm $1$, and 
write the corresponding final vector at $t=0$ as $\Psi^{(N)}=\Psi_{0,m}^{(N)}$.  
Then, the ground-state expectation value of the approximate local current operator 
$J_\mathcal{N}^{(1)}(k,\ell)$ is given by 
\begin{equation}
\bigl\langle J_\mathcal{N}^{(1)}(k,\ell)\bigr\rangle:=
\frac{1}{q}\sum_{m=1}^q \bigl\langle\Psi_{0,m}^{(N)},J_\mathcal{N}^{(1)}(k,\ell)
\Psi_{0,m}^{(N)}\bigr\rangle.
\label{expApproJ1}
\end{equation}

Let $P_0^{(\Lambda,N)}$ be the projection onto the sector of the ground state 
of $H_0^{(\Lambda,N)}$, and we denote the ground-state expectation by 
\[
\omega_0^{(\Lambda,N)}(\cdots):=\frac{1}{q}\;{\rm Tr}\ (\cdots)P_0^{(\Lambda,N)},  
\]
where $q$ is the degeneracy of the ground state.
Using the linear perturbation (\ref{PsiN}), the expectation value of (\ref{expApproJ1}) is decomposed into 
three parts as  
\[
\bigl\langle J_\mathcal{N}^{(1)}(k,\ell)\bigr\rangle=
\bigl\langle J_\mathcal{N}^{(1)}(k,\ell)\bigr\rangle_0+
\bigl\langle J_\mathcal{N}^{(1)}(k,\ell)\bigr\rangle_1+o(\lambda),
\]
where 
\[
\bigl\langle J_\mathcal{N}^{(1)}(k,\ell)\bigr\rangle_0
=\omega_0^{(\Lambda,N)}\bigl(J_\mathcal{N}^{(1)}(k,\ell)\bigr),
\]
and
\begin{multline}
\bigl\langle J_\mathcal{N}^{(1)}(k,\ell)\bigr\rangle_1\\
=-i\lambda\int_{-T}^0 ds\; \omega_0^{(\Lambda,N)}\bigl(J_\mathcal{N}^{(1)}(k,\ell)
U_0^{(\Lambda)}(0,s)W(s)U_0^{(\Lambda)}(s,0)\bigr)+{\rm c}.{\rm c}.
\label{expApproJ11} 
\end{multline}
The first term $\bigl\langle J_\mathcal{N}^{(1)}(k,\ell)\bigr\rangle_0$ 
is the persistent current which is usually vanishing. We are interested 
in the second term $\bigl\langle J_\mathcal{N}^{(1)}(k,\ell)\bigr\rangle_1$ 
which gives the linear response coefficient, i.e., the Hall conductance.  
We write 
\begin{equation}
\chi^{(\Lambda)}(\Gamma_\mathcal{M}(k,\ell);s)
:=U_0^{(\Lambda)}(0,s)\chi(\Gamma_\mathcal{M}(k,\ell))U_0^{(\Lambda)}(s,0).
\label{defchiGammas}
\end{equation}
Using this and the definition (\ref{Wt}) of $W(t)$, 
the contribution (\ref{expApproJ11}) of the expectation value of the local current is written 
\begin{equation}
\bigl\langle J_\mathcal{N}^{(1)}(k,\ell)\bigr\rangle_1
=i\lambda\int_{-T}^0ds \; e^{\eta s}\omega_0^{(\Lambda,N)}
\bigl([\chi^{(\Lambda)}(\Gamma_\mathcal{M}(k,\ell);s),J_\mathcal{N}^{(1)}(k,\ell)]\bigr).
\label{expApproJ11C}
\end{equation}
Note that, by using integral by parts, we have 
\begin{multline}
\label{corlinearres}
\int_{-T}^0ds \; e^{\eta s}\omega_0^{(\Lambda,N)}
\bigl([\chi^{(\Lambda)}(\Gamma_\mathcal{M}(k,\ell);s),J_\mathcal{N}^{(1)}(k,\ell)]\bigr)\\ 
=\left[se^{\eta s}\omega_0^{(\Lambda,N)}
\bigl([\chi^{(\Lambda)}(\Gamma_\mathcal{M}(k,\ell);s),J_\mathcal{N}^{(1)}(k,\ell)]\bigr)\right]_{-T}^0\\ 
-\int_{-T}^0ds \; \eta s e^{\eta s}\omega_0^{(\Lambda,N)}
\bigl([\chi^{(\Lambda)}(\Gamma_\mathcal{M}(k,\ell);s),J_\mathcal{N}^{(1)}(k,\ell)]\bigr)\\ 
-\int_{-T}^0ds \; s e^{\eta s}\omega_0^{(\Lambda,N)}
\left(\biggl[\frac{d}{ds}\chi^{(\Lambda)}(\Gamma_\mathcal{M}(k,\ell);s),J_\mathcal{N}^{(1)}(k,\ell)\biggr]
\right).
\end{multline}
Clearly, the first term in the right-hand side is vanishing as $T\rightarrow \infty$. 
The second term is also vanishing as $\eta\rightarrow 0$ after taking the limit $T\rightarrow\infty$. 
We show this in Appendix~\ref{CorrHallcon}. 
The third term in the right-hand side is written 
\[
\int_{-T}^0ds \; s e^{\eta s}\omega_0^{(\Lambda,N)}
\bigl(\bigl[J_\mathcal{N}^{(1)}(k,\ell),J^{(\Lambda)}(\Gamma_\mathcal{M}(k,\ell);s)\bigr]\bigr), 
\]
where 
\begin{multline*}
J^{(\Lambda)}(\Gamma_\mathcal{M}(k,\ell);s)\\
:=\frac{d}{ds}\chi^{(\Lambda)}(\Gamma_\mathcal{M}(k,\ell);s)=U_0^{(\Lambda)}(0,s)
i[H_0^{(\Lambda)},\chi(\Gamma_\mathcal{M}(k,\ell))]U_0^{(\Lambda)}(s,0).
\end{multline*}
For getting the second equality in the right-hand side, we have used the definition (\ref{U0}) of the time 
evolution operator $U_0^{(\Lambda)}(t,s)$ and the definition (\ref{defchiGammas}) of the operator 
$\chi^{(\Lambda)}(\Gamma_\mathcal{M}(k,\ell);s)$. 
Clearly, the above operator $J^{(\Lambda)}(\Gamma_\mathcal{M}(k,\ell);s)$ 
is interpreted as the current across the boundary 
of the region $\Gamma_\mathcal{M}(k,\ell)$ at the time $s$. From these observations, 
we define the Hall conductance for the finite lattice $\Lambda$ as 
\begin{multline}
\label{tildesigma12}
{\tilde\sigma}_{12}^{(\Lambda,N)}(\eta,T,\mathcal{N},\mathcal{M})\\
:=i\int_{-T}^0ds \; s e^{\eta s}\omega_0^{(\Lambda,N)}
\bigl(\bigl[J_\mathcal{N}^{(1)}(k,\ell),J^{(\Lambda)}(\Gamma_\mathcal{M}(k,\ell);s)\bigr]\bigr)
\end{multline}
because $\lambda$ is the voltage difference. 
The Hall conductance in the infinite-volume limit is given by 
\begin{equation}
\sigma_{12}:=\lim_{\mathcal{N}\rightarrow\infty}\lim_{\mathcal{M}\rightarrow\infty}
\lim_{\eta\rightarrow 0}\lim_{T\rightarrow\infty}\lim_{\Lambda\nearrow\ze^2}
{\tilde\sigma}_{12}^{(\Lambda,N)}(\eta,T,\mathcal{N},\mathcal{M}).
\label{sigma12}
\end{equation}
\smallskip

\section{Twisting the Phases}
\label{sec:TP}

Consider 
\[
H_0^{(\ze^2)}(\phi_j,k):=\exp[-i\phi_j\theta^{(j)}(k)]H_0^{(\ze^2)}
\exp[i\phi_j\theta^{(j)}(k)],\quad \text{for \ } \phi_j\in\re, 
\]
where $\theta^{(j)}(k)$ is the step function of (\ref{stepfunc}). 
This transformation changes the hopping amplitudes as 
\[
t_{xy}\rightarrow t_{xy}e^{i\phi_j}\quad \text{for \ } x^{(j)}<k,\; y^{(j)}\ge k, 
\]
and the opposite hopping amplitude $t_{yx}$ is determined 
by the Hermitian condition $t_{yx}=t_{xy}^\ast$. 
The rest of the hopping amplitudes do not change. 
In both of the first and the second directions, we define the twist by 
\begin{multline*}
H_0^{(\ze^2)}(\phi_1,k_1;\phi_2,k_2)\\
:=\exp\bigl[-i\bigl(\phi_1\theta^{(1)}(k_1)+\phi_2\theta^{(2)}(k_2)\bigr)\bigr]
H_0^{(\ze^2)}
\exp\bigl[i\bigl(\phi_1\theta^{(1)}(k_1)+\phi_2\theta^{(2)}(k_2)\bigr)\bigr].
\end{multline*}
This transformation is justified for localized wavefunctions. 

Relying on this rule, we can twist the phases of the hopping amplitudes 
of the finite-volume Hamiltonian $H_0^{(\Lambda)}$.
We denote by $H_0^{(\Lambda)}(\phi_1,k_1;\phi_2,k_2)$ 
the corresponding finite-volume Hamiltonian on $\Lambda$. 
When both of $k_1$ and $k_2$ are placed at the boundaries of the lattice $\Lambda$, 
the boundary conditions of the system become the usual twisted boundary conditions.  
We write $H_0^{(\Lambda)}(\phi_1;\phi_2)$ for the Hamiltonian 
with the twisted boundary conditions for short. 
Further, one has 
\[
H_0^{(\ze^2)}(\phi_j,k-1)=\exp\biggl[-i\phi_j\sum_{\substack{x\in\ze^2\\ x^{(j)}=k-1}}n_x\biggr]
H_0^{(\ze^2)}(\phi_j,k)
\exp\biggl[i\phi_j\sum_{\substack{x\in\ze^2\\ x^{(j)}=k-1}}n_x\biggr].
\]
Thus, one can change the position of the twisted hopping amplitudes 
by using the unitary transformation which is local in the $j$-th direction.  

Similarly, we can consider 
\begin{multline*}
J^{(j)}(k;\phi_1,k_1;\phi_2,k_2)\\
:=\exp\bigl[-i\bigl(\phi_1\theta^{(1)}(k_1)+\phi_2\theta^{(2)}(k_2)\bigr)\bigr]J^{(j)}(k)
\exp\bigl[i\bigl(\phi_1\theta^{(1)}(k_1)+\phi_2\theta^{(2)}(k_2)\bigr)\bigr]
\end{multline*}
for the local current operator $J^{(j)}(k)$. Therefore, 
in the same way as in the case of the Hamiltonian, 
one can change the position of the twisted phases as   
\begin{multline*}
J^{(j)}(k;\phi_1,k_1-1;\phi_2,k_2)\\
=\exp\biggl[-i\phi_1\sum_{\substack{x\in \ze^2\\ x^{(1)}=k_1-1}}n_x\biggr]
J^{(j)}(k;\phi_1,k_1;\phi_2,k_2)\exp\biggl[i\phi_1\sum_{\substack{x\in\ze^2 \\ x^{(1)}=k_1-1}}n_x\biggr].
\end{multline*}
Further, one has 
\begin{align*}
\frac{\partial}{\partial\phi_j}H_0^{(\ze^2)}(\phi_j,k)&=
\exp[-i\phi_j\theta^{(j)}(k)]i[H_0^{(\ze^2)},\theta^{(j)}(k)]
\exp[i\phi_j\theta^{(j)}(k)]\\
&=\exp[-i\phi_j\theta^{(j)}(k)]J^{(j)}(k)\exp[i\phi_j\theta^{(j)}(k)].
\end{align*}
Namely, we can obtain the local current by differentiating the Hamiltonian. 
Thus, one has
\begin{equation}
\label{H0J}
\frac{\partial}{\partial \phi_j}H_0^{(\Lambda)}(\phi_1,k_1;\phi_2,k_2)
=J^{(j)}(k_j;\phi_1,k_1;\phi_2,k_2),\quad j=1,2,  
\end{equation}
for the finite-volume Hamiltonian $H_0^{(\Lambda)}$. 

\begin{thm}
\label{thm:twistDegeneGap}
Under the same assumptions as in Theorem~\ref{thm:main}, we have that 
the ground state of the Hamiltonian $H_0^{(\Lambda)}(\phi_1,k_1;\phi_2,k_2)$ 
having the twisted phases $\phi_1$ and $\phi_2$ at the positions $k_1$ and $k_2$  
in the first and the second directions, respectively, 
has the same q-fold (quasi)degeneracy as that of the Hamiltonian $H_0^{(\Lambda)}$ 
without the twisted phases. Further, there exists a uniform spectral gap 
above the sector of the ground state for any $\phi_1$ and $\phi_2$, 
and there exists a positive lower bound of the spectral gap 
such that the bound is independent of the phases $\phi_1, \phi_2$. 
\end{thm}

The proof is given in Appendix~\ref{ProoftwistDegeneGap}. 

\begin{rem}
As mentioned in Sec.~\ref{sec:intro}, the method of twisting the phases at the boundaries was first 
introduced by \cite{NTW} for the quantum Hall systems. 
Later, the method was extended to strongly correlated quantum systems such as 
quantum spin systems in \cite{Hatsugai,HKH1,HKH}, in order to measure the local 
topological properties such as a local singlet pair of two spins. 
\end{rem}

\section{A Topological Invariant}
\label{sec:TopoInv}

Consider the Hamiltonian $H_0^{(\Lambda)}(\phi_1,k;\phi_2,\ell)$. 
Relying on Theorem~\ref{thm:twistDegeneGap}, we denote the $q$ vectors of 
the ground state by $\Phi_{0,m}^{(N)}(\phi_1,k;\phi_2,\ell)$ 
with the energy eigenvalue $E_{0,m}^{(N)}(\phi_1;\phi_2)$, 
$m=1,2,\ldots,q$, and denote the excited-state vectors $\Phi_n^{(N)}(\phi_1,k;\phi_2,\ell)$ with 
the energy eigenvalue $E_n^{(N)}(\phi_1;\phi_2)$, $n\ge 1$.
We write $\boldsymbol{\phi}=(\phi_1,\phi_2)$, 
and $\tilde{\boldsymbol{\phi}}=(\phi_1,k;\phi_2,\ell)$ for short. 
We define 
\begin{multline}
\label{hatsigma12tildephi}
{\hat\sigma}_{12}^{(\Lambda,N)}(\tilde{\boldsymbol{\phi}})\\
:=\frac{i}{q}\sum_{m=1}^q 
\sum_{n\ge 1}\left[
\frac{
\langle \Phi_{0,m}^{(N)}(\tilde{\boldsymbol{\phi}}), J^{(1)}(k;\tilde{\boldsymbol{\phi}})
\Phi_n^{(N)}(\tilde{\boldsymbol{\phi}})\rangle
\langle\Phi_n^{(N)}(\tilde{\boldsymbol{\phi}}),J^{(2)}(\ell;\tilde{\boldsymbol{\phi}})
\Phi_{0,m}^{(N)}(\tilde{\boldsymbol{\phi}})\rangle}
{(E_{0,m}^{(N)}(\boldsymbol{\phi})-E_n^{(N)}(\boldsymbol{\phi}))^2}\right.\\
-\left.\frac{
\langle \Phi_{0,m}^{(N)}(\tilde{\boldsymbol{\phi}}), J^{(2)}(\ell;\tilde{\boldsymbol{\phi}})
\Phi_n^{(N)}(\tilde{\boldsymbol{\phi}})\rangle
\langle\Phi_n^{(N)}(\tilde{\boldsymbol{\phi}}),J^{(1)}(k;\tilde{\boldsymbol{\phi}})
\Phi_{0,m}^{(N)}(\tilde{\boldsymbol{\phi}})\rangle}
{(E_{0,m}^{(N)}(\boldsymbol{\phi})-E_n^{(N)}(\boldsymbol{\phi}))^2}\right].
\end{multline}
We write $\Phi_{0,m}^{(N)}(\boldsymbol{\phi})$ for the vectors of the ground state of 
the Hamiltonian $H_0^{(\Lambda)}(\phi_1;\phi_2)$ with the twisted boundary conditions 
with the angles $\phi_1$ and $\phi_2$, and $\Phi_n^{(N)}(\boldsymbol{\phi})$ for 
the vectors of the excited states. We define  
\begin{multline}
\label{hatsigma12phi}
{\hat\sigma}_{12}^{(\Lambda,N)}(\boldsymbol{\phi})\\
:=\frac{i}{q}\sum_{m=1}^q 
\sum_{n\ge 1}\left[
\frac{
\langle \Phi_{0,m}^{(N)}(\boldsymbol{\phi}), J^{(1)}(k)
\Phi_n^{(N)}(\boldsymbol{\phi})\rangle
\langle\Phi_n^{(N)}(\boldsymbol{\phi}),J^{(2)}(\ell)
\Phi_{0,m}^{(N)}(\boldsymbol{\phi})\rangle}
{(E_{0,m}^{(N)}(\boldsymbol{\phi})-E_n^{(N)}(\boldsymbol{\phi}))^2}\right.\\
-\left.\frac{
\langle \Phi_{0,m}^{(N)}(\boldsymbol{\phi}), J^{(2)}(\ell)
\Phi_n^{(N)}(\boldsymbol{\phi})\rangle
\langle\Phi_n^{(N)}(\boldsymbol{\phi}),J^{(1)}(k)
\Phi_{0,m}^{(N)}(\boldsymbol{\phi})\rangle}
{(E_{0,m}^{(N)}(\boldsymbol{\phi})-E_n^{(N)}(\boldsymbol{\phi}))^2}\right].
\end{multline}
This is the standard form of the Hall conductance for the degenerate ground state 
with the twisted boundary conditions. Since we can change the position of 
the twisted hopping amplitudes by using the unitary transformation as shown in the preceding section, 
this conductance ${\hat\sigma}_{12}^{(\Lambda,N)}(\boldsymbol{\phi})$ is equal to 
the above ${\hat\sigma}_{12}^{(\Lambda,N)}(\tilde{\boldsymbol{\phi}})$ 
of (\ref{hatsigma12tildephi}).  

Niu, Thouless and Wu \cite{NTW} showed the following: 
When averaging the conductance ${\hat\sigma}_{12}^{(\Lambda,N)}(\boldsymbol{\phi})$ over 
the phases $\phi_1$ and $\phi_2$, the averaged Hall conductance is fractionally quantized 
as in (\ref{FAQHC}) below. Along the lines of \cite{Koma1}, we shall give the proof. 

Using a contour integral, the projection $P_0^{(\Lambda,N)}(\tilde{\boldsymbol{\phi}})$ 
onto the sector of the ground state is written 
\[
P_0^{(\Lambda,N)}(\tilde{\boldsymbol{\phi}})=\frac{1}{2\pi i}\oint dz 
\frac{1}{z-H_0^{(\Lambda,N)}(\tilde{\boldsymbol{\phi}})}
\]
because the spectral gap exists above the sector of the ground state 
as we proved in Theorem~\ref{thm:twistDegeneGap}. 
Note that
\begin{align*}
P_{0,1}^{(\Lambda,N)}(\tilde{\boldsymbol{\phi}})
&:=\frac{\partial}{\partial \phi_1}P_0^{(\Lambda,N)}(\tilde{\boldsymbol{\phi}})\\
&=\frac{1}{2\pi i}\oint dz 
\frac{1}{z-H_0^{(\Lambda,N)}(\tilde{\boldsymbol{\phi}})}
\left[\frac{\partial H_0^{(\Lambda,N)}(\tilde{\boldsymbol{\phi}})}{\partial\phi_1}\right]
\frac{1}{z-H_0^{(\Lambda,N)}(\tilde{\boldsymbol{\phi}})}\\
&=\frac{1}{2\pi i}\oint dz 
\frac{1}{z-H_0^{(\Lambda,N)}(\tilde{\boldsymbol{\phi}})}
J^{(1)}(k;\tilde{\boldsymbol{\phi}})
\frac{1}{z-H_0^{(\Lambda,N)}(\tilde{\boldsymbol{\phi}})},
\end{align*}
where we have used (\ref{H0J}). Similarly, one has 
\begin{align*}
P_{0,2}^{(\Lambda,N)}(\tilde{\boldsymbol{\phi}})
&:=\frac{\partial}{\partial \phi_2}P_0^{(\Lambda,N)}(\tilde{\boldsymbol{\phi}})\\
&=\frac{1}{2\pi i}\oint dz 
\frac{1}{z-H_0^{(\Lambda,N)}(\tilde{\boldsymbol{\phi}})}
J^{(2)}(\ell;\tilde{\boldsymbol{\phi}})
\frac{1}{z-H_0^{(\Lambda,N)}(\tilde{\boldsymbol{\phi}})}.
\end{align*}
Using these, we have 
\begin{equation}
\label{hatsigma12tildephiTrP}
{\hat\sigma}_{12}^{(\Lambda,N)}(\tilde{\boldsymbol{\phi}})
=\frac{i}{q}\;{\rm Tr}\;P_0^{(\Lambda,N)}(\tilde{\boldsymbol{\phi}})
\left[P_{0,1}^{(\Lambda,N)}(\tilde{\boldsymbol{\phi}}),
P_{0,2}^{(\Lambda,N)}(\tilde{\boldsymbol{\phi}})\right].
\end{equation}

\begin{rem}
In the infinite-volume limit, $\Lambda\nearrow\ze^2$ and $N\nearrow\infty$, one formally has 
$$
\hat{\sigma}_{12}^{(\ze^2,\infty)}
=-\frac{i}{q}\;{\rm Tr}\;P_0^{(\ze^2,\infty)}
\left[\big[P_0^{(\ze^2,\infty)},\theta^{(1)}(k_1)\big],
\big[P_0^{(\ze^2,\infty)},\theta^{(2)}(k_2)\big]\right],
$$
where we have used 
$$
P_{0,j}^{(\ze^2,\infty)}=i\big[P_0^{(\ze^2,\infty)},\theta^{(j)}(k_j)\big]\quad \mbox{for } j=1,2,
$$
which are derived from (\ref{localcurrent}).  
This expression of the Hall conductance $\hat{\sigma}_{12}^{(\ze^2,\infty)}$ 
has the same form as that for formally applying the method of noncommutative geometry \cite{Connes,BVS,AG} to 
an interacting fermion system. Actually, the expression can be obtained 
by replacing the single-fermion step function\footnote{See, e.g., \cite{ASS,ElSch,Koma2}.}  
with that for many fermions.    
But, justifying the expression in a mathematical rigorous manner remains a challenge 
because the many-body step function $\theta^{(j)}(k_j)$ is expressed 
in terms of the infinite sum of the number operators. 
\end{rem}

The following proposition is essentially due to Kato~\cite{Kato}. 
See also Proposition~5.1 in \cite{Koma1}.

\begin{prop}
\label{prop:Kato}
There exist orthonormal vectors $\hat{\Phi}_{0,m}^{(N)}(\tilde{\boldsymbol{\phi}})$, $m=1,2,\ldots,q$ 
such that the sector of the ground state of $H_0^{(\Lambda,N)}(\tilde{\boldsymbol{\phi}})$ is 
spanned by the $q$ vectors $\hat{\Phi}_{0,m}^{(N)}(\tilde{\boldsymbol{\phi}})$, and that 
all the vectors $\hat{\Phi}_{0,m}^{(N)}(\tilde{\boldsymbol{\phi}})$ are infinitely 
differentiable with respect to the phase parameters $\boldsymbol{\phi}\in[0,2\pi]\times[0,2\pi]$. 
\end{prop}

Relying on this proposition, the right-hand side of (\ref{hatsigma12tildephiTrP}) is written 
\begin{multline*}
{\rm Tr}\;P_0^{(\Lambda,N)}(\tilde{\boldsymbol{\phi}})
\left[P_{0,1}^{(\Lambda,N)}(\tilde{\boldsymbol{\phi}}),
P_{0,2}^{(\Lambda,N)}(\tilde{\boldsymbol{\phi}})\right]\\
=\sum_{m=1}^q \left[\big\langle\hat{\Phi}_{0,m}^{(N)}(\tilde{\boldsymbol{\phi}}), 
P_{0,1}^{(\Lambda,N)}(\tilde{\boldsymbol{\phi}})
P_{0,2}^{(\Lambda,N)}(\tilde{\boldsymbol{\phi}})\hat{\Phi}_{0,m}^{(N)}(\tilde{\boldsymbol{\phi}})\big\rangle\right.\\
-\left.\big\langle\hat{\Phi}_{0,m}^{(N)}(\tilde{\boldsymbol{\phi}}),
P_{0,2}^{(\Lambda,N)}(\tilde{\boldsymbol{\phi}})P_{0,1}^{(\Lambda,N)}(\tilde{\boldsymbol{\phi}})
\hat{\Phi}_{0,m}^{(N)}(\tilde{\boldsymbol{\phi}})\big\rangle
\right]\\
=\sum_{m=1}^q \left[\big\langle P_{0,1}^{(\Lambda,N)}(\tilde{\boldsymbol{\phi}})
\hat{\Phi}_{0,m}^{(N)}(\tilde{\boldsymbol{\phi}}), 
P_{0,2}^{(\Lambda,N)}(\tilde{\boldsymbol{\phi}})\hat{\Phi}_{0,m}^{(N)}(\tilde{\boldsymbol{\phi}})\big\rangle\right.\\
-\left.\big\langle P_{0,2}^{(\Lambda,N)}(\tilde{\boldsymbol{\phi}})\hat{\Phi}_{0,m}^{(N)}(\tilde{\boldsymbol{\phi}}),
P_{0,1}^{(\Lambda,N)}(\tilde{\boldsymbol{\phi}})
\hat{\Phi}_{0,m}^{(N)}(\tilde{\boldsymbol{\phi}})\big\rangle
\right]
\end{multline*}
except for the factor $i/q$. Note that 
\begin{align*}
\frac{\partial}{\partial \phi_j}\hat{\Phi}_{0,m}^{(N)}(\tilde{\boldsymbol{\phi}})
&=\frac{\partial}{\partial \phi_j}P_0^{(\Lambda,N)}(\tilde{\boldsymbol{\phi}})
\hat{\Phi}_{0,m}^{(N)}(\tilde{\boldsymbol{\phi}})\\
&=P_{0,j}^{(\Lambda,N)}(\tilde{\boldsymbol{\phi}})
\hat{\Phi}_{0,m}^{(N)}(\tilde{\boldsymbol{\phi}})+P_0^{(\Lambda,N)}(\tilde{\boldsymbol{\phi}})
\frac{\partial}{\partial \phi_j}\hat{\Phi}_{0,m}^{(N)}(\tilde{\boldsymbol{\phi}})
\end{align*}
for $j=1,2$. Therefore, one has 
\begin{equation}
\label{diffPhiid}
\left[1-P_0^{(\Lambda,N)}(\tilde{\boldsymbol{\phi}})\right]
\frac{\partial}{\partial \phi_j}\hat{\Phi}_{0,m}^{(N)}(\tilde{\boldsymbol{\phi}})=
P_{0,j}^{(\Lambda,N)}(\tilde{\boldsymbol{\phi}})
\hat{\Phi}_{0,m}^{(N)}(\tilde{\boldsymbol{\phi}})\quad \mbox{for\ } j=1,2. 
\end{equation}
Using these identities, the above trace is written 
\begin{align*}
&{\rm Tr}\;P_0^{(\Lambda,N)}(\tilde{\boldsymbol{\phi}})
\left[P_{0,1}^{(\Lambda,N)}(\tilde{\boldsymbol{\phi}}),
P_{0,2}^{(\Lambda,N)}(\tilde{\boldsymbol{\phi}})\right]\\
&=\sum_{m=1}^q \left\{\big\langle\big[1-P_0^{(\Lambda,N)}(\tilde{\boldsymbol{\phi}})\big]
\frac{\partial}{\partial \phi_1}\hat{\Phi}_{0,m}^{(N)}(\tilde{\boldsymbol{\phi}}),
\big[1-P_0^{(\Lambda,N)}(\tilde{\boldsymbol{\phi}})\big]
\frac{\partial}{\partial \phi_2}\hat{\Phi}_{0,m}^{(N)}(\tilde{\boldsymbol{\phi}})\big\rangle\right.\\
&-\left.\big\langle\big[1-P_0^{(\Lambda,N)}(\tilde{\boldsymbol{\phi}})\big]
\frac{\partial}{\partial \phi_2}\hat{\Phi}_{0,m}^{(N)}(\tilde{\boldsymbol{\phi}}),
\big[1-P_0^{(\Lambda,N)}(\tilde{\boldsymbol{\phi}})\big]
\frac{\partial}{\partial \phi_1}\hat{\Phi}_{0,m}^{(N)}(\tilde{\boldsymbol{\phi}})\big\rangle\right\}\\
&=\sum_{m=1}^q \left\{\biggl\langle
\frac{\partial}{\partial \phi_1}\hat{\Phi}_{0,m}^{(N)}(\tilde{\boldsymbol{\phi}}),
\big[1-P_0^{(\Lambda,N)}(\tilde{\boldsymbol{\phi}})\big]
\frac{\partial}{\partial \phi_2}\hat{\Phi}_{0,m}^{(N)}(\tilde{\boldsymbol{\phi}})\biggr\rangle\right.\\
&-\left.\biggl\langle
\frac{\partial}{\partial \phi_2}\hat{\Phi}_{0,m}^{(N)}(\tilde{\boldsymbol{\phi}}),
\big[1-P_0^{(\Lambda,N)}(\tilde{\boldsymbol{\phi}})\big]
\frac{\partial}{\partial \phi_1}\hat{\Phi}_{0,m}^{(N)}(\tilde{\boldsymbol{\phi}})\biggr\rangle\right\}\\
&=\sum_{m=1}^q \left\{\biggl\langle
\frac{\partial}{\partial \phi_1}\hat{\Phi}_{0,m}^{(N)}(\tilde{\boldsymbol{\phi}}),
\frac{\partial}{\partial \phi_2}\hat{\Phi}_{0,m}^{(N)}(\tilde{\boldsymbol{\phi}})\biggr\rangle\right.
-\left.\biggl\langle
\frac{\partial}{\partial \phi_2}\hat{\Phi}_{0,m}^{(N)}(\tilde{\boldsymbol{\phi}}),
\frac{\partial}{\partial \phi_1}\hat{\Phi}_{0,m}^{(N)}(\tilde{\boldsymbol{\phi}})\biggr\rangle\right\},
\end{align*}
where we have used the identity, 
\begin{align*}
0&=\frac{\partial}{\partial \phi_j}\langle\hat{\Phi}_{0,m}^{(N)}(\tilde{\boldsymbol{\phi}}),
\hat{\Phi}_{0,m'}^{(N)}(\tilde{\boldsymbol{\phi}})\rangle\\
&=
\Big\langle\frac{\partial}{\partial \phi_j}\hat{\Phi}_{0,m}^{(N)}(\tilde{\boldsymbol{\phi}}),
\hat{\Phi}_{0,m'}^{(N)}(\tilde{\boldsymbol{\phi}})\Big\rangle+
\Big\langle\hat{\Phi}_{0,m}^{(N)}(\tilde{\boldsymbol{\phi}}),
\frac{\partial}{\partial \phi_j}\hat{\Phi}_{0,m'}^{(N)}(\tilde{\boldsymbol{\phi}})\Big\rangle,
\end{align*}
for $j=1,2$, and $m,m'=1,2,\ldots,q$, in order to get the third equality. 
Consequently, the Hall conductance ${\hat\sigma}_{12}^{(\Lambda,N)}(\tilde{\boldsymbol{\phi}})$ 
of (\ref{hatsigma12tildephiTrP}) is written \cite{Koma1}
\begin{multline}
\label{expHallconduTopo}
{\hat\sigma}_{12}^{(\Lambda,N)}(\tilde{\boldsymbol{\phi}})\\
=\frac{i}{q}\sum_{m=1}^q
\left[\frac{\partial}{\partial\phi_1}\bigl\langle\hat{\Phi}_{0,m}^{(N)}(\tilde{\boldsymbol{\phi}}),
\frac{\partial}{\partial\phi_2}\hat{\Phi}_{0,m}^{(N)}(\tilde{\boldsymbol{\phi}})\bigr\rangle
-\frac{\partial}{\partial\phi_2}\bigl\langle\hat{\Phi}_{0,m}^{(N)}(\tilde{\boldsymbol{\phi}}),
\frac{\partial}{\partial\phi_1}\hat{\Phi}_{0,m}^{(N)}(\tilde{\boldsymbol{\phi}})\bigr\rangle\right].
\end{multline}

Next, following \cite{Koma1}, we show that the Hall conductance averaged over the phase $\boldsymbol{\phi}$ 
is fractionally quantized. We define the averaged Hall conductance as  
\begin{equation}
\label{avHallcond}
\overline{{\hat\sigma}_{12}^{(\Lambda,N)}(\tilde{\boldsymbol{\phi}})}:=
\frac{1}{(2\pi)^2}\iint_{[0,2\pi]\times[0,2\pi]}d\phi_1d\phi_2\; 
{\hat\sigma}_{12}^{(\Lambda,N)}(\tilde{\boldsymbol{\phi}}).
\end{equation}
Using the expression (\ref{expHallconduTopo}) of the Hall conductance, the integral in the right-hand side 
is written 
\begin{equation}
\label{intHallcond}
\iint_{[0,2\pi]\times[0,2\pi]}d\phi_1d\phi_2\; 
{\hat\sigma}_{12}^{(\Lambda,N)}(\tilde{\boldsymbol{\phi}})=\frac{i}{q}\left[I^{(1)}-I^{(2)}\right],
\end{equation}
where 
\begin{multline*}
I^{(1)}=\sum_{m=1}^q \int_0^{2\pi}d\phi_2\left[\Big\langle\hat{\Phi}_{0,m}^{(N)}(2\pi,\phi_2), 
\frac{\partial}{\partial\phi_2}\hat{\Phi}_{0,m}^{(N)}(2\pi,\phi_2)\Big\rangle\right.\\
-\left.\Big\langle\hat{\Phi}_{0,m}^{(N)}(0,\phi_2), 
\frac{\partial}{\partial\phi_2}\hat{\Phi}_{0,m}^{(N)}(0,\phi_2)\Big\rangle\right]
\end{multline*}
and
\begin{multline*}
I^{(2)}=\sum_{m=1}^q \int_0^{2\pi}d\phi_1\left[\Big\langle\hat{\Phi}_{0,m}^{(N)}(\phi_1,2\pi), 
\frac{\partial}{\partial\phi_1}\hat{\Phi}_{0,m}^{(N)}(\phi_1,2\pi)\Big\rangle\right.\\
-\left.\Big\langle\hat{\Phi}_{0,m}^{(N)}(\phi_1,0), 
\frac{\partial}{\partial\phi_1}\hat{\Phi}_{0,m}^{(N)}(\phi_1,0)\Big\rangle\right].
\end{multline*}
Here, we have dropped the $k,\ell$-dependence of the vectors $\hat{\Phi}_{0,m}^{(N)}(\phi_1,k;\phi_2,\ell)$ 
for short. Clearly, the set of the ground-state vectors, $\hat{\Phi}_{0,m}^{(N)}(\phi_1,2\pi)$, 
$m=1,2,\ldots,q$, connects with that of $\hat{\Phi}_{0,m}^{(N)}(\phi_1,0)$, 
$m=1,2,\ldots,q$, through a $q\times q$ unitary matrix $C^{(2)}(\phi_1)$ as   
\begin{equation}
\label{PhiC(2)Phi}
\hat{\Phi}_{0,m}^{(N)}(\phi_1,2\pi)=\sum_{m'=1}^q C_{m,m'}^{(2)}(\phi_1)
\hat{\Phi}_{0,m'}^{(N)}(\phi_1,0)\quad \mbox{for\ } m=1,2,\ldots,q, 
\end{equation}
where $C_{m,m'}^{(2)}(\phi_1)$ are the matrix elements of the unitary matrix $C^{(2)}(\phi_1)$. 
By differentiating this, one has 
\begin{multline*}
\frac{\partial}{\partial\phi_1}\hat{\Phi}_{0,m}^{(N)}(\phi_1,2\pi)\\
=\sum_{m'=1}^q \Biggl[\frac{\partial C_{m,m'}^{(2)}(\phi_1)}{\partial\phi_1}
\hat{\Phi}_{0,m'}^{(N)}(\phi_1,0)
+C_{m,m'}^{(2)}(\phi_1)
\frac{\partial}{\partial\phi_1}\hat{\Phi}_{0,m'}^{(N)}(\phi_1,0)\Biggr] 
\end{multline*}
for $m=1,2,\ldots,q$. By using these, one obtains
\begin{align*}
&\sum_{m=1}^q\Big\langle\hat{\Phi}_{0,m}^{(N)}(\phi_1,2\pi), 
\frac{\partial}{\partial\phi_1}\hat{\Phi}_{0,m}^{(N)}(\phi_1,2\pi)\Big\rangle\\
&=\sum_{m=1}^q\sum_{m'=1}^q\sum_{m''=1}^q \Big[C_{m,m'}^{(2)}(\phi_1)^\ast
\frac{\partial C_{m,m''}^{(2)}(\phi_1)}{\partial\phi_1}
\big\langle\hat{\Phi}_{0,m'}^{(N)}(\phi_1,0),\hat{\Phi}_{0,m''}^{(N)}(\phi_1,0)\big\rangle\\
&+C_{m,m'}^{(2)}(\phi_1)^\ast C_{m,m''}^{(2)}(\phi_1)
\big\langle\hat{\Phi}_{0,m'}^{(N)}(\phi_1,0),\frac{\partial}{\partial\phi_1}
\hat{\Phi}_{0,m''}^{(N)}(\phi_1,0)\big\rangle\Big]\\
&=\sum_{m'=1}^q \Big[\sum_{m=1}^q C_{m,m'}^{(2)}(\phi_1)^\ast
\frac{\partial C_{m,m'}^{(2)}(\phi_1)}{\partial\phi_1}
+\big\langle\hat{\Phi}_{0,m'}^{(N)}(\phi_1,0),\frac{\partial}{\partial\phi_1}
\hat{\Phi}_{0,m'}^{(N)}(\phi_1,0)\big\rangle\Big].
\end{align*}
Substituting this into the expression of $I^{(2)}$, we have 
$$
I^{(2)}=\int_0^{2\pi}d\phi_1\; {\rm Tr}\; C^{(2)}(\phi_1)^\dagger 
\frac{\partial}{\partial\phi_1}C^{(2)}(\phi_1). 
$$
Since the unitary matrix $C^{(2)}(\phi_1)$ can be expressed 
in terms of a Hermitian matrix $\Theta^{(2)}(\phi_1)$ as 
$$
C^{(2)}(\phi_1)=\exp\big[i\Theta^{(2)}(\phi_1)\big], 
$$
we obtain\footnote{The differentiability of $\Theta^{(2)}(\phi_1)$ is justified in Appendix~\ref{DiffTheta}.} 
\begin{equation}
\label{I(2)TrTheta}
I^{(2)}=i\int_0^{2\pi}d\phi_1\; {\rm Tr}\;\frac{\partial}{\partial\phi_1}\Theta^{(2)}(\phi_1)
=i\big[{\rm Tr}\;\Theta^{(2)}(2\pi)-{\rm Tr}\;\Theta^{(2)}(0)\big].
\end{equation}

Similarly, one has 
\begin{equation}
\label{PhiC(1)Phi}
\hat{\Phi}_{0,m}^{(N)}(2\pi,\phi_2)=\sum_{m'=1}^q C_{m,m'}^{(1)}(\phi_2)
\hat{\Phi}_{0,m'}^{(N)}(0,\phi_2)\quad \mbox{for\ } m=1,2,\ldots,q, 
\end{equation}
with a unitary matrix $C^{(1)}(\phi_2)$. Therefore, in the same way, we obtain  
$$
I^{(1)}=i\big[{\rm Tr}\;\Theta^{(1)}(2\pi)-{\rm Tr}\;\Theta^{(1)}(0)\big],
$$
where we have written 
$$
C^{(1)}(\phi_2)=\exp\big[i\Theta^{(1)}(\phi_2)\big] 
$$
in terms of the Hermitian matrix $\Theta^{(1)}(\phi_1)$. 
{From} (\ref{avHallcond}) and (\ref{intHallcond}), we obtain 
$$
\overline{{\hat\sigma}_{12}^{(\Lambda,N)}(\tilde{\boldsymbol{\phi}})}
=-\frac{1}{(2\pi)^2q}\big[{\rm Tr}\;\Theta^{(2)}(2\pi)-{\rm Tr}\;\Theta^{(2)}(0)
-{\rm Tr}\;\Theta^{(1)}(2\pi)+{\rm Tr}\;\Theta^{(1)}(0)\big].
$$
Therefore, it is sufficient to show 
\begin{equation}
\label{sumTrTheta}
{\rm Tr}\;\Theta^{(2)}(2\pi)-{\rm Tr}\;\Theta^{(2)}(0)
-{\rm Tr}\;\Theta^{(1)}(2\pi)+{\rm Tr}\;\Theta^{(1)}(0)=-2\pi p
\end{equation}
with some integer $p$. 

For this purpose, we denote by $\hat{\Phi}_0^{(N)}(\phi_1,\phi_2)$ the $q$-component vector 
with the $m$-th component, $\hat{\Phi}_{0,m}^{(N)}(\phi_1,\phi_2)$, $m=1,2,\ldots,q$. 
Then, the relation (\ref{PhiC(2)Phi}) is written  
$$
\hat{\Phi}_0^{(N)}(\phi_1,2\pi)=C^{(2)}(\phi_1)\hat{\Phi}_0^{(N)}(\phi_1,0).
$$
Similarly, the relation (\ref{PhiC(1)Phi}) is written 
$$
\hat{\Phi}_0^{(N)}(2\pi,\phi_2)=C^{(1)}(\phi_2)\hat{\Phi}_0^{(N)}(0,\phi_2).
$$
For $\phi_1=\phi_2=2\pi$, one has 
$$
\hat{\Phi}_0^{(N)}(2\pi,2\pi)=C^{(2)}(2\pi)\hat{\Phi}_0^{(N)}(2\pi,0)
=C^{(1)}(2\pi)\hat{\Phi}_0^{(N)}(0,2\pi).
$$
For $\phi_1=\phi_2=0$, 
$$
\hat{\Phi}_0^{(N)}(0,2\pi)=C^{(2)}(0)\hat{\Phi}_0^{(N)}(0,0)
$$
and 
$$
\hat{\Phi}_0^{(N)}(2\pi,0)=C^{(1)}(0)\hat{\Phi}_0^{(N)}(0,0).
$$
{From} these equations, the following relation must hold:  
$$
C^{(2)}(2\pi)C^{(1)}(0)=C^{(1)}(2\pi)C^{(2)}(0). 
$$
Taking the determinant for both sides of this equation 
and using the formula ${\rm det}\exp[i\Theta^{(j)}(\phi)]=\exp[i{\rm Tr}\;\Theta^{(j)}(\phi)]$ 
for $j=1,2$, one has (\ref{sumTrTheta}). 

Consequently, one has 
\begin{equation}
\label{FAQHC}
\overline{{\hat\sigma}_{12}^{(\Lambda,N)}(\tilde{\boldsymbol{\phi}})}:=
\frac{1}{(2\pi)^2}\iint_{[0,2\pi]\times[0,2\pi]}d\phi_1d\phi_2\; 
{\hat\sigma}_{12}^{(\Lambda,N)}(\tilde{\boldsymbol{\phi}})=\frac{1}{2\pi}\frac{p}{q}
\end{equation}
with an integer $p$. Since ${\hat\sigma}_{12}^{(\Lambda,N)}(\boldsymbol{\phi})
={\hat\sigma}_{12}^{(\Lambda,N)}(\tilde{\boldsymbol{\phi}})$ as mentioned above, we obtain 
\begin{equation}
\label{avHallformula}
\frac{1}{(2\pi)^2}\iint_{[0,2\pi]\times[0,2\pi]}d\phi_1d\phi_2\; 
{\hat\sigma}_{12}^{(\Lambda,N)}(\boldsymbol{\phi})=\frac{1}{2\pi}\frac{p}{q}
\end{equation}
for the Hall conductance ${\hat\sigma}_{12}^{(\Lambda,N)}(\boldsymbol{\phi})$ with the simple 
twisted boundary conditions.

\section{Topological Currents}
\label{sec:topocurr}

We write $H_0^{(\Lambda)}(\tilde{\boldsymbol{\phi}})$ for the Hamiltonian 
$H_0^{(\Lambda)}(\phi_1,k;\phi_2,\ell)$ in Sec.~\ref{sec:TP}. 
Consider a further deformation of the Hamiltonian $H_0^{(\Lambda)}(\tilde{\boldsymbol{\phi}})$ as 
\[
H_0^{(\Lambda)}(\tilde{\boldsymbol{\phi}},\tilde{\alpha}):=
\exp[-i\alpha\phi_1n_x]H_0^{(\Lambda)}(\tilde{\boldsymbol{\phi}})
\exp[i\alpha\phi_1n_x],
\]
where we have written $\tilde{\alpha}=(\alpha,x)$ for short.
Then, the relation (\ref{H0J}) is modified as 
\begin{multline*}
\frac{\partial}{\partial \phi_1}H_0^{(\Lambda)}(\tilde{\boldsymbol{\phi}},\tilde{\alpha})\\
=\exp[-i\alpha\phi_1n_x]\left\{J^{(1)}(k;\tilde{\boldsymbol{\phi}})
+i\alpha[H_0^{(\Lambda)}(\tilde{\boldsymbol{\phi}}),n_x]\right\} 
\exp[i\alpha\phi_1n_x].
\end{multline*}
We write $J^{(1)}(k;\tilde{\boldsymbol{\phi}},\tilde{\alpha})$ for this right-hand side. 
Namely, we have 
\begin{equation}
\frac{\partial}{\partial \phi_1}H_0^{(\Lambda)}(\tilde{\boldsymbol{\phi}},\tilde{\alpha})\\
=J^{(1)}(k;\tilde{\boldsymbol{\phi}},\tilde{\alpha}). 
\end{equation} 
Clearly, this operator $J^{(1)}(k;\tilde{\boldsymbol{\phi}},\tilde{\alpha})$ is the local current operator 
which is obtained by twisting the phases of 
\begin{equation}
\label{localcurrentalpha}
J^{(1)}(k,\tilde{\alpha}):=J^{(1)}(k)+i\alpha[H_0^{(\Lambda)},n_x]
\end{equation}
with the angles $\phi_1$ and $\phi_2$. In the infinite volume, the operator is formally written 
\begin{align*}
J^{(1)}(k,\tilde{\alpha})&=i[H_0^{(\ze^2)},\theta^{(1)}(k)]
+i\alpha[H_0^{(\ze^2)},n_x]\\ 
&=i[H_0^{(\ze^2)},\theta^{(1)}(k)+\alpha n_x].
\end{align*}
That is to say, the step potential is slightly deformed with $\alpha$ at the site $x$. 

In the same way, the projection operator onto the sector of the ground state is transformed as 
\begin{align*}
P_0^{(\Lambda,N)}(\tilde{\boldsymbol{\phi}},\tilde{\alpha})
&:=\exp[-i\alpha\phi_1n_x]P_0^{(\Lambda,N)}(\tilde{\boldsymbol{\phi}})\exp[i\alpha\phi_1n_x]\\
&=\frac{1}{2\pi i}\oint \frac{dz}{z-H_0^{(\Lambda,N)}(\tilde{\boldsymbol{\phi}},\tilde{\alpha})}. 
\end{align*}
We write
\[
P_{0,1}^{(\Lambda,N)}(\tilde{\boldsymbol{\phi}},\tilde{\alpha})
:=\frac{\partial}{\partial \phi_1}P_0^{(\Lambda,N)}(\tilde{\boldsymbol{\phi}},\tilde{\alpha})
\]
and
\[
P_{0,2}^{(\Lambda,N)}(\tilde{\boldsymbol{\phi}},\tilde{\alpha})
:=\frac{\partial}{\partial \phi_2}P_0^{(\Lambda,N)}(\tilde{\boldsymbol{\phi}},\tilde{\alpha}).
\]
Then we have:
\begin{lem}
\label{lem:TopoInvUnch}
\begin{align*}
&\iint_{[0,{2\pi}]\times[0,2\pi]}d\phi_1d\phi_2\;
{\rm Tr}\;P_0^{(\Lambda,N)}(\tilde{\boldsymbol{\phi}},\tilde{\alpha})
\left[P_{0,1}^{(\Lambda,N)}(\tilde{\boldsymbol{\phi}},\tilde{\alpha}),
P_{0,2}^{(\Lambda,N)}(\tilde{\boldsymbol{\phi}},\tilde{\alpha})\right]\\
&=\iint_{[0,{2\pi}]\times[0,2\pi]}d\phi_1d\phi_2\;
{\rm Tr}\;P_0^{(\Lambda,N)}(\tilde{\boldsymbol{\phi}})
\left[P_{0,1}^{(\Lambda,N)}(\tilde{\boldsymbol{\phi}}),
P_{0,2}^{(\Lambda,N)}(\tilde{\boldsymbol{\phi}})\right].
\end{align*}
\end{lem}
The proof is given in Appendix~\ref{ProofTopoInvUnch}. 
This lemma implies that the topological invariant, the Chern number, does not change 
for the local deformation of the phases of the Hamiltonian. 
But the local current operator in the first direction changes from (\ref{localcurrent}) 
to (\ref{localcurrentalpha}). 
Even when starting from this deformed local current operator $J^{(1)}(k,\tilde{\alpha})$, 
our argument holds in the same way. In consequence, we can prove 
the same statement as in Theorem~\ref{thm:main} for the deformed local current operator. 
Since we can repeatedly apply local deformations, we can define more generic local 
currents in the following way. 

Consider generic functions $\vartheta^{(j)}(x;k,\ell)$, $j=1,2$, which satisfy  
\[
\vartheta^{(1)}(x;k,\ell)=\theta^{(1)}(x;k), \quad 
\vartheta^{(2)}(x;k,\ell)=\theta^{(2)}(x,\ell) \quad\text{for \ } {\rm dist}(x,(k,\ell))\ge R_0
\]
with a large positive constant $R_0$. Namely, the functions $\vartheta^{(j)}(x;k,\ell)$ 
coincide with the step functions at the large distances. At short distances, 
the functions $\vartheta^{(j)}(x;k,\ell)$ can take any real values. 
We define  
\[
\vartheta^{(j)}(k,\ell):=\sum_{x\in\ze^2}\vartheta^{(j)}(x;k,\ell)\; n_x.
\]
Then, the twisted Hamiltonian is formally given by 
\[
H_0^{(\ze^2)}(\boldsymbol{\phi}):=\exp\bigl[-\sum_{j=1,2}i\phi_j\vartheta^{(j)}(k,\ell)\bigr]
H_0^{(\ze^2)}
\exp\bigl[\sum_{j=1,2}i\phi_j\vartheta^{(j)}(k,\ell)\bigr]. 
\]
The generic local current operators are 
\[
\mathcal{J}^{(j)}(k,\ell):=i[H_0^{(\ze^2)},\vartheta^{(j)}(k,\ell)]. 
\]
The potential operator of (\ref{chiGamma}) can be also replaced by a generic 
electric potential such that, except the boundary ribbon region,  
the potential takes the unit inside the finite large region and zero outside region. 
Namely, the potential coincides with the characteristic function of the finite region 
at the large distances from the boundary of the region, 
while the potential takes any values at the short distances from the boundary of the region. 

\section{Estimates of the Current-Current Correlations}
\label{CurrentCurrentCorr}

The conductance (\ref{tildesigma12}) is written 
\begin{align*}
&{\tilde\sigma}_{12}^{(\Lambda,N)}(\eta,T,\mathcal{N},\mathcal{M})=
i\int_{-T}^0ds\; se^{\eta s}\\
&\times\frac{1}{q}\sum_{m=1}\sum_{n\ge 1}
\left[\langle\Phi_{0,m}^{(N)},J_\mathcal{N}^{(1)}(k,\ell)\Phi_n^{(N)}\rangle
\langle\Phi_n^{(N)},J(\Gamma_\mathcal{M}(k,\ell))\Phi_{0,m}^{(N)}\rangle e^{i(E_n^{(N)}-E_{0,m}^{(N)})s}\right.\\
&-\left.\langle \Phi_{0,m}^{(N)}, J(\Gamma_\mathcal{M}(k,\ell))\Phi_n^{(N)}\rangle
\langle\Phi_n^{(N)},J_\mathcal{N}^{(1)}(k,\ell)\Phi_{0,m}^{(N)}\rangle e^{-i(E_n^{(N)}-E_{0,m}^{(N)})s}\right]
\end{align*}
in terms of the ground-state vectors $\Phi_{0,m}^{(N)}$ with the eigenvalue $E_{0,m}^{(N)}$, 
$m=1,2,\ldots,q$, and 
the excited-state vectors $\Phi_n^{(N)}$ with the eigenvalue $E_n^{(N)}$, $n\ge1$. 
Here, we have dropped the twisted phase dependence for short.  
Consider the limit 
\[
{\tilde\sigma}_{12}^{(\Lambda,N)}(0,\infty,\mathcal{N},\mathcal{M})
:=\lim_{\eta\rightarrow 0}\lim_{T\rightarrow\infty}
{\tilde\sigma}_{12}^{(\Lambda,N)}(\eta,T,\mathcal{N},\mathcal{M}).
\]
Using the integral formula (\ref{IntsexpE}) in Appendix~\ref{CorrHallcon}, this is computed as  
\begin{multline*}
\tilde{\sigma}_{12}^{(\Lambda,N)}(0,\infty,\mathcal{N},\mathcal{M})\\
=\frac{i}{q}\sum_{m=1}^q 
\sum_{n\ge 1}\left[
\frac{
\langle \Phi_{0,m}^{(N)}, J_\mathcal{N}^{(1)}(k,\ell)\Phi_n^{(N)}\rangle
\langle\Phi_n^{(N)},J(\Gamma_\mathcal{M}(k,\ell))\Phi_{0,m}^{(N)}\rangle}
{(E_{0,m}^{(N)}-E_n^{(N)})^2}\right.\\
-\left.\frac{
\langle \Phi_{0,m}^{(N)}, J(\Gamma_\mathcal{M}(k,\ell))\Phi_n^{(N)}\rangle
\langle\Phi_n^{(N)},J_\mathcal{N}^{(1)}(k,\ell)\Phi_{0,m}^{(N)}\rangle}
{(E_{0,m}^{(N)}-E_n^{(N)})^2}\right],
\end{multline*}
where we have written 
\[
J(\Gamma_\mathcal{M}(k,\ell)):=J^{(\Lambda)}(\Gamma_\mathcal{M}(k,\ell);0) 
\]
for short. In the same way as in Appendix~\ref{CorrHallcon}, we have 
\begin{multline}
\label{sigmaetaTdepend}
\left|{\tilde\sigma}_{12}^{(\Lambda,N)}(\eta,T,\mathcal{N},\mathcal{M})
-{\tilde\sigma}_{12}^{(\Lambda,N)}(0,\infty,\mathcal{N},\mathcal{M})\right|\\
\le 2\left[\frac{2\Delta E+\eta}{\Delta E^4}\eta+\frac{1+T\Delta E}{\Delta E^2}e^{-\eta T}\right]
\big\Vert J_\mathcal{N}^{(1)}(k,\ell)\big\Vert\; \big\Vert J(\Gamma_\mathcal{M}(k,\ell))\big\Vert.
\end{multline}

Both of the current operators, $J_\mathcal{N}^{(1)}(k,\ell)$ and 
$J(\Gamma_\mathcal{M}(k,\ell))$, are written as a sum of local operators. 
We want to show that the dominant contributions to the conductance in the double sum are given by 
local operators in the neighborhood of the point $(k,\ell)$. 
In other words, if the distance between two local operators in the sums is 
sufficiently large, then the corresponding contribution are negligible.  

We define 
\[
\langle\!\langle A; B\rangle\!\rangle:=\frac{1}{q}\sum_{q=1}^q\sum_{n\ge 1}
\langle \Phi_{0,m}^{(N)},A\Phi_n^{(N)}\rangle\frac{1}{(E_{0,m}^{(N)}-E_n^{(N)})^2}
\langle \Phi_n^{(N)}, B\Phi_{0,m}^{(N)}\rangle
\]
for local observables $A$ and $B$. 

\begin{prop}
\label{prop:corrABdecay}
The following bound holds: 
\[
\left|\langle\!\langle A; B\rangle\!\rangle\right|\le Ce^{-\kappa r},
\]
where $r={\rm dist}({\rm supp}\> A,{\rm supp}\> B)$, and $C$ and $\kappa$ are positive constants. 
\end{prop}

\begin{proof}
We define 
\[
{\tilde B}:=B-P_0^{(\Lambda,N)}BP_0^{(\Lambda,N)}
\]
and 
\[
{\tilde B}^{(\Lambda)}(z):=e^{iH_0^{(\Lambda)}z}{\tilde B}e^{-iH_0^{(\Lambda)}z}\quad\text{for \ }z\in \co. 
\]
Using the integral identity
\[
\frac{1}{\Delta \mathcal{E}^2}=\int_0^\infty ds \; s\> e^{-\Delta \mathcal{E}s}
\]
with $\Delta \mathcal{E}>0$, we have 
\begin{multline*}
\langle\!\langle A; B\rangle\!\rangle\\
:=\frac{1}{q}\sum_{m=1}^q\sum_{n\ge 1}\int_0^\infty ds\; s
\langle \Phi_{0,m}^{(N)},A\Phi_n^{(N)}\rangle\exp[-(E_n^{(N)}-E_{0,m}^{(N)})s]
\langle \Phi_n^{(N)}, \tilde{B}\Phi_{0,m}^{(N)}\rangle\\
=\frac{1}{q}\sum_{m=1}^q\sum_{n\ge 1}\int_0^{cr} ds\; s
\langle \Phi_{0,m}^{(N)},A\Phi_n^{(N)}\rangle\exp[-(E_n^{(N)}-E_{0,m}^{(N)})s]
\langle \Phi_n^{(N)}, \tilde{B}\Phi_{0,m}^{(N)}\rangle\\
+\frac{1}{q}\sum_{m=1}^q\sum_{n\ge 1}\int_{cr}^\infty ds\; s
\langle \Phi_{0,m}^{(N)},A\Phi_n^{(N)}\rangle\exp[-(E_n^{(N)}-E_{0,m}^{(N)})s]
\langle \Phi_n^{(N)}, \tilde{B}\Phi_{0,m}^{(N)}\rangle,
\end{multline*}
where $c$ is a positive constant. Clearly, the second sum leads to the desired bound 
because of the spectral gap above the ground state. The first sum is written 
\[
\int_0^{cr}ds\; s\;\omega_0^{(\Lambda,N)}(A\tilde{B}^{(\Lambda)}(is)). 
\]
Therefore, it is sufficient to estimate $\omega_0^{(\Lambda,N)}(A\tilde{B}(ib))$ for 
$0\le b\le cr$. By definition, one has 
\[
\omega_0^{(\Lambda,N)}(A\tilde{B}^{(\Lambda)}(ib))=\omega_0^{(\Lambda,N)}(AB^{(\Lambda)}(ib))-
\omega_0^{(\Lambda,N)}(AP_0^{(\Lambda,N)}B^{(\Lambda)}(ib))
\]
where we have written 
\[
B^{(\Lambda)}(z):=e^{iH_0^{(\Lambda)}z}Be^{-iH_0^{(\Lambda)}z}\quad\text{for \ }z\in \co. 
\]
For a large distance $r={\rm dist}({\rm supp}\; A,{\rm supp}\; B)$, we can expect 
the exponential clustering   
\[
\omega_0^{(\Lambda,N)}(AB^{(\Lambda)}(ib))\sim 
\omega_0^{(\Lambda,N)}(AP_0^{(\Lambda,N)}B^{(\Lambda)}(ib))+\mathcal{O}(e^{-\kappa r})
\]
for a small $b$ because of the spectral gap above the ground state. Actually, we can prove this by 
using the Lieb-Robinson bounds in the same way as that for 
evaluating the integrand in (\ref{intomega0aastP0aB}) in Appendix~\ref{subsecTPDEE}. 
(See also Theorem~2 in \cite{NachtergaeleSims}.) 
\end{proof}

Using Proposition~\ref{prop:corrABdecay}, we have 
\begin{equation}
\label{sigmaNMdepend}
\lim_{\mathcal{N}\rightarrow\infty}\lim_{\mathcal{M}\rightarrow\infty}
\lim_{\Lambda\nearrow\ze^2}
\left|\tilde{\sigma}_{12}^{(\Lambda,N)}(0,\infty,\mathcal{N},\mathcal{M},\boldsymbol{\phi})
-{\hat\sigma}_{12}^{(\Lambda,N)}(\boldsymbol{\phi})\right|=0
\end{equation}
for the Hall conductance ${\hat\sigma}_{12}^{(\Lambda,N)}(\boldsymbol{\phi})$ of (\ref{hatsigma12phi}) 
with the twisted boundary phase $\boldsymbol{\phi}$.
 
\section{Twisted Phase Dependence of the Hall Conductance}
\label{sec:TPDHallcond}

We write
\[
{\tilde\sigma}_{12}^{(\Lambda,N)}(\eta,T,\mathcal{N},\mathcal{M},\boldsymbol{\phi})
\]
for the conductance of (\ref{tildesigma12}) in the case of the twisted boundary conditions. 

We write 
\[
A^{(\Lambda)}(t):=e^{itH_0^{(\Lambda)}}Ae^{-itH_0^{(\Lambda)}}
\]
for the time evolution for a local operator $A$. Let $\Omega$ be a subset of 
$\Lambda$, and write $H_0^{(\Omega)}$ for the Hamiltonian $H_0^{(\Lambda)}$ restricted 
to the subset $\Omega$. 
Let $A$ be a local observable with ${\rm supp}\; A\subset\Omega$, and 
define 
\[
A^{(\Omega)}(t):=e^{itH_0^{(\Omega)}}Ae^{-itH_0^{(\Omega)}}.
\]
This is the time evolution of $A$ on the region $\Omega$. 
Then, one has 
\begin{equation}
\bigl\Vert A^{(\Lambda)}(t)-A^{(\Omega)}(t)\bigr\Vert
\le \int_0^t ds\; \bigl\Vert [(H_0^{(\Lambda)}-H_0^{(\Omega)}),A^{(\Omega)}(t-s)]\bigr\Vert  
\label{timeevolbound}
\end{equation}
for $t\ge 0$. When the distance between the supports of 
two observables $(H_0^{(\Lambda)}-H_0^{(\Omega)})$ and $A$ is sufficiently large, 
the norm of the commutator in the integrand in the right-hand side becomes very small for a finite $t$.  
This inequality can be proved by using the Lieb-Robinson bounds in Appendix~\ref{Proofa0tdiffinequal}. 

We define 
\begin{multline}
\label{tildesigma12Omega}
{\tilde\sigma}_{12}^{(\Lambda,N)}(\eta,T,\mathcal{N},\mathcal{M},\Omega,\boldsymbol{\phi})\\
:=i\int_{-T}^0ds \; s e^{\eta s}\omega_0^{(\Lambda,N)}
\bigl(\bigl[J_\mathcal{N}^{(1)}(k,\ell),
J^{(\Omega)}(\Gamma_\mathcal{M}(k,\ell);s)\bigr];\boldsymbol{\phi}\bigr).
\end{multline}
This is given by replacing the current operator $J^{(\Lambda)}(\Gamma_\mathcal{M}(k,\ell);s)$ 
by 
\[
J^{(\Omega)}(\Gamma_\mathcal{M}(k,\ell);s):=e^{itH_0^{(\Omega)}}
J(\Gamma_\mathcal{M}(k,\ell))e^{-itH_0^{(\Omega)}}
\]
in the expression (\ref{tildesigma12}) of the conductance 
${\tilde\sigma}_{12}^{(\Lambda,N)}(\eta,T,\mathcal{N},\mathcal{M},\boldsymbol{\phi})$ 
in the case of the twisted boundary conditions. 
Note that 
\begin{multline}
{\tilde\sigma}_{12}^{(\Lambda,N)}(\eta,T,\mathcal{N},\mathcal{M},\boldsymbol{\phi})
-{\tilde\sigma}_{12}^{(\Lambda,N)}(\eta,T,\mathcal{N},\mathcal{M},\Omega,\boldsymbol{\phi})\\
=i\int_{-T}^0ds \> s e^{\eta s}\omega_0^{(\Lambda,N)}
\bigl(\bigl[J_\mathcal{N}^{(1)}(k,\ell),J^{(\Lambda)}(\Gamma_\mathcal{M}(k,\ell);s)-
J^{(\Omega)}(\Gamma_\mathcal{M}(k,\ell);s)\bigr];\boldsymbol{\phi}\bigr).
\end{multline}
Using the above bound (\ref{timeevolbound}), we can prove that 
the difference between the two conductances becomes very small for a large lattice $\Lambda$.  
Actually, we choose the region $\Omega$ so that 
$$
{\rm dist}(\Omega,\partial\Lambda)=\mathcal{O}(L)
$$
and 
$$
{\rm dist}(\partial\Omega,\Gamma_\mathcal{M}(k,\ell))=\mathcal{O}(L),
$$
where $\partial\Lambda$ and $\partial\Omega$ denote the boundary of the regions $\Lambda$ and 
$\Omega$, respectively, and $L=\min\{L^{(1)},L^{(2)}\}$. Then, we have 
\begin{equation}
\label{diffsigmadiffLambdaOmega}
{\tilde\sigma}_{12}^{(\Lambda,N)}(\eta,T,\mathcal{N},\mathcal{M},\boldsymbol{\phi})
-{\tilde\sigma}_{12}^{(\Lambda,N)}(\eta,T,\mathcal{N},\mathcal{M},\Omega,\boldsymbol{\phi})\rightarrow 
0\quad\mbox{as}\ \Lambda \nearrow\ze^2
\end{equation}
by using the Lieb-Robinson bounds (\ref{LiebRobinson}) below in Appendix~\ref{ConstLEE}.
(See also \cite{LiebRobinson,Hastings,HastingsKoma,NachtergaeleSims}.)   

In the right-hand side of (\ref{tildesigma12Omega}), 
the support of the commutator of the two current operators is a finite subset of $\Lambda$ and 
apart from the boundary of $\Lambda$. Therefore, we can expect that the effect of the twisted 
boundary condition is exponentially small in the size $L$ 
in the approximated Hall conductance (\ref{tildesigma12Omega}). 
Actually, we can prove the inequality (\ref{diffsigmadiffphi}) below. 

To begin with, we note that the expectation value of a local operator $A$ for the ground state of 
the Hamiltonian $H_0^{(\Lambda,N)}(\boldsymbol{\phi})$ with the twisted boundary conditions is 
written 
\[
\omega_0^{(\Lambda,N)}(A;\boldsymbol{\phi})
=\frac{1}{q}{\rm Tr}\; AP_0^{(\Lambda,N)}(\boldsymbol{\phi}),
\]
where $P_0^{(\Lambda,N)}(\boldsymbol{\phi})$ is the projection onto the sector of the ground state.  
Using contour integral, one has 
\[
P_0^{(\Lambda,N)}(\boldsymbol{\phi})=\frac{1}{2\pi i}\oint dz 
\frac{1}{z-H_0^{(\Lambda,N)}(\boldsymbol{\phi})}. 
\]
Further, by relying on the existence of the spectral gap above the sector of the ground-state, 
one obtains 
\[
\frac{\partial}{\partial \phi_j}P_0^{(\Lambda,N)}(\boldsymbol{\phi})=\frac{1}{2\pi i}\oint dz 
\frac{1}{z-H_0^{(\Lambda,N)}(\boldsymbol{\phi})}
B_j(\boldsymbol{\phi})
\frac{1}{z-H_0^{(\Lambda,N)}(\boldsymbol{\phi})}, 
\]
where we have written 
\[
B_j(\boldsymbol{\phi}):=\frac{\partial}{\partial\phi_j}H_0^{(\Lambda)}(\boldsymbol{\phi}).
\]
{From} this, we have  
\[
\frac{\partial}{\partial \phi_j}{\rm Tr}\; AP_0^{(\Lambda,N)}(\boldsymbol{\phi})
=\frac{1}{2\pi i}\oint dz\; {\rm Tr}\; A
\frac{1}{z-H_0^{(\Lambda,N)}(\boldsymbol{\phi})}
B_j(\boldsymbol{\phi})
\frac{1}{z-H_0^{(\Lambda,N)}(\boldsymbol{\phi})}. 
\]
Integrating both side in the case of $j=1$, we obtain 
\begin{multline*}
{\rm Tr}\; AP_0^{(\Lambda,N)}(\phi_1,\phi_2)-{\rm Tr}\; AP_0^{(\Lambda,N)}(0,\phi_2)\\
=\int_0^{\phi_1}d\phi' \frac{1}{2\pi i}\oint dz\; {\rm Tr}\; A
\frac{1}{z-H_0^{(\Lambda,N)}(\boldsymbol{\phi}')}
B_1(\boldsymbol{\phi}')
\frac{1}{z-H_0^{(\Lambda,N)}(\boldsymbol{\phi}')},
\end{multline*}
where we have written $\boldsymbol{\phi}'=(\phi',\phi_2)$.
Therefore, 
\begin{multline*}
\omega_0^{(\Lambda,N)}(A;\boldsymbol{\phi})-\omega_0^{(\Lambda,N)}(A;0,\phi_2)\\
=\int_0^{\phi_1}d\phi' \frac{1}{q}\sum_{m=1}^q\sum_{n\ge 1}
\frac{\langle\Phi_{0,m}^{(N)}(\boldsymbol{\phi}'), A\Phi_n^{(N)}(\boldsymbol{\phi}')\rangle
\langle\Phi_n^{(N)}(\boldsymbol{\phi}'), B_1(\boldsymbol{\phi}')\Phi_{0,m}^{(N)}(\boldsymbol{\phi}')\rangle}
{E_{0,m}^{(N)}(\boldsymbol{\phi}')-E_n^{(N)}(\boldsymbol{\phi}')}\\
+(A\leftrightarrow  B_1(\boldsymbol{\phi}')).
\end{multline*}
This right-hand side can be evaluated 
in the same way as in the preceding Sec.~\ref{CurrentCurrentCorr}. 
Thus, the difference of the two conductances with the different phases is exponentially 
small in the linear size $L=\min\{L^{(1)},L^{(2)}\}$ of the lattice $\Lambda$ as 
\begin{multline*} 
\left|{\tilde\sigma}_{12}^{(\Lambda,N)}(\eta,T,\mathcal{N},\mathcal{M},\Omega,\phi_1,\phi_2)
-{\tilde\sigma}_{12}^{(\Lambda,N)}(\eta,T,\mathcal{N},\mathcal{M},\Omega,0,\phi_2)\right|\\
\le\mathcal{C}(\eta,T)\times \mathcal{O}(\exp[-{\rm Const.}L]), 
\end{multline*}
where  
$$
\mathcal{C}(\eta,T):=\left|\int_{-T}^0ds\; se^{\eta s}\right|. 
$$
In the same way, we obtain  
\begin{align}
\label{diffsigmadiffphi}
&\left|{\tilde\sigma}_{12}^{(\Lambda,N)}(\eta,T,\mathcal{N},\mathcal{M},\Omega,\phi_1,\phi_2)
-{\tilde\sigma}_{12}^{(\Lambda,N)}(\eta,T,\mathcal{N},\mathcal{M},\Omega,0,0)\right|\\
\nonumber &\le \left|{\tilde\sigma}_{12}^{(\Lambda,N)}(\eta,T,\mathcal{N},\mathcal{M},\Omega,\phi_1,\phi_2)
-{\tilde\sigma}_{12}^{(\Lambda,N)}(\eta,T,\mathcal{N},\mathcal{M},\Omega,0,\phi_2)\right|\\
\nonumber &+\left|{\tilde\sigma}_{12}^{(\Lambda,N)}(\eta,T,\mathcal{N},\mathcal{M},\Omega,0,\phi_2)
-{\tilde\sigma}_{12}^{(\Lambda,N)}(\eta,T,\mathcal{N},\mathcal{M},\Omega,0,0)\right|\\
\nonumber &\le\mathcal{C}(\eta,T)\times \mathcal{O}(\exp[-{\rm Const.}L]).  
\end{align}
This is the desired inequality.

\section{Proof of Theorem~\ref{thm:main}}
\label{sec:proofmain}

Now, we shall complete the proof of Theorem~\ref{thm:main}. 
The difference between the two conductances with the twisted phases, $\boldsymbol{\phi}=(\phi_1,\phi_2)$ and 
$(0,0)$, is evaluated as  
\begin{align}
\label{sigmaTphasedepen}
&\left|{\tilde\sigma}_{12}^{(\Lambda,N)}(\eta,T,\mathcal{N},\mathcal{M},\phi_1,\phi_2)
-{\tilde\sigma}_{12}^{(\Lambda,N)}(\eta,T,\mathcal{N},\mathcal{M},0,0)\right|\\
\nonumber &\le \left|{\tilde\sigma}_{12}^{(\Lambda,N)}(\eta,T,\mathcal{N},\mathcal{M},\phi_1,\phi_2)
-{\tilde\sigma}_{12}^{(\Lambda,N)}(\eta,T,\mathcal{N},\mathcal{M},\Omega,\phi_1,\phi_2)\right|\\
\nonumber &+\left|{\tilde\sigma}_{12}^{(\Lambda,N)}(\eta,T,\mathcal{N},\mathcal{M},\Omega,\phi_1,\phi_2)
-{\tilde\sigma}_{12}^{(\Lambda,N)}(\eta,T,\mathcal{N},\mathcal{M},\Omega,0,0)\right|\\
\nonumber &+\left|{\tilde\sigma}_{12}^{(\Lambda,N)}(\eta,T,\mathcal{N},\mathcal{M},\Omega,0,0)
-{\tilde\sigma}_{12}^{(\Lambda,N)}(\eta,T,\mathcal{N},\mathcal{M},0,0)\right|.
\end{align}
{From} (\ref{diffsigmadiffLambdaOmega}) and (\ref{diffsigmadiffphi}), the right-hand side is 
vanishing in the limit $\Lambda\nearrow\ze^2$. 
By using the formula (\ref{avHallformula}) for the averaged Hall conductance, we have  
\begin{multline*}
\left|{\tilde\sigma}_{12}^{(\Lambda,N)}(\eta,T,\mathcal{N},\mathcal{M},0,0)-\frac{1}{2\pi}\frac{p}{q}\right|\\
=\left|\frac{1}{(2\pi)^2}\int_{[0,2\pi)\times[0,2\pi)}d\phi_1d\phi_2\;
{\tilde\sigma}_{12}^{(\Lambda,N)}(\eta,T,\mathcal{N},\mathcal{M},0,0)\right.\\
-\left.\frac{1}{(2\pi)^2}\int_{[0,2\pi)\times[0,2\pi)}d\phi_1d\phi_2\; 
\hat{\sigma}_{12}^{(\Lambda,N)}(\boldsymbol{\phi})\right|\\
\le\frac{1}{(2\pi)^2}\int_{[0,2\pi)\times[0,2\pi)}d\phi_1d\phi_2\;
\left|{\tilde\sigma}_{12}^{(\Lambda,N)}(\eta,T,\mathcal{N},\mathcal{M},0,0)
-\hat{\sigma}_{12}^{(\Lambda,N)}(\boldsymbol{\phi})\right|.
\end{multline*}
Therefore, in order to 
prove Theorem~\ref{thm:main}, it is sufficient to show that the integrand in the right-hand side is 
vanishing in the multiple limit in (\ref{sigma12}). The integrand is estimated as 
\begin{align*}
&\left|{\tilde\sigma}_{12}^{(\Lambda,N)}(\eta,T,\mathcal{N},\mathcal{M},0,0)
-\hat{\sigma}_{12}^{(\Lambda,N)}(\boldsymbol{\phi})\right|\\
&\le \left|{\tilde\sigma}_{12}^{(\Lambda,N)}(\eta,T,\mathcal{N},\mathcal{M},0,0)
-{\tilde\sigma}_{12}^{(\Lambda,N)}(\eta,T,\mathcal{N},\mathcal{M},\phi_1,\phi_2)\right|\\
&+\left|{\tilde\sigma}_{12}^{(\Lambda,N)}(\eta,T,\mathcal{N},\mathcal{M},\phi_1,\phi_2)
-{\tilde\sigma}_{12}^{(\Lambda,N)}(0,\infty,\mathcal{N},\mathcal{M},\phi_1,\phi_2)\right|\\
&+\left|{\tilde\sigma}_{12}^{(\Lambda,N)}(0,\infty,\mathcal{N},\mathcal{M},\phi_1,\phi_2)
-\hat{\sigma}_{12}^{(\Lambda,N)}(\boldsymbol{\phi})\right|.
\end{align*}
As shown in (\ref{sigmaTphasedepen}) above in the present section, the first term in the right-hand side 
is vanishing in the infinite-volume limit $\Lambda\nearrow\ze^2$. 
Relying on the estimate (\ref{sigmaetaTdepend}), we can show that the second term is vanishing 
in the double limit $\eta\rightarrow 0$ and $T\rightarrow\infty$ after taking the infinite-volume limit.   
Finally, the result (\ref{sigmaNMdepend}) implies that the third term in the right-hand side is 
vanishing in the limit. 

Since the above Hall conductance ${\tilde\sigma}_{12}^{(\Lambda,N)}(\eta,T,\mathcal{N},\mathcal{M},0,0)$ 
with the vanishing phases is 
equal to the Hall conductance ${\tilde\sigma}_{12}^{(\Lambda,N)}(\eta,T,\mathcal{N},\mathcal{M})$ 
in the case of the periodic boundary condition 
in the right-hand side of (\ref{sigma12}), we have proved the fractional quantization (\ref{FQHC}) 
for the Hall conductance $\sigma_{12}$ in the limit.  

\appendix

\section{The Correction to the Hall Conductance}
\label{CorrHallcon}

In this appendix, we show that the correction 
\[
I_{\rm cor}:=\int_{-T}^0ds \; \eta s e^{\eta s}\omega_0^{(\Lambda,N)}
\bigl([\chi^{(\Lambda)}(\Gamma_\mathcal{M}(k,\ell);s),J_\mathcal{N}^{(1)}(k,\ell)]\bigr)
\]
in (\ref{corlinearres}) to the Hall conductance is vanishing in the limit $\eta\rightarrow 0$ 
after taking the limit $T\rightarrow \infty$. 

We denote the ground-state vectors of the unperturbed Hamiltonian\break\hfill 
$H_0^{(\Lambda,N)}$ by 
$\Phi_{0,m}^{(N)}$ with the energy eigenvalue $E_{0,m}^{(N)}$, $m=1,2,\ldots,q$, 
and the excited-state vectors by $\Phi_n^{(N)}$ with the energy eigenvalues $E_n^{(N)}$, $n=1,2,\ldots$. 
Using the definition (\ref{defchiGammas}) of $\chi^{(\Lambda)}(\Gamma_\mathcal{M}(k,\ell);s)$, 
one has 
\begin{multline*}
\omega_0^{(\Lambda,N)}
\bigl([\chi^{(\Lambda)}(\Gamma_\mathcal{M}(k,\ell);s),J_\mathcal{N}^{(1)}(k,\ell)]\bigr)\\
=\frac{1}{q}\sum_{m=1}^q\sum_{n\ge 1}\biggl[\langle\Phi_{0,m}^{(N)},\chi^{(\Lambda)}(\Gamma_\mathcal{M}(k,\ell))
\Phi_n^{(N)}\rangle\langle\Phi_n^{(N)},J_\mathcal{N}^{(1)}(k,\ell)\Phi_{0,m}^{(N)}\rangle
e^{i(E_{0,m}^{(N)}-E_n^{(N)})s}\\
-{\rm c.c.}\biggr].
\end{multline*}
Further, one has 
\begin{multline}
\label{IntsexpE}
\int_{-T}^0ds \; s e^{\eta s}e^{i(E_{0,m}^{(N)}-E_n^{(N)})s}\\
=\frac{iT}{E_n^{(N)}-E_{0,m}^{(N)}+i\eta}e^{-\eta T}e^{i(E_n^{(N)}-E_{0,m}^{(N)})T}
+\frac{1}{\left(E_n^{(N)}-E_{0,m}^{(N)}+i\eta\right)^2}\\
-\frac{1}{\left(E_n^{(N)}-E_{0,m}^{(N)}+i\eta\right)^2}e^{-\eta T}e^{i(E_n^{(N)}-E_{0,m}^{(N)})T}.
\end{multline}
{From} these observations, we can evaluate the integral as 
\[
\lim_{T\rightarrow\infty}\left|I_{\rm cor}\right|\le \frac{2\eta}{\Delta E^2}
\bigl\Vert\chi^{(\Lambda)}(\Gamma_\mathcal{M}(k,\ell))\bigr\Vert \> 
\bigl\Vert J_\mathcal{N}^{(1)}(k,\ell)\bigr\Vert, 
\] 
where we have used the Schwarz inequality and the assumption of the spectral gap, 
$E_n^{(N)}-E_{0,m}^{(N)}\ge \Delta E$. Thus, we can obtain the desired result 
in the limit $\eta\rightarrow 0$ for $\mathcal{M}$ and $\mathcal{N}$ which are finite.

\section{Proof of Theorem~\ref{thm:twistDegeneGap}}
\label{ProoftwistDegeneGap}

\subsection{Constructing the low-energy excitation}
\label{ConstLEE}

Let us consider a generic Hamiltonian $H^{(\Lambda)}$ with finite range hopping amplitudes and 
finite range interactions on the lattice $\Lambda$. 
We assume that the Hamiltonian $H^{(\Lambda)}$ commutes with the total number operator of 
the fermions. Let $N$ be the number of fermions, and fix the filling factor $\nu=N/|\Lambda|$. 
We denote by $H^{(\Lambda,N)}$ the restriction of $H^{(\Lambda)}$ onto the eigenspace of 
the total number operator with the eigenvalue $N$. 
Further we assume that 
there exists a uniform spectral gap $\Delta E$ above the lowest $q$ eigenenergies $E_{0,m}^{(N)}$, 
$m=1,2,\ldots,q$. We write $\Phi_{0,m}^{(N)}$ for the corresponding $q$ eigenvectors, and define 
\[
\delta E:=\max_{m,m'}\bigl\{\bigl|E_{0,m}^{(N)}-E_{0,m'}^{(N)}\bigr|\bigl\}.
\]
We do not assume that the $q$ energies $E_{0,m}^{(N)}$ are (quasi)degenerate. 
Therefore, we allow $\delta E>0$ uniformly in the size $|\Lambda|$ of the lattice. 
We write $E_1^{(N)}$ for the $q+1$-th eigenenergy of the Hamiltonian $H^{(\Lambda,N)}$ from 
the bottom of the spectrum. 
Let $P_0^{(\Lambda,N)}$ be the projection onto the sector which is spanned by  
the $q$ vectors $\Phi_{0,m}^{(N)}$, and define the expectation for the sector as   
\[
\omega_0^{(\Lambda,N)}(\cdots)
:=\frac{1}{q}\sum_{m=1}^q\langle \Phi_{0,m}^{(N)},\cdots\Phi_{0,m}^{(N)}\rangle.
\]

\begin{prop}
\label{prop:expexcit}
Suppose that, for any given small $\varepsilon>0$, there exist a local observable $a_0$ 
with a compact support, and a positive constant $c_0$ such that  
\begin{equation}
\label{excitationcond}
\langle\Phi_{0,i}^{(N)},a_0^\ast P_{\varepsilon/2} a_0\Phi_{0,i}^{(N)}\rangle >c_0>0
\end{equation}
for a vector $\Phi_{0,i}^{(N)}$ in the set of the $q$ vectors $\{\Phi_{0,m}^{(N)}\}$, 
where $P_{\varepsilon/2}$ is the projection onto 
the energy interval $[E_1^{(N)},E_1^{(N)}+\varepsilon/2]$. 
Here both the support size of the observable $a_0$ and 
the constant $c$ can be chosen to be independent of the size $|\Lambda|$ of the lattice. 
Then, there exist a local observable $a$ with a compact support and a positive constant $V_0$ such that 
\begin{equation}
\label{excitbound}
\Delta E\le \frac{\omega_0^{(\Lambda,N)}(a^\ast (1-P_0^{(\Lambda,N)})[H^{(\Lambda)},a])}
{\omega_0^{(\Lambda,N)}(a^\ast(1-P_0^{(\Lambda,N)})a)}\le \Delta E+2\delta E+ 2\varepsilon
\end{equation}
for any $\Lambda$ satisfying $|\Lambda|>{V}_0$. Here the size of the support of the observable $a$ 
is independent of the size $|\Lambda|$ of the lattice. 
\end{prop}

\begin{rem}
If the condition (\ref{excitationcond}) does not hold, then there must exist  
another infinite-volume ground state which cannot be derived from the $q$ vectors. 
In other words, all of the low energy excited states are derived from a local perturbation 
for the ground state.  
\end{rem}

In the rest of this subsection, we will prove the upper bound of (\ref{excitbound}) 
in Proposition~\ref{prop:expexcit} because the lower bound can be easily obtained from the definitions. 

We decompose the energy interval in Proposition~\ref{prop:expexcit} into $M$ small intervals as 
\[
[E_1^{(N)},E_1^{(N)}+\varepsilon/2]
=\bigcup_{n=1}^M\> [E_1^{(N)}+(n-1)\varepsilon/(2M),E_1^{(N)}+n\varepsilon/(2M)], 
\]
where $M$ is a large positive integer. Then, there exists 
a small interval 
\[
[E_1^{(N)}+(\kappa-1)\varepsilon/(2M),E_1^{(N)}+\kappa\varepsilon/(2M)]
\]  
such that
\begin{equation}
\label{Pvarepsilon/2mat}
\langle\Phi_{0,i}^{(N)},a_0^\ast P_{\varepsilon/2}^\kappa a_0\Phi_{0,i}^{(N)}\rangle \ge\frac{c_0}{M},
\end{equation}
where $\kappa$ is an integer in $\{1,2,\ldots,M\}$, 
and $P_{\varepsilon/2}^\kappa$ is the spectral projection onto the small interval. 

Let $K$ be a large positive number satisfying $\varepsilon\sqrt{K}/4\le M<\varepsilon\sqrt{K}/4-1$, 
and write  
\[
\Delta\tilde{E}:=E_1^{(N)}-E_{0,i}^{(N)},
\]
and
\[
\tilde{\kappa}:=(\kappa-{1}/{2})\frac{\varepsilon}{2M}. 
\]
We use the idea of an energy filter \cite{HMNOSV}. We define three operators as follows: 
\[
\tilde{a}_0:=\frac{1}{\sqrt{2\pi K}}\int_{-\infty}^{+\infty}dt\; e^{itH^{(\Lambda)}}a_0
e^{-itH^{(\Lambda)}}e^{-i(\Delta\tilde{E}+\tilde{\kappa})t}e^{-t^2/(2K)},
\]
\[
\tilde{a}_0(T_1):=\frac{1}{\sqrt{2\pi K}}\int_{-T_1}^{+T_1}dt\; e^{itH^{(\Lambda)}}a_0
e^{-itH^{(\Lambda)}}e^{-i(\Delta\tilde{E}+\tilde{\kappa})t}e^{-t^2/(2K)},
\]
and
\[
a(T_1):=\frac{1}{\sqrt{2\pi K}}\int_{-T_1}^{+T_1}dt\; e^{itH^{(\Omega)}}a_0
e^{-itH^{(\Omega)}}e^{-i(\Delta\tilde{E}+\tilde{\kappa})t}e^{-t^2/(2K)},
\]
where the cutoff $T_1$ is a large positive number, and $H^{(\Omega)}$ is the restriction 
of the Hamiltonian $H^{(\Lambda)}$ to the region $\Omega\subset\Lambda$. 
We choose $\Omega$ so that ${\rm supp}\; a_0\subset \Omega$. 
Clearly, the operator $a(T_1)$ is local and has a compact support in $\Omega$. 
We write $a=a(T_1)$ for short, and we will prove that this observable $a$ satisfies 
the upper bound of (\ref{excitbound}) in Proposition~\ref{prop:expexcit} for appropriately 
choosing the parameters. 

We decompose the projection $(1-P_0^{(\Lambda,N)})$ into two parts as 
\[
1-P_0^{(\Lambda,N)}=P_{\rm low}+P_{\rm high},
\]
where $P_{\rm low}$ is the spectral projection onto the interval 
$[E_1^{(N)},E_1^{(N)}+\varepsilon+\delta E]$, and $P_{\rm high}$ the spectral projection 
onto $(E_1^{(N)}+\varepsilon+\delta E,+\infty)$. Then, we have 
\begin{multline}
\label{decomp1-p0}
\omega_0^{(\Lambda,N)}(a^\ast(1-P_0^{(\Lambda,N)})[H^{(\Lambda)},a])\\
=\omega_0^{(\Lambda,N)}(a^\ast P_{\rm low}[H^{(\Lambda)},a])
+\omega_0^{(\Lambda,N)}(a^\ast P_{\rm high}[H^{(\Lambda)},a]). 
\end{multline}
Since ${\rm supp}\> a\subset \Omega$, we can find a region $\tilde{\Omega}\supset \Omega$ 
so that $[H^{(\Lambda)},a]=[H^{(\tilde{\Omega})},a]$. 
Using this and the Schwarz inequality, the second term in the right-hand side of (\ref{decomp1-p0}) 
is evaluated as 
\begin{align}
\label{omega0a*P>[H,a]}
\omega_0^{(\Lambda,N)}(a^\ast P_{\rm high}[H^{(\Lambda)},a])&=
\bigl|\omega_0^{(\Lambda,N)}(a^\ast P_{\rm high}[H^{(\tilde{\Omega})},a])\bigr|\\ \nonumber
&\le \sqrt{\omega_0^{(\Lambda,N)}(a^\ast P_{\rm high}a)
\omega_0^{(\Lambda,N)}([H^{(\tilde{\Omega})},a]^\ast[H^{(\tilde{\Omega})},a])}\\ \nonumber
&\le 2\Vert H^{(\tilde{\Omega})}\Vert \>\Vert a\Vert\sqrt{\omega_0^{(\Lambda,N)}(a^\ast P_{\rm high}a)}. 
\end{align}
{From} the definition of $a=a(T_1)$, we have 
\[
\Vert a\Vert\le 
\frac{1}{\sqrt{2\pi K}}\int_{-T_1}^{+T_1}dt\; \Vert a_0\Vert e^{-t^2/(2K)}\le\Vert a_0\Vert. 
\]
Substituting this into (\ref{omega0a*P>[H,a]}), one has 
\begin{equation}
\label{omega0aastPhighHa}
\omega_0^{(\Lambda,N)}(a^\ast P_{\rm high}[H^{(\Lambda)},a])\le
2\Vert H^{(\tilde{\Omega})}\Vert \>\Vert a_0\Vert\sqrt{\omega_0^{(\Lambda,N)}(a^\ast P_{\rm high}a)}. 
\end{equation}

The first term in the right-hand side of (\ref{decomp1-p0}) is written as 
\begin{align*}
&\omega_0^{(\Lambda,N)}(a^\ast P_{\rm low}[H^{(\Lambda)},a])\\
&=
\frac{1}{q}\sum_{m=1}^q \langle\Phi_{0,m}^{(N)},a^\ast P_{\rm low} [H^{(\Lambda)},a]\Phi_{0,m}^{(N)}\rangle\\
&=\frac{1}{q}\sum_{m=1}^q\sum_{E_n^{(N)}\in[E_1^{(N)},E_1^{(N)}+\varepsilon+\delta E]} 
\langle\Phi_{0,m}^{(N)},a^\ast \Phi_n^{(N)}\rangle \langle \Phi_n^{(N)},a \Phi_{0,m}^{(N)}\rangle
(E_n^{(N)}-E_{0,m}^{(N)}),
\end{align*}
where we have written $\Phi_n^{(N)}$ for the excited-state vectors with the eigenenergies $E_n^{(N)}$. 
Therefore, we obtain 
\begin{equation}
\label{omega0aastPlowHa}
\omega_0^{(\Lambda,N)}(a^\ast P_{\rm low}[H^{(\Lambda)},a])
\le (\Delta E+2\delta E+\varepsilon)\omega_0^{(\Lambda,N)}(a^\ast P_{\rm low} a). 
\end{equation}

Clearly,  the denominator in the expression (\ref{excitbound}) of the excitation energy can be decomposed 
into two parts as 
\[
\omega_0^{(\Lambda,N)}(a^\ast(1-P_0^{(\Lambda,N)})a)
=\omega_0^{(\Lambda,N)}(a^\ast P_{\rm low} a)+\omega_0^{(\Lambda,N)}(a^\ast P_{\rm high}a). 
\]
{From} these observations, it is sufficient to evaluate the first and the second terms 
in this right-hand side.  

First, we will show that the operator $a$ is approximated by the operator $\tilde{a}_0$. 
Since one has 
\[
a-\tilde{a}_0=a-\tilde{a}_0(T_1)+\tilde{a}_0(T_1)-\tilde{a}_0, 
\]
we will show this right-hand side becomes small. Note that 
\begin{align*}
\tilde{a}_0-\tilde{a}_0(T_1)&=\frac{1}{\sqrt{2\pi K}}\int_{T_1}^\infty dt\; e^{itH^{(\Lambda)}}
a_0e^{-itH^{(\Lambda)}}e^{-i(\Delta\tilde{E}+\tilde{\kappa})t}e^{-t^2/(2K)}\\ 
&+\frac{1}{\sqrt{2\pi K}}\int_{-\infty}^{-T_1} dt\; e^{itH^{(\Lambda)}}
a_0e^{-itH^{(\Lambda)}}e^{-i(\Delta\tilde{E}+\tilde{\kappa})t}e^{-t^2/(2K)}.
\end{align*}
Therefore, we have 
\begin{equation}
\label{difftildea0tildea0T1}
\Vert\tilde{a}_0-\tilde{a}_0(T_1)\Vert\le \frac{2\Vert a_0\Vert}{\sqrt{2\pi K}}
\int_{T_1}^\infty e^{-t^2/(2K)}\le \Vert a_0\Vert e^{-T_1^2/(2K)}.
\end{equation}

We write 
\[
a_0^{(\Lambda)}(t):=e^{itH^{(\Lambda)}}a_0e^{-itH^{(\Lambda)}}
\]
and 
\[
a_0^{(\Omega)}(t):=e^{itH^{(\Omega)}}a_0e^{-itH^{(\Omega)}}.
\]
We have 
\[
\tilde{a}_0(T_1)-a
=\frac{1}{\sqrt{2\pi K}}\int_{-T_1}^{+T_1} dt\; 
[a_0^{(\Lambda)}(t)-a_0^{(\Omega)}(t)]
e^{-i(\Delta\tilde{E}+\tilde{\kappa})t}e^{-t^2/(2K)}. 
\]
The norm is estimated as 
\begin{equation}
\label{tildea0T1-a}
\Vert\tilde{a}_0(T_1)-a\Vert \le\frac{1}{\sqrt{2\pi K}}\int_{-T_1}^{+T_1} dt\;  
e^{-t^2/(2K)}\bigl\Vert a_0^{(\Lambda)}(t)-a_0^{(\Omega)}(t)\bigr\Vert.
\end{equation}
The norm in the integrand is estimated as \cite{NOS} 
\begin{equation}
\label{a0tdiff}
\bigl\Vert a_0^{(\Lambda)}(t)-a_0^{(\Omega)}(t)\bigr\Vert
\le {\rm sgn}\>t\int_0^t ds\; \bigl\Vert [(H^{(\Lambda)}-H^{(\Omega)}),a_0^{(\Omega)}(t-s)]\bigl\Vert,
\end{equation}
where 
\[
{\rm sgn}\>t:=\begin{cases} \ \ 1, & \text{$t>0$};\\
\ \ 0, & \text{$t=0$};\\
-1, & \text{$t<0$}.
\end{cases}
\]
The proof of the above inequality (\ref{a0tdiff}) is given in Appendix~\ref{Proofa0tdiffinequal}. 

In terms of local operators $h_X$ with the support $X$, 
the difference between the two Hamiltonians are written as  
\[
H^{(\Lambda)}-H^{(\Omega)}
=\sum_{\substack{X\\ X\cap(\Lambda\backslash\Omega)\ne \emptyset}} h_X
\]
Therefore, we have
\begin{equation} 
\bigl\Vert [(H^{(\Lambda)}-H^{(\Omega)}),a_0^{(\Omega)}(t-s)]\bigl\Vert
\le\sum_{\substack{X\\ X\cap(\Lambda\backslash\Omega)\cap\Omega\ne \emptyset}}
\bigl\Vert [h_X,a_0^{(\Omega)}(t-s)]\bigl\Vert.
\label{diffHamcomma0}
\end{equation}
In order to estimate the summand in the right-hand side, 
we recall the Lieb-Robinson bounds \cite{LiebRobinson,Hastings,HastingsKoma,NachtergaeleSims}: 
Let $A_Y, B_Z$ be a pair of observables with the compact supports, $Y,Z$, respectively. 
Then, the following bound is valid \cite{HastingsKoma}: 
\begin{equation}
\label{LiebRobinson}
\left\Vert[A_Y(t),B_Z]\right\Vert\le {\rm C}\Vert A_Y\Vert \Vert B_Z\Vert |Y||Z|
\exp[-\mu\; {\rm dist}(Y,Z)][e^{v|t|}-1],
\end{equation}
where $A_Y(t)=e^{itH^{(\Lambda)}}A_Y e^{-itH^{(\Lambda)}}$ for 
the Hamiltonian $H^{(\Lambda)}$ with finite-range interactions, and  
the positive constants, $C$, $v$ and $\mu$, depend only on the interactions of the Hamiltonian and 
the metric of the lattice. 
Using the Lieb-Robinson bounds, the right-hand side of (\ref{diffHamcomma0}) can be evaluated as 
\begin{align*}
&\bigl\Vert [(H^{(\Lambda)}-H^{(\Omega)}),a_0^{(\Omega)}(t-s)]\bigl\Vert\\
&\le{\rm Const.}\sum_{\substack{X\\ X\cap(\Lambda\backslash\Omega)\cap\Omega\ne \emptyset}}
|{\rm supp}\> a_0|\>\Vert a_0\Vert (e^{v|t-s|}-1)e^{-\mu r}\\
&\le {\rm Const.} |{\rm supp}\> a_0|\>\Vert a_0\Vert (e^{v|t-s|}-1)Re^{-\mu R},
\end{align*}
where $r={\rm dist}(X, {\rm supp}\> a_0)$, 
$R={\rm dist}({\rm supp}\>(H^{(\Lambda)}-H^{(\Omega)}), {\rm supp}\> a_0)$, 
and the positive constants, $\mu$ and $v$, depend only on the parameters of the present model. 
Combining this, (\ref{tildea0T1-a}) and (\ref{a0tdiff}), we have 
\[
\Vert\tilde{a}_0(T_1)-a\Vert
\le 
{\rm Const.} |{\rm supp}\> a_0|\>\Vert a_0\Vert R\exp[vT_1-\mu R]. 
\]
Combining this with the inequality (\ref{difftildea0tildea0T1}), we obtain 
\begin{align}
\label{diffa-tildea0}
\Vert a-\tilde{a}_0\Vert&\le \Vert a-\tilde{a}_0(T_1)\Vert
+\Vert\tilde{a}_0(T_1)-\tilde{a}_0\Vert\\ \nonumber
&\le{\rm Const.} |{\rm supp}\> a_0|\>\Vert a_0\Vert R\exp[vT_1-\mu R]+\Vert a_0\Vert e^{-T_1^2/(2K)}.
\end{align}
Thus, one can make the difference between $a$ and $\tilde{a}_0$ small 
by appropriately choosing the parameters, $T_1$, $R$ and $K$. 

Next, we prepare some estimates about the operator $\tilde{a}_0$. 
We denote by $\Phi_n^{(N)}$ the excited-state vector with the energy eigenvalue $E_n^{(N)}\ge E_1^{(N)}$ 
for $n\ge 1$. The matrix elements of $\tilde{a}_0$ between the excited states 
and the lowest energy $q$ states are computed as 
\begin{align*}
\langle \Phi_n^{(N)},\tilde{a}_0\Phi_{0,m}^{(N)}\rangle&=\frac{1}{\sqrt{2\pi K}}
\int_{-\infty}^{+\infty}dt\>\langle \Phi_n^{(N)},a_0\Phi_{0,m}^{(N)}\rangle 
e^{i\mathcal{E}_{n,m}^\kappa}e^{-t^2/(2K)}\\
&=\langle \Phi_n^{(N)},a_0\Phi_{0,m}^{(N)}\rangle\exp\bigl[-K(\mathcal{E}_{n,m}^\kappa)^2/2\bigr],
\end{align*}
where we have written 
\[
\mathcal{E}_{n,m}^\kappa:=E_n^{(N)}-E_{0,m}^{(N)}-\Delta\tilde{E}-\tilde{\kappa}.
\]
Using this relations, one has 
\begin{align*}
&\hspace{-1.3cm}\langle\Phi_{0,m}^{(N)},\tilde{a}_0^\ast P_{\rm high}\tilde{a}_0\Phi_{0,m}^{(N)}\rangle\\
&=\sum_{E_n^{(N)}>E_1^{(N)}+\varepsilon+\delta E}
\langle\Phi_{0,m}^{(N)},a_0^\ast \Phi_n^{(N)}\rangle \langle\Phi_n^{(N)},a_0\Phi_{0,m}^{(N)}\rangle
e^{-K(\mathcal{E}_{n,m}^\kappa)^2}\\
&\le \sum_{E_n^{(N)}>E_1^{(N)}+\varepsilon+\delta E}
\langle\Phi_{0,m}^{(N)},a_0^\ast \Phi_n^{(N)}\rangle \langle\Phi_n^{(N)},a_0\Phi_{0,m}^{(N)}\rangle
e^{-K\varepsilon^2/4}\\
&=\langle\Phi_{0,m}^{(N)},a_0^\ast P_{\rm high}a_0\Phi_{0,m}^{(N)}\rangle
e^{-K\varepsilon^2/4}.
\end{align*}
Immediately, 
\begin{equation}
\label{omega0a0Phigha0}
\omega_0^{(\Lambda,N)}(\tilde{a}_0^\ast P_{\rm high}\tilde{a}_0)
\le \omega_0^{(\Lambda,N)}(a_0^\ast P_{\rm high}a_0)e^{-K\varepsilon^2/4}
\le \Vert a_0\Vert^2 e^{-K\varepsilon^2/4}.
\end{equation}
Similarly, we have 
\begin{align*}
&\hspace{-0.5cm}q\times \omega_0^{(\Lambda,N)}(\tilde{a}_0^\ast P_{\rm low}\tilde{a}_0)\\
&=\sum_{m=1}^q \langle\Phi_{0,m}^{(N)},\tilde{a}_0^\ast P_{\rm low}\tilde{a}_0\Phi_{0,m}^{(N)}\rangle\\
&=\sum_{m=1}^q\sum_{E_n^{(N)}\in[E_1^{(N)},E_1^{(N)}+\varepsilon+\delta E]}
\langle\Phi_{0,m}^{(N)},\tilde{a}_0^\ast \Phi_n^{(N)}\rangle 
\langle \Phi_n^{(N)},\tilde{a}_0\Phi_{0,m}^{(N)}\rangle\\
&=\sum_{m=1}^q\sum_{E_n^{(N)}\in[E_1^{(N)},E_1^{(N)}+\varepsilon+\delta E]}
\langle\Phi_{0,m}^{(N)},a_0^\ast \Phi_n^{(N)}\rangle \langle\Phi_n^{(N)},a_0\Phi_{0,m}^{(N)}\rangle
e^{-K(\mathcal{E}_{n,m}^\kappa)^2}\\
&\ge \sum_{E_n^{(N)}\in[E_1^{(N)},E_1^{(N)}+\varepsilon+\delta E]}
\langle\Phi_{0,i}^{(N)},a_0^\ast \Phi_n^{(N)}\rangle \langle\Phi_n^{(N)},a_0\Phi_{0,i}^{(N)}\rangle
e^{-K(\mathcal{E}_{n,i}^\kappa)^2}.
\end{align*}
For  
\[
E_n^{(N)}\in[E_1^{(N)}+(\kappa-1)\varepsilon/(2M),E_1^{(N)}+\kappa\varepsilon/(2M)], 
\]
we have 
\[ 
-\frac{\varepsilon}{4M}\le\mathcal{E}_{n,i}^\kappa\le\frac{\varepsilon}{4M}.
\]
Combining these observations, the inequality (\ref{Pvarepsilon/2mat}) 
and the assumption $M\ge \varepsilon\sqrt{K}/4$, we obtain
\begin{equation}
\label{omega0tildea0Plowtildea0}
\omega_0^{(\Lambda,N)}(\tilde{a}_0^\ast P_{\rm low}\tilde{a}_0)\ge 
\langle\Phi_{0,i}^{(N)},a_0^\ast P_{\varepsilon/2}^\kappa a_0\Phi_{0,i}^{(N)}\rangle \frac{e^{-1}}{q}
\ge \frac{e^{-1}c_0}{qM}.
\end{equation}

We choose the parameters, $T_1$, $K$ and $R$ so that they satisfy 
\begin{equation}
T_1=\frac{1}{\sqrt{2}}K\varepsilon, \quad \text{and}\quad 
\mu R=\frac{1}{4}(2\sqrt{2}v+\varepsilon)K\varepsilon.
\label{T1muR}
\end{equation}
Then, the inequalities, (\ref{diffa-tildea0}) and (\ref{omega0a0Phigha0}), are, respectively, written 
\[
\Vert a-\tilde{a}_0\Vert\le{\rm Const.} |{\rm supp}\> a_0|\>\Vert a_0\Vert Re^{-\tilde{\mu}\varepsilon R}
\]
and
\[
\omega_0^{(\Lambda,N)}(\tilde{a}_0^\ast P_{\rm high}\tilde{a}_0)
\le \Vert a_0\Vert^2 e^{-\tilde{\mu}\varepsilon R}, 
\]
where we have written 
\[
\tilde{\mu}:=\frac{\mu}{2\sqrt{2}v+\varepsilon}.
\]
Using these inequalities, we have 
\begin{align}
\label{omega0aastPhigha}
\omega_0^{(\Lambda,N)}(a^\ast P_{\rm high}a)&= 
\omega_0^{(\Lambda,N)}((a-\tilde{a}_0)^\ast P_{\rm high}a)
+\omega_0^{(\Lambda,N)}(\tilde{a}_0^\ast P_{\rm high}(a-\tilde{a}_0))\\ \nonumber
&+\omega_0^{(\Lambda,N)}(\tilde{a}_0^\ast P_{\rm high}\tilde{a}_0)\\ \nonumber
&\le 2\Vert a_0\Vert \Vert a-\tilde{a}_0\Vert 
+\omega_0^{(\Lambda,N)}(\tilde{a}_0^\ast P_{\rm high}\tilde{a}_0)\\ \nonumber
&\le{\rm Const.} |{\rm supp}\> a_0|\>\Vert a_0\Vert^2 Re^{-\tilde{\mu}\varepsilon R}. 
\end{align}
Similarly, 
\begin{align}
\label{omega0aastPlowa}
&|\omega_0^{(\Lambda,N)}(a^\ast P_{\rm low}a)-\omega_0^{(\Lambda,N)}(\tilde{a}_0 P_{\rm low}\tilde{a}_0)|\\ 
\nonumber
&\le|\omega_0^{(\Lambda,N)}((a-\tilde{a}_0)^\ast P_{\rm low}a)|
+|\omega_0^{(\Lambda,N)}(\tilde{a}_0 P_{\rm low}(a-\tilde{a}_0))|\\ \nonumber
&\le 2\Vert a_0\Vert \Vert a-\tilde{a}_0\Vert\\ \nonumber
&\le {\rm Const.} |{\rm supp}\> a_0|\>\Vert a_0\Vert^2 Re^{-\tilde{\mu}\varepsilon R}. 
\end{align}

Now let us estimate the excitation energy. Using the inequalities, (\ref{omega0aastPhighHa}) and 
(\ref{omega0aastPlowHa}), we obtain 
\begin{align*}
&\frac{\omega_0^{(\Lambda,N)}(a^\ast (1-P_0^{(\Lambda,N)})[H^{(\Lambda)},a])}
{\omega_0^{(\Lambda,N)}(a^\ast(1-P_0^{(\Lambda,N)})a)}\\
&=\frac{\omega_0^{(\Lambda,N)}(a^\ast P_{\rm low}[H^{(\Lambda)},a])
+\omega_0^{(\Lambda,N)}(a^\ast P_{\rm high}[H^{(\Lambda)},a])}
{\omega_0^{(\Lambda,N)}(a^\ast P_{\rm low}a)
+\omega_0^{(\Lambda,N)}(a^\ast P_{\rm high}a)}\\
&\le \frac{(\Delta E+2\delta E+\varepsilon)\omega_0^{(\Lambda,N)}(a^\ast P_{\rm low} a)
+2\Vert H^{(\tilde{\Omega})}\Vert \>\Vert a_0\Vert\sqrt{\omega_0^{(\Lambda,N)}(a^\ast P_{\rm high}a)}}
{\omega_0^{(\Lambda,N)}(a^\ast P_{\rm low}a)
+\omega_0^{(\Lambda,N)}(a^\ast P_{\rm high}a)}.
\end{align*}
We can choose the region $\tilde{\Omega}$ so that 
\[ 
\Vert H^{(\tilde{\Omega})}\Vert\le {\rm Const.}R^2 
\]
for a large $R$. 
Combining these observations with the inequalities (\ref{omega0aastPhigha}) and (\ref{omega0aastPlowa}), 
we obtain
\begin{multline}
\frac{\omega_0^{(\Lambda,N)}(a^\ast (1-P_0^{(\Lambda,N)})[H^{(\Lambda)},a])}
{\omega_0^{(\Lambda,N)}(a^\ast(1-P_0^{(\Lambda,N)})a)}\\
\le 
\frac{(\Delta E+2\delta E+\varepsilon)\omega_0^{(\Lambda,N)}(\tilde{a}_0^\ast P_{\rm low}\tilde{a}_0)
+\mathcal{O}(R^{5/2}e^{-\tilde{\mu}\varepsilon R/2})}
{\omega_0^{(\Lambda,N)}(\tilde{a}_0^\ast P_{\rm low}\tilde{a}_0)
+\mathcal{O}(Re^{-\tilde{\mu}\varepsilon R})}
\label{excitomega0a}
\end{multline}
for a large $R$. In order to estimate the right-hand side, we write 
the inequality (\ref{omega0tildea0Plowtildea0}) as  
\[
\omega_0^{(\Lambda,N)}(\tilde{a}_0^\ast P_{\rm low}\tilde{a}_0)M\ge \frac{e^{-1}c_0}{q}>0.
\]
Combining this bound, the assumption $M<\varepsilon\sqrt{K}/4-1$ and (\ref{T1muR}), 
the right-hand side of (\ref{excitomega0a}) with a large $R$ yields the desired upper bound 
(\ref{excitbound}) for the excitation energy.  

\subsection{Twisted phase dependence of the excitation energy}
\label{subsecTPDEE}

We denote by $H_0^{(\Lambda)}(\phi)$ the unperturbed Hamiltonian of the present system 
with the twisted boundary condition in the first direction with the angle $\phi$. 
We write $H_0^{(\Lambda,N)}(\phi)$ for the Hamiltonian restricted onto the sector of $N$ fermions. 
We denote by $P_0^{(\Lambda,N)}(\phi)$ the spectral projection onto the sector which are spanned by 
the $q$ vectors with the lowest $q$ eigenenergies. 
We denote by $\Phi_{0,m}^{(N)}(\phi)$ the low-energy vector with the eigenenergy $E_{0,m}^{(N)}(\phi)$, 
$m=1,2,\ldots,q$. The expectation value in the low-energy sector is given by  
\[
\omega_0^{(\Lambda,N)}(\cdots;\phi):=\frac{1}{q}{\rm Tr}\>(\cdots)P_0^{(\Lambda,N)}(\phi).
\]
Although we have assumed that the ground state for $\phi=0$ is $q$-fold degenerate, we 
do not necessarily assume that the low-energy sector for $\phi\ne 0$ is degenerate. 
 
Let $a$ be a local observable with a compact support. The excitation energy 
due to the local perturbation $a$ is given by 
\[
\frac{\omega_0^{(\Lambda,N)}\bigl(a^\ast(1-P_0^{(\Lambda,N)}(\phi))[H_0^{(\Lambda)}(\phi),a];\phi\bigr)}
{\omega_0^{(\Lambda,N)}\bigl(a^\ast(1-P_0^{(\Lambda,N)}(\phi))a;\phi\bigr)}.
\]
We want to show that this excitation energy is almost independent of the phase $\phi$ 
for a large lattice $\Lambda$ 
under the assumption of the nonvanishing spectral gap above the low-energy sector.  

The numerator of the excitation energy is written 
\begin{multline*}
\omega_0^{(\Lambda,N)}\bigl(a^\ast(1-P_0^{(\Lambda,N)}(\phi))[H_0^{(\Lambda)}(\phi),a];\phi\bigr)\\
=\omega_0^{(\Lambda,N)}\bigl(a^\ast[H_0^{(\Lambda)}(\phi),a];\phi\bigr)
-\omega_0^{(\Lambda,N)}\bigl(a^\ast P_0^{(\Lambda,N)}(\phi)[H_0^{(\Lambda)}(\phi),a];\phi\bigr)
\end{multline*}
The second term in the right-hand side is written 
\begin{align*}
&\omega_0^{(\Lambda,N)}\bigl(a^\ast P_0^{(\Lambda,N)}(\phi)[H_0^{(\Lambda)}(\phi),a];\phi\bigr)\\
&=\frac{1}{q}\sum_{m=1}^q\sum_{m'=1}^q\langle \Phi_{0,m}^{(N)}(\phi),a^\ast\Phi_{0,m'}^{(N)}(\phi)\rangle
\langle\Phi_{0,m'}^{(N)}(\phi),[H_0^{(\Lambda)},a]\Phi_{0,m}^{(N)}(\phi)\rangle\\
&=\frac{1}{q}\sum_{m,m'}
\langle \Phi_{0,m}^{(N)}(\phi),a^\ast\Phi_{0,m'}^{(N)}(\phi)\rangle
\langle \Phi_{0,m'}^{(N)}(\phi), a\Phi_{0,m}^{(N)}(\phi)\rangle
(E_{0,m'}^{(N)}(\phi)-E_{0,m}^{(N)}(\phi))
\end{align*}
in terms of the low-energy vectors $\Phi_{0,m}^{(N)}(\phi)$. Therefore, we have 
\begin{equation}
\label{omegaH0commudeltaEbound}
\bigl|\omega_0^{(\Lambda,N)}(a^\ast
P_0^{(\Lambda,N)}(\phi)[H_0^{(\Lambda)}(\phi),a];\phi)\bigr|\le 
\delta E(\phi)\>\omega_0^{(\Lambda,N)}\bigl(a^\ast P_0^{(\Lambda,N)}(\phi)a;\phi\bigr),
\end{equation}
where 
\[
\delta E(\phi):=\max_{m,m'}\{|E_{0,m}^{(N)}(\phi)-E_{0,m'}^{(N)}(\phi)|\}. 
\] 

Consider the excited states with the energies which are larger than the maximum value, 
$\max_{1\le i\le q}\{E_{0,i}^{(N)}(\phi)\}$, of the lowest $q$ eigenenergies. 
We denote by $\Phi_n^{(N)}(\phi)$ the excitation vector with the eigenenergy $E_n^{(N)}(\phi)$ 
for $n\ge 1$. We take the subscript $n$ of the eigenenergies $\{E_n^{(N)}(\phi)\}$ 
to satisfy $E_1^{(N)}(\phi)\le E_2^{(N)}(\phi)\le \ldots$.  
We assume that there exists $\phi_0\in[0,2\pi]$ such that  
\begin{equation}
\label{assumpGapphi0}
\min_{\phi\in[0,\phi_0]}\{E_1^{(N)}(\phi)\}
-\max_{\phi\in[0,\phi_0]}\max\{E_{0,1}^{(N)}(\phi),\ldots,E_{0,q}^{(N)}(\phi)\}
\ge \Delta E^{\rm min}>0,
\end{equation}
where $\Delta E^{\rm min}$ is independent of the size $|\Lambda|$ of the lattice. 
That is to say, there is a uniform lower bound for the spectral gap above the low-energy sector. 

The denominator of the excitation energy is written 
\[
\omega_0^{(\Lambda,N)}\bigl(a^\ast(1-P_0^{(\Lambda,N)}(\phi))a;\phi\bigr)
=\omega_0^{(\Lambda,N)}\bigl(a^\ast a;\phi\bigr)
-\omega_0^{(\Lambda,N)}\bigl(a^\ast P_0^{(\Lambda,N)}(\phi)a;\phi\bigr).
\]
The second term in the right-hand side is written 
\begin{equation}
\label{omega0aastP0phia}
\omega_0^{(\Lambda,N)}\bigl(a^\ast P_0^{(\Lambda,N)}(\phi)a;\phi\bigr)
=\frac{1}{q}{\rm Tr}\> a^\ast P_0^{(\Lambda,N)}(\phi)a P_0^{(\Lambda,N)}(\phi).
\end{equation}
Under the above assumption on the spectral gap, the projection $P_0^{(\Lambda,N)}(\phi)$ can be 
written 
\[
P_0^{(\Lambda,N)}(\phi)=\frac{1}{2\pi i}\oint dz \frac{1}{z-H_0^{(\Lambda)}(\phi)}
\]
on the sector of $N$ fermions for $\phi\in(0,\phi_0]$. 
This is also differentiable with respect to $\phi$ as 
\[
\frac{d}{d\phi}P_0^{(\Lambda,N)}(\phi)
=\frac{1}{2\pi i}\oint dz \frac{1}{z-H_0^{(\Lambda)}(\phi)}
\biggl[\frac{d H_0^{(\Lambda)}(\phi)}{d\phi}\biggr]
\frac{1}{z-H_0^{(\Lambda)}(\phi)}
\]
As we showed in Sec.~\ref{sec:TP}, the operator 
\[
B(\phi):=\frac{d H_0^{(\Lambda)}(\phi)}{d\phi}
\]
has a support which is localized near the boundary. 
By differentiating and integrating (\ref{omega0aastP0phia}), one has 
\begin{multline*}
\omega_0^{(\Lambda,N)}\bigl(a^\ast P_0^{(\Lambda,N)}(\phi_0)a;\phi_0\bigr)
-\omega_0^{(\Lambda,N)}\bigl(a^\ast P_0^{(\Lambda,N)}(0)a;0\bigr)\\
=\frac{1}{q}\int_0^{\phi_0} d\phi\frac{d}{d\phi}
{\rm Tr}\> a^\ast P_0^{(\Lambda,N)}(\phi)a P_0^{(\Lambda,N)}(\phi).
\end{multline*}
We will show that the integrand in the right-hand side is small for a large volume. 
In consequence, the difference of the expectation values in 
the left-hand side is almost independent of $\phi_0$ for a large volume. 

Relying on the differentiability of the projection operator, one has  
\begin{align}
\label{diffTraastP0aP0}
&\frac{d}{d\phi}{\rm Tr}\> a^\ast P_0^{(\Lambda,N)}(\phi)a P_0^{(\Lambda,N)}(\phi)\\ \nonumber
&=
{\rm Tr}\> a^\ast \biggl[\frac{d}{d\phi}P_0^{(\Lambda,N)}(\phi)\biggr]a P_0^{(\Lambda,N)}(\phi)
+{\rm Tr}\> a^\ast P_0^{(\Lambda,N)}(\phi)a \frac{d}{d\phi}P_0^{(\Lambda,N)}(\phi)\\ \nonumber
&={\rm Tr}\> a P_0^{(\Lambda,N)}(\phi)a^\ast \frac{d}{d\phi}P_0^{(\Lambda,N)}(\phi)
+{\rm Tr}\> a^\ast P_0^{(\Lambda,N)}(\phi)a \frac{d}{d\phi}P_0^{(\Lambda,N)}(\phi). 
\end{align}
Since the first term in the right-hand side is obtained by interchanging $a^\ast$ and $a$ 
in the second term, we will treat the second term only. The second term is written 
\begin{align*}
&{\rm Tr}\> a^\ast P_0^{(\Lambda,N)}(\phi)a \frac{d}{d\phi}P_0^{(\Lambda,N)}(\phi)\\
&=
\frac{1}{2\pi i}\oint dz\> {\rm Tr}\> a^\ast P_0^{(\Lambda,N)}(\phi)a 
\frac{1}{z-H_0^{(\Lambda)}(\phi)}B(\phi)\frac{1}{z-H_0^{(\Lambda)}(\phi)}\\
&=\sum_{m=1}^q\sum_{n\ge 1}\frac{\langle\Phi_{0,m}^{(N)}(\phi),a^\ast P_0^{(\Lambda,N)}(\phi)a
\Phi_n^{(N)}(\phi)\rangle
\langle\Phi_n^{(N)}(\phi),B(\phi)\Phi_{0,m}^{(N)}(\phi)\rangle}{E_{0,m}^{(N)}(\phi)-E_n^{(N)}(\phi)}
+{\rm c.c.}
\end{align*}
Using the integral identity 
\[
\frac{1}{\Delta \mathcal{E}}=\int_0^\infty ds\> e^{-\Delta\mathcal{E}s},
\]
the first sum in the right-hand side is written 
\begin{align*}
&-\int_0^\infty ds\sum_{m=1}^q\sum_{n\ge 1}\langle\Phi_{0,m}^{(N)}(\phi),a^\ast P_0^{(\Lambda,N)}(\phi)a
\Phi_n^{(N)}(\phi)\rangle\langle\Phi_n^{(N)}(\phi),B(\phi)\Phi_{0,m}^{(N)}(\phi)\rangle\\
&\times \exp[-(E_n^{(N)}(\phi)-E_{0,m}^{(N)}(\phi))s]\\
&=-\int_0^{cr} ds\sum_{m=1}^q\sum_{n\ge 1}\langle\Phi_{0,m}^{(N)}(\phi),a^\ast P_0^{(\Lambda,N)}(\phi)a
\Phi_n^{(N)}(\phi)\rangle\langle\Phi_n^{(N)}(\phi),B(\phi)\Phi_{0,m}^{(N)}(\phi)\rangle\\
&\times \exp[-(E_n^{(N)}(\phi)-E_{0,m}^{(N)}(\phi))s]\\
&-\int_{cr}^\infty ds\sum_{m=1}^q\sum_{n\ge 1}\langle\Phi_{0,m}^{(N)}(\phi),a^\ast P_0^{(\Lambda,N)}(\phi)a
\Phi_n^{(N)}(\phi)\rangle\langle\Phi_n^{(N)}(\phi),B(\phi)\Phi_{0,m}^{(N)}(\phi)\rangle\\
&\times \exp[-(E_n^{(N)}(\phi)-E_{0,m}^{(N)}(\phi))s],
\end{align*}
where $c$ is a positive constant, and $r={\rm dist}({\rm supp}\; a,{\rm supp}\; B(\phi))$. 
The second integral in the right-hand side is estimated as 
\begin{align*}
&\biggl|\int_{cr}^\infty ds\>\sum_{m=1}^q\sum_{n\ge 1}\langle\Phi_{0,m}^{(N)}(\phi),a^\ast P_0^{(\Lambda,N)}(\phi)a
\Phi_n^{(N)}(\phi)\rangle\langle\Phi_n^{(N)}(\phi),B(\phi)\Phi_{0,m}^{(N)}(\phi)\rangle\\
&\times \exp\bigl[-(E_n^{(N)}(\phi)-E_{0,m}^{(N)}(\phi))s\bigr]\biggr|
\le \frac{q}{\Delta E^{\rm min}}\Vert a\Vert^2\Vert B(\phi)\Vert \exp[-\Delta E^{\rm min}cr],
\end{align*}
where we have used the assumption (\ref{assumpGapphi0}) on the spectral gap. 
Since $\Vert B(\phi)\Vert=\mathcal{O}(L^{(2)})$ and 
$r={\rm dist}({\rm supp}\> a,{\rm supp}\> B(\phi))=\mathcal{O}(L^{(1)})$, 
the corresponding contribution is exponentially small in the linear size of the system. 

The first integral in the right-hand side is written 
\begin{equation}
\label{intomega0aastP0aB}
-q\int_0^{cr}ds\>\omega_0^{(\Lambda,N)}(a^\ast P_0^{(\Lambda,N)}(\phi)a\tilde{B}^{(\Lambda)}
(is;\phi);\phi),
\end{equation}
where 
\[
\tilde{B}^{(\Lambda)}(z;\phi):=e^{iH_0^{(\Lambda)}(\phi)z}\tilde{B}(\phi)e^{-iH_0^{(\Lambda)}(\phi)z}
\quad\text{for \ } z\in \co
\]
with
\[
\tilde{B}(\phi):=B(\phi)-P_0^{(\Lambda,N)}(\phi)B(\phi)P_0^{(\Lambda,N)}(\phi).
\]
Since $r=\mathcal{O}(L^{(1)})$, it is sufficient to show that the integrand of (\ref{intomega0aastP0aB}) 
is exponentially small in the size $L^{(1)}$. 
For the integrand of (\ref{intomega0aastP0aB}), we write 
\[
f(z)=\omega_0^{(\Lambda,N)}(a^\ast P_0^{(\Lambda,N)}(\phi)a\tilde{B}^{(\Lambda)}(z;\phi);\phi)
\quad\text{for \ } z\in \co.
\]
Using the contour integral in the complex plane, one has 
\begin{multline}
\label{omega0contour}
\omega_0^{(\Lambda,N)}(a^\ast P_0^{(\Lambda,N)}(\phi)a\tilde{B}^{(\Lambda)}
(is;\phi);\phi)=\frac{1}{2\pi i}\oint\frac{f(z)}{z-is}dz \\
=\frac{1}{2\pi i}\int_{-T_2}^{T_2} \frac{f(t)}{t-is}dt 
+\frac{1}{2\pi i}\int \frac{f(z)}{z-is}dz
\end{multline}
for $s>0$. Here, the first term in the right-hand side of the second equality
is the integral along the real axis from $-T_2$ to $T_2$ with a large positive number $T_2$, 
and the second term is the integral along the semi-circle 
$z=T_2e^{i\theta}$ for $\theta\in[0,\pi]$. 
The second integral is estimated as 
\begin{equation}
\label{intf(z)semicirclebound}
\biggl|\frac{1}{2\pi i}\int \frac{f(z)}{z-is}dz\biggr|
\le \frac{\Vert a\Vert^2\Vert B(\phi)\Vert}{2\Delta E^{\rm min}} 
\frac{1-e^{-T_2\Delta E^{\rm min}}}{T_2\sqrt{1-2s/T_2}}
\end{equation}
for $2s<T_2$. The proof is given in Appendix~\ref{prooff(z)semicirclebound}.
We recall that $\Vert B(\phi)\Vert=\mathcal{O}(L^{(2)})$ 
and that $0<s\le cr$ with $r=\mathcal{O}(L^{(1)})$. 
Therefore, the corresponding contribution is vanishing in the limit $T_2\rightarrow\infty$.

In order to evaluate the first integral in the right-hand side of the second equality 
of (\ref{omega0contour}), we use the technique in \cite{Hastings,HastingsKoma,NachtergaeleSims}. 
Following them, we write 
\[
f(t)e^{\alpha s^2}=f(t)e^{-\alpha t^2}+f(t)(e^{\alpha s^2}-e^{-\alpha t^2}).
\]
By definition, 
\[
f(t)=\omega_0^{(\Lambda,N)}(a^\ast P_0^{(\Lambda,N)}(\phi)a \tilde{B}^{(\Lambda)}(t;\phi);\phi).
\]
Using the identity, 
\[
a\tilde{B}^{(\Lambda)}(t;\phi)=\tilde{B}^{(\Lambda)}(t;\phi)a+[a,\tilde{B}^{(\Lambda)}(t;\phi)],
\]
one has 
\begin{align*}
f(t)&=\omega_0^{(\Lambda,N)}(a^\ast P_0^{(\Lambda,N)}(\phi)\tilde{B}^{(\Lambda)}(t;\phi)a;\phi)\\
&+\omega_0^{(\Lambda,N)}(a^\ast P_0^{(\Lambda,N)}(\phi)
[a,\tilde{B}^{(\Lambda)}(t;\phi)];\phi)\\
&=\omega_0^{(\Lambda,N)}(\tilde{B}^{(\Lambda)}(t;\phi)a P_0^{(\Lambda,N)}(\phi)a^\ast;\phi)\\
&+\omega_0^{(\Lambda,N)}(a^\ast P_0^{(\Lambda,N)}(\phi)
[a,B^{(\Lambda)}(t;\phi)];\phi)\\
&-\omega_0^{(\Lambda,N)}(a^\ast P_0^{(\Lambda,N)}(\phi)
[a,P_0^{(\Lambda,N)}(\phi)B^{(\Lambda)}(t;\phi)P_0^{(\Lambda,N)}(\phi)];\phi),
\end{align*}
where 
\[
B^{(\Lambda)}(t;\phi):=e^{iH_0^{(\Lambda)}t}B(\phi)e^{-iH_0^{(\Lambda)}t}
\]
for $t\in \re$. Consider the third term in the right-hand side of the second equality. 
Except for the factor $q^{-1}$, the term is written 
\begin{align*}
&{\rm Tr}\> a^\ast P_0^{(\Lambda,N)}(\phi)
[a,P_0^{(\Lambda,N)}(\phi)B^{(\Lambda)}(t;\phi)P_0^{(\Lambda,N)}(\phi)]P_0^{(\Lambda,N)}(\phi)\\ 
&={\rm Tr}\> a^\ast P_0^{(\Lambda,N)}(\phi)aP_0^{(\Lambda,N)}(\phi)B^{(\Lambda)}(t;\phi)P_0^{(\Lambda,N)}(\phi)\\
&-{\rm Tr}\> a^\ast P_0^{(\Lambda,N)}(\phi)B^{(\Lambda)}(t;\phi)P_0^{(\Lambda,N)}(\phi)aP_0^{(\Lambda,N)}(\phi)\\
&={\rm Tr}\> a^\ast P_0^{(\Lambda,N)}(\phi)aP_0^{(\Lambda,N)}(\phi)B^{(\Lambda)}(t;\phi)P_0^{(\Lambda,N)}(\phi)\\
&-{\rm Tr}\> aP_0^{(\Lambda,N)}(\phi)a^\ast P_0^{(\Lambda,N)}(\phi)B^{(\Lambda)}(t;\phi)P_0^{(\Lambda,N)}(\phi).
\end{align*}
Clearly, by interchanging $a^\ast$ and $a$, this contribution changes the sign.   
Therefore, the contribution is canceled by the corresponding contribution coming from 
the first term in the right-hand side of the second equality of (\ref{diffTraastP0aP0}) 
because the first term is obtained by interchanging $a^\ast$ and $a$ in the second term in (\ref{diffTraastP0aP0}).  

{From} these observations,   
in order to evaluate the first integral in the right-hand side of the second equality 
of (\ref{diffTraastP0aP0}), it is sufficient to estimate the following three integrals: 
\begin{equation}
\label{I1}
I_1:=\frac{1}{2\pi i}\int_{-T_2}^{T_2}dt\>\frac{1}{t-is}
\omega_0^{(\Lambda,N)}(\tilde{B}^{(\Lambda)}(t;\phi)a P_0^{(\Lambda,N)}(\phi)a^\ast;\phi)e^{-\alpha t^2}
\end{equation}
\begin{equation}
\label{I2}
I_2:=\frac{1}{2\pi i}\int_{-T_2}^{T_2}dt\>\frac{1}{t-is}
\omega_0^{(\Lambda,N)}(a^\ast P_0^{(\Lambda,N)}(\phi)[a,B^{(\Lambda)}(t;\phi)];\phi)e^{-\alpha t^2}
\end{equation}
\begin{equation}
\label{I3}
I_3:=\frac{1}{2\pi i}\int_{-T_2}^{T_2}dt\>
\frac{f(t)}{t-is}(e^{\alpha s^2}-e^{-\alpha t^2})
\end{equation}
The first and second integrals can be estimated as follows: 
\begin{equation}
\label{I1bound}
|I_1|\le \frac{\Vert a\Vert^2\Vert B(\phi)\Vert}{2}\exp[-(\Delta E^{\rm min})^2/(4\alpha)]
\end{equation}
in the limit $T_2\rightarrow\infty$, and 
\begin{equation}
\label{I2bound}
|I_2|\le\Vert a\Vert^2\Vert B(\phi)\Vert\biggl(\frac{1}{\sqrt{\pi\alpha}}\frac{1}{S}e^{-\alpha S^2}
+{\rm Const.}|{\rm supp}\> a|e^{vS-\mu r}\biggr).
\end{equation}
Here, $S$ is a positive number; $v$ and $\mu$ are positive constants 
which are determined by the model's parameters.  
The proofs of (\ref{I1bound}) and (\ref{I2bound}) are given in Appendices~\ref{estI1} and \ref{estI2}, respectively. 
As to the third integral with $0<s\le cr$, we choose $\alpha=\Delta E^{\rm min}/(2cr)$. Then, one has  
\begin{equation}
\label{I3bound}
|I_3|\le\frac{1}{2}\Vert a\Vert^2\Vert B(\phi)\Vert e^{-\Delta E^{\rm min}cr/2}
\end{equation}
in the limit $T_2\rightarrow\infty$. 
The proof is given in Appendix~\ref{estI3}.  

For the same $\alpha=\Delta E^{\rm min}/(2cr)$, 
the upper bound (\ref{I1bound}) for $|I_1|$ becomes the same as that for $|I_3|$ as   
\[ 
|I_1|\le\frac{1}{2}\Vert a\Vert^2\Vert B(\phi)\Vert e^{-\Delta E^{\rm min}cr/2}. 
\]

As to the upper bound (\ref{I2bound}) for $|I_2|$, we choose $S=cr$ and 
\[
c=\frac{\mu}{v+\Delta E^{\rm min}/2}
\]
with the same $\alpha=\Delta E^{\rm min}/(2cr)$ as in the above. Then, we have 
\[
|I_2|\le\Vert a\Vert^2\Vert B(\phi)\Vert\biggl(\frac{1}{\sqrt{\pi\tilde{\mu}r}}
+{\rm Const.}|{\rm supp}\> a|\biggr)e^{-\tilde{\mu}r}, 
\]
where 
\[
\tilde{\mu}=\frac{\Delta E^{\rm min}}{2v+\Delta E^{\rm min}}\mu.
\]
Since $\Delta E^{\rm min}c/2=\tilde{\mu}$, we have 
\[
|I_1|+|I_3|\le\Vert a\Vert^2\Vert B(\phi)\Vert e^{-\tilde{\mu}r}.
\] 
We recall $\Vert B(\phi)\Vert =\mathcal{O}(L^{(2)})$ and 
$r={\rm dist}({\rm supp}\; a,{\rm supp}\; B(\phi))=\mathcal{O}(L^{(1)})$. 
Therefore, all of the above integrals are exponentially small in the linear size $L^{(1)}$ of the system.  
As a result, we obtain 
\begin{multline}
\left|\omega_0^{(\Lambda,N)}\bigl(a^\ast P_0^{(\Lambda,N)}(\phi_0)a;\phi_0\bigr)
-\omega_0^{(\Lambda,N)}\bigl(a^\ast P_0^{(\Lambda,N)}(0)a;0\bigr)\right|\\
\le{\rm Const.}|{\rm supp}\> a|\Vert a\Vert^2\Vert B(\phi)\Vert e^{-\tilde{\mu}r}. 
\end{multline}
Clearly, we can write 
\begin{equation}
\label{difomega1}
\omega_0^{(\Lambda,N)}\bigl(a^\ast P_0^{(\Lambda,N)}(\phi_0)a;\phi_0\bigr)
=\omega_0^{(\Lambda,N)}\bigl(a^\ast P_0^{(\Lambda,N)}(0)a;0\bigr)+\epsilon_1(\Lambda),
\end{equation}
where $\epsilon_1(\Lambda)$ is  the exponentially small correction. 

Since only the hopping terms near the boundaries in the Hamiltonian $H_0^{(\Lambda)}(\phi)$   
depend on the phase $\phi$, 
the commutator $[H_0^{(\Lambda)}(\phi),a]$ is independent of the phase $\phi$ 
and becomes a local operator with a compact support 
which is independent of the lattice $\Lambda$. Therefore, we obtain 
\begin{equation}
\label{difomega2}
\omega_0^{(\Lambda,N)}\bigl(a^\ast [H_0^{(\Lambda)}(\phi_0),a];\phi_0\bigr)
=\omega_0^{(\Lambda,N)}\bigl(a^\ast [H_0^{(\Lambda)}(0),a];0\bigr)+\epsilon_2(\Lambda)
\end{equation}
with the small correction $\epsilon_2(\Lambda)$ in the same way. Furthermore, we have 
\begin{equation}
\label{difomega3}
\omega_0^{(\Lambda,N)}\bigl(a^\ast a;\phi_0\bigr)
=\omega_0^{(\Lambda,N)}\bigl(a^\ast a;0\bigr)+\epsilon_3(\Lambda)
\end{equation}
with the small correction $\epsilon_3(\Lambda)$. 

We define 
\[
\delta E^{\rm max}:=\max_{\phi\in[0,\phi_0]}\{\delta E(\phi)\}, 
\]
and 
\[
\Delta E(\phi):=E_1^{(N)}(\phi)-\max\{E_{0,1}^{(N)}(\phi),\ldots,E_{0,q}^{(N)}(\phi)\}.
\]
We assume that there exists a small $\varepsilon>0$ such that 
\begin{equation}
\Delta E^{\rm min}-\delta E^{\rm max}\ge \varepsilon.
\end{equation}

Let us consider the Hamiltonian $H_0^{(\Lambda)}(\phi_0)$ with a fixed $\phi_0$. 
We require the same assumption as in Proposition~\ref{prop:expexcit}. 
For the same $\varepsilon$ as in the above, we can find a local operator $a$ which satisfies 
the excitation energy bound 
\begin{equation}
\label{upperbound excitphi0}
\frac{\omega_0^{(\Lambda,N)}\bigl(a^\ast(1-P_0^{(\Lambda,N)}(\phi_0))[H_0^{(\Lambda)}(\phi_0),a];\phi_0\bigr)}
{\omega_0^{(\Lambda,N)}\bigl(a^\ast(1-P_0^{(\Lambda,N)}(\phi_0))a;\phi_0\bigr)}
\le \Delta E(\phi_0)+2\delta E(\phi_0)+2\varepsilon
\end{equation}
as in Proposition~\ref{prop:expexcit}. We write 
\[
\epsilon_4(R):=\omega_0^{(\Lambda,N)}\bigl(a^\ast P_0^{(\Lambda,N)}(\phi_0)[H_0^{(\Lambda)}(\phi_0),a];\phi_0\bigr)
\]
for short. The same argument in the subsection~\ref{ConstLEE} yields
\begin{equation}
\label{omega0aastP0phi0aO}
\omega_0^{(\Lambda,N)}\bigl(a^\ast P_0^{(\Lambda,N)}(\phi_0)a;\phi_0\bigr)
=\mathcal{O}(Re^{-\tilde{\mu}\varepsilon R}). 
\end{equation}
Combining this with the bound (\ref{omegaH0commudeltaEbound}), we have  
\[
\left|\epsilon_4(R)\right|
\le \delta E(\phi_0)\omega_0^{(\Lambda,N)}\bigl(a^\ast P_0^{(\Lambda,N)}(\phi_0)a;\phi_0\bigr)
=\mathcal{O}(Re^{-\tilde{\mu}\varepsilon R}). 
\]
{From} (\ref{difomega1}), (\ref{difomega2}) and (\ref{difomega3}), the excitation energy 
for the Hamiltonian $H_0^{(\Lambda)}(\phi_0)$ is written 
\begin{align}
\label{exciteneomega0phi0}
&\frac{\omega_0^{(\Lambda,N)}\bigl(a^\ast(1-P_0^{(\Lambda,N)}(\phi_0))[H_0^{(\Lambda)}(\phi_0),a];\phi_0\bigr)}
{\omega_0^{(\Lambda,N)}\bigl(a^\ast(1-P_0^{(\Lambda,N)}(\phi_0))a;\phi_0\bigr)}\\ \nonumber
&=\frac{\omega_0^{(\Lambda,N)}\bigl(a^\ast[H_0^{(\Lambda)}(\phi_0),a];\phi_0\bigr)
-\omega_0^{(\Lambda,N)}\bigl(a^\ast P_0^{(\Lambda,N)}(\phi_0)[H_0^{(\Lambda)}(\phi_0),a];\phi_0\bigr)}
{\omega_0^{(\Lambda,N)}\bigl(a^\ast a;\phi_0\bigr)
-\omega_0^{(\Lambda,N)}\bigl(a^\ast P_0^{(\Lambda,N)}(\phi_0)a;\phi_0\bigr)}\\ \nonumber
&=\frac{\omega_0^{(\Lambda,N)}\bigl(a^\ast[H_0^{(\Lambda)}(0),a];0\bigr)+\epsilon_2(\Lambda)
-\epsilon_4(R)}
{\omega_0^{(\Lambda,N)}\bigl(a^\ast a;0\bigr)+\epsilon_3(\Lambda)
-\omega_0^{(\Lambda,N)}\bigl(a^\ast P_0^{(\Lambda,N)}(0)a;0\bigr)-\epsilon_1(\Lambda)}.
\end{align}
We assume that the ground state of the Hamiltonian $H_0^{(\Lambda,N)}(0)$ with $\phi=0$ 
is (quasi)degenerate, i.e.,  
\[
\delta E(0)\rightarrow 0 \ \mbox{as\ } |\Lambda|\rightarrow \infty.
\]
Under this assumption, from (\ref{omegaH0commudeltaEbound}), we have 
\[
\omega_0^{(\Lambda,N)}\bigl(a^\ast P_0^{(\Lambda,N)}(0)[H_0^{(\Lambda)}(0),a];0\bigr)
\rightarrow 0 \ \mbox{as\ } |\Lambda|\rightarrow \infty.
\]
Further, from the same argument in the subsection~\ref{ConstLEE},  
(\ref{difomega1}), (\ref{difomega3}) and (\ref{omega0aastP0phi0aO}), one can show that 
there exist a sufficiently large $|\Lambda|$, a sufficiently large $R$ and a positive constant $c_0'$ such that  
\[
\omega_0^{(\Lambda,N)}\bigl(a^\ast(1- P_0^{(\Lambda,N)}(0))a;0\bigr)\ge \frac{c_0'}{M}, 
\]
where the positive number $M$ is defined in the subsection~\ref{ConstLEE}. From these observations, 
the excitation energy of (\ref{exciteneomega0phi0}) can be estimated as
\begin{align*}
&\frac{\omega_0^{(\Lambda,N)}\bigl(a^\ast(1-P_0^{(\Lambda,N)}(\phi_0))[H_0^{(\Lambda)}(\phi_0),a];\phi_0\bigr)}
{\omega_0^{(\Lambda,N)}\bigl(a^\ast(1-P_0^{(\Lambda,N)}(\phi_0))a;\phi_0\bigr)}\\
&=\frac{\omega_0^{(\Lambda,N)}\bigl(a^\ast(1-P_0^{(\Lambda,N)}(0))[H_0^{(\Lambda)}(0),a];0\bigr)}
{\omega_0^{(\Lambda,N)}\bigl(a^\ast(1-P_0^{(\Lambda,N)}(0))a;0\bigr)}-\varepsilon\\
&\ge \Delta E(0)-\varepsilon
\end{align*}
for a sufficiently large $|\Lambda|$ and a sufficiently large $R$. 
Combining this with the upper bound (\ref{upperbound excitphi0}) of the excitation energy,  
we obtain
\begin{equation}
\label{Delta Ephi0Lbound}
\Delta E(\phi_0)\ge \Delta E(0)-2\delta E(\phi_0)-3\varepsilon. 
\end{equation}
Thus, if $\Delta E(0)>2\delta E(\phi_0)$ in the infinite-volume limit,  
then the excitation energy for $\phi_0$ is strictly positive, i.e., $\Delta E(\phi_0)>0$.  
Further, if the splitting $\delta E(\phi_0)$ of the energies in the low-energy sector 
is vanishing in the infinite-volume limit, then we have the bound $\Delta E(\phi_0)\ge \Delta E(0)$.  

\subsection{Twisted phase dependence of the averaged energy}
\label{subSec:TPDavE}

Consider the same situation as in the preceding Section~\ref{subsecTPDEE}. 
But, instead of the energy gap condition (\ref{assumpGapphi0}), we consider more generic condition, 
\begin{equation}
\label{gapassumpphi1phi2}
\min_{\phi\in[\phi_1,\phi_2]}\{E_1^{(N)}(\phi)\}
-\max_{\phi\in[\phi_1,\phi_2]}\max\{E_{0,1}^{(N)}(\phi),\ldots,E_{0,q}^{(N)}(\phi)\}
\ge \Delta E^{\rm min}>0,
\end{equation}
for the interval $[\phi_1,\phi_2]$ of the phase $\phi$. Here, we assume that 
the positive constant $\Delta E^{\rm min}$ is independent of the size $|\Lambda|$ of the lattice. 

We decompose the Hamiltonian $H_0^{(\Lambda)}(\phi)$ into two parts as  
$$
H_0^{(\Lambda)}(\phi)=H_0^{(\Lambda^+)}(\phi)+H_0^{(\Lambda^-)}(\phi),
$$
where the supports of the operators $H_0^{(\Lambda^\pm)}(\phi)$ are included 
in the regions, 
$$
\Lambda^+=\left([-L^{(1)}/2,-L^{(1)}/2+R_0]\cup[-R_0,L^{(1)}/2]\right)\times[-L^{(2)}/2,L^{(2)}/2] 
$$
and 
$$
\Lambda^-=\left([-L^{(1)}/2,R_0]\cup[-R_0+L^{(1)}/2,L^{(1)}/2]\right)\times[-L^{(2)}/2,L^{(2)}/2], 
$$
respectively. Here, $R_0$ is a positive constant which is independent of the size $|\Lambda|$ of 
the lattice. Clearly, we have 
\begin{equation}
\label{omega0H0decompo}
\omega_0^{(\Lambda,N)}\big(H_0^{(\Lambda)}(\phi);\phi\big)=
\omega_0^{(\Lambda,N)}\big(H_0^{(\Lambda^+)}(\phi);\phi\big)+
\omega_0^{(\Lambda,N)}\big(H_0^{(\Lambda^-)}(\phi);\phi\big).
\end{equation}
Consider the first term in the right-hand side. 
As shown in Sec.~\ref{sec:TP}, one can change the position of the twisted hopping amplitudes 
of the local Hamiltonians in the total Hamiltonian by using the unitary transformation. 
Therefore, we can find the unitary operator $U_-$ such that in the transformed Hamiltonian 
$$
H_{0,-}^{(\Lambda)}(\phi):=U_-^\dagger H_0^{(\Lambda)}(\phi)U_-,
$$
the local Hamiltonians with the twisted hopping amplitudes 
are arranged along the center line with the first coordinate $x^{(1)}\approx -L^{(1)}/4$ 
in the region $\Lambda^-$. We write 
$$
H_{0,-}^{(\Lambda^+)}:=U_-^\dagger H_0^{(\Lambda^+)}(\phi) U_-.
$$
By definition, this Hamiltonian $H_{0,-}^{(\Lambda^+)}$ is independent of the phase $\phi$. 
Further, one has 
\begin{equation}
\label{distH0-B-}
{\rm dist}\big({\rm supp}\; H_{0,-}^{(\Lambda^+)},{\rm supp}\; B_-(\phi)\big)=\mathcal{O}(L^{(1)}),
\end{equation}
where the operator 
$$
B_-(\phi):=\frac{d}{d\phi}H_{0,-}^{(\Lambda)}(\phi)
$$
comes from the twisted hopping terms in the Hamiltonian. 
In other words, we choose the unitary transformation $U_-$ so that 
the condition (\ref{distH0-B-}) is satisfied. 

The first term in the right-hand side of (\ref{omega0H0decompo}) is written 
\begin{align*}
\omega_0^{(\Lambda,N)}\big(H_0^{(\Lambda^+)}(\phi);\phi\big)&=
\frac{1}{q}{\rm Tr}\; H_{0,-}^{(\Lambda^+)}P_0^{(\Lambda,N)}(\phi)\\
&=\frac{1}{q}\frac{1}{2\pi i}\oint dz\;{\rm Tr}\;H_{0,-}^{(\Lambda^+)}\frac{1}{z-H_{0,-}^{(\Lambda)}(\phi)}.
\end{align*}
Relying on the assumption (\ref{gapassumpphi1phi2}) on the spectral gap, one has 
\begin{multline*}
\frac{d}{d\phi}\omega_0^{(\Lambda,N)}\big(H_0^{(\Lambda^+)}(\phi);\phi\big)\\
=\frac{1}{q}\frac{1}{2\pi i}\oint dz\;{\rm Tr}\;H_{0,-}^{(\Lambda^+)}
\frac{1}{z-H_{0,-}^{(\Lambda)}(\phi)}B_-(\phi)\frac{1}{z-H_{0,-}^{(\Lambda)}(\phi)}.
\end{multline*}
By integrating the left-hand side, one obtains
\begin{multline*}
\omega_0^{(\Lambda,N)}\big(H_0^{(\Lambda^+)}(\phi_2);\phi_2\big)
-\omega_0^{(\Lambda,N)}\big(H_0^{(\Lambda^+)}(\phi_1);\phi_1\big)\\
=\int_{\phi_1}^{\phi_2}d\phi\;
\frac{d}{d\phi}\omega_0^{(\Lambda,N)}\big(H_0^{(\Lambda^+)}(\phi);\phi\big).
\end{multline*}
Combining these with (\ref{distH0-B-}), in the same way as in the preceding Section~\ref{subsecTPDEE}, 
we can prove that this left-hand side is exponentially small in the linear size $L^{(1)}$ of the system. 
Since the second term in the right-hand side of (\ref{omega0H0decompo}) can be handled in the same way, 
we can prove that the difference 
$$
\omega_0^{(\Lambda,N)}(H_0^{(\Lambda)}(\phi_2);\phi_2)
-\omega_0^{(\Lambda,N)}(H_0^{(\Lambda)}(\phi_1);\phi_1)
$$
is exponentially small in the linear size $L^{(1)}$ of the system. 

Next consider the situation that the degeneracy of the energies in the low-energy sector is lifted. 
Namely, we consider the situation that there is a spectral gap between their energies 
in addition to the above assumption (\ref{gapassumpphi1phi2}) on the spectral gap.   

Let $m_1,m_2$ be integers satisfying $1\le m_1<m_2\le q$. We assume the existence of the nonvanishing 
spectral gaps as  
\begin{multline*}
\min_{\phi\in[\phi_1,\phi_2]}\min\{E_{0,m_1+1}^{(N)}(\phi),\ldots,E_{0,m_2}^{(N)}(\phi)\}\\
-\max_{\phi\in[\phi_1,\phi_2]}\max\{E_{0,1}^{(N)}(\phi),\ldots,E_{0,m_1}^{(N)}(\phi)\}
\ge \Delta E_-^{\rm min}>0
\end{multline*}
and 
\begin{multline*}
\min_{\phi\in[\phi_1,\phi_2]}\min\{E_{0,m_2+1}^{(N)}(\phi),\ldots,E_{0,q}^{(N)}(\phi)\}\\
-\max_{\phi\in[\phi_1,\phi_2]}\max\{E_{0,m_1+1}^{(N)}(\phi),\ldots,E_{0,m_2}^{(N)}(\phi)\}
\ge \Delta E_+^{\rm min}>0,
\end{multline*}
where $\Delta E_-^{\rm min}$ and $\Delta E_+^{\rm min}$ are independent of the size $|\Lambda|$ of 
the lattice. When $m_2=q$, the second condition is replaced with the condition (\ref{gapassumpphi1phi2}). 
Let $P_{0,-}^{(\Lambda,N)}(\phi)$ be the projection onto the sector spanned by the eigenvectors 
with the eigenenergies $\{E_{0,1}^{(N)}(\phi),\ldots,E_{0,m_1}^{(N)}(\phi)\}$, 
and $P_{0,+}^{(\Lambda,N)}(\phi)$ the projection onto that with 
$\{E_{0,m_1+1}^{(N)}(\phi),\ldots,E_{0,m_2}^{(N)}(\phi)\}$. 
We write 
$$
\omega_{0,-}^{(\Lambda,N)}(\cdots;\phi)=\frac{1}{m_1}{\rm Tr}\;(\cdots)P_{0,-}^{(\Lambda,N)}(\phi), 
$$
$$
\omega_{0,+}^{(\Lambda,N)}(\cdots;\phi)=\frac{1}{m_2-m_1}{\rm Tr}\;(\cdots)P_{0,+}^{(\Lambda,N)}(\phi) 
$$
and 
$$
\omega_{0,{\rm tot}}^{(\Lambda,N)}(\cdots;\phi)
=\frac{1}{m_2}{\rm Tr}\;(\cdots)\big[P_{0,+}^{(\Lambda,N)}(\phi)+P_{0,-}^{(\Lambda,N)}(\phi)\big]. 
$$
Then, one has 
\begin{multline*}
\omega_{0,+}^{(\Lambda,N)}(H_0^{(\Lambda)}(\phi);\phi)\\
=\frac{m_2}{m_2-m_1}\omega_{0,{\rm tot}}^{(\Lambda,N)}(H_0^{(\Lambda)}(\phi);\phi)
-\frac{m_1}{m_2-m_1}\omega_{0,-}^{(\Lambda,N)}(H_0^{(\Lambda)}(\phi);\phi).
\end{multline*}
Clearly, 
\begin{align*}
&\omega_{0,+}^{(\Lambda,N)}(H_0^{(\Lambda)}(\phi_2);\phi_2)
-\omega_{0,+}^{(\Lambda,N)}(H_0^{(\Lambda)}(\phi_1);\phi_1)\\
&=\frac{m_2}{m_2-m_1}\big[\omega_{0,{\rm tot}}^{(\Lambda,N)}(H_0^{(\Lambda)}(\phi_2);\phi_2)
-\omega_{0,{\rm tot}}^{(\Lambda,N)}(H_0^{(\Lambda)}(\phi_1);\phi_1)\big]\\
&-\frac{m_1}{m_2-m_1}\big[\omega_{0,-}^{(\Lambda,N)}(H_0^{(\Lambda)}(\phi_2);\phi_2)
-\omega_{0,-}^{(\Lambda,N)}(H_0^{(\Lambda)}(\phi_1);\phi_1)\big]. 
\end{align*}
{From} the above argument, this right-hand side is exponentially small in the linear size $L^{(1)}$ of 
the system. Thus, in the sector of isolated eigenenergies from the rest of the spectrum with finite energy gaps, 
the twisted phase dependence of the arithmetic mean of the energies in the sector is exponentially small 
in the linear size $L^{(1)}$.

\subsection{Degeneracy of the sector of the ground state} 

By assuming that the degeneracy of the eigenenergies in the sector of the ground state is lifted 
when changing the value of the twisted phase $\phi$, we will deduce a contradiction. 
In consequence, we prove that the degeneracy of the sector of the ground state is not lifted 
for any value of the twisted phase $\phi\in[0,2\pi)$. 
In other words, the sector of the ground state is (quasi)degenerate irrespective of the twisted phase $\phi$. 
Clearly, by combining this with the result (\ref{Delta Ephi0Lbound}) in Sec.~\ref{subsecTPDEE}, 
the proof of Theorem~\ref{thm:twistDegeneGap} is completed. 

Consider again the same situation as in Sec.~\ref{subsecTPDEE}. 
To begin with, we note that all of the eigenenergies are a continuous function of the twisted phase $\phi$. 
Relying on this fact, we can assume that 
there exist a sequence of lattices, $\Lambda_1\subset\Lambda_2\subset\cdots$, 
and a sufficiently small phase $\phi'\in(0,2\pi)$ such that for all the lattices $\Lambda_n$, 
the phase $\phi'$ satisfies the following two conditions:
$$
\delta E(\phi')=\max_{\phi\in(0,\phi']}\{\delta E(\phi)\}
$$
and 
$$
0<\delta E^{\rm min}\le \delta E(\phi')\le \frac{1}{4}\Delta E(0).
$$
Here, the phase $\phi'$ may depend on the sizes $|\Lambda_n|$ of the lattices $\Lambda_n$, 
and the positive constant $\delta E^{\rm min}$ is independent of the sizes $|\Lambda_n|$ of the lattices. 
Clearly, the second condition implies that the degeneracy of the sector of the ground state is lifted at 
the twisted phase $\phi=\phi'$.

First, fix $\phi=\phi'$. 
Then, there exists a subsequence of lattices $\{\Lambda_n^{(0)}\}_n\subset\{\Lambda_n\}_n$ such that 
the sector of the $q$-fold ground state splits into $\ell$ sectors each of which is $q_i$-fold degenerate 
for $i=1,2,\ldots,\ell$, and that there exist a nonvanishing spectral gap $\gamma_m$ between two adjacent sectors 
for $m=1,2,\ldots,\ell-1$. Clearly, the degeneracy and the gaps must satisfy 
$$
\sum_{i=1}^\ell q_i=q \quad \mbox{and} \quad \sum_{m=1}^{\ell-1}\gamma_m\le \frac{1}{4}\Delta E(0),
$$
respectively. 
We denote by $P_{0,i}^{(\Lambda,N)}(\phi)$ the projection onto the $i$-th sector, 
and define    
$$
\omega_{0,i}^{(\Lambda,N)}(\cdots;\phi):=\frac{1}{q_i}{\rm Tr}\;(\cdots)P_{0,i}^{(\Lambda,N)}(\phi),
\quad i=1,2,\ldots,\ell.
$$
Then, the above conditions about the spectral gaps are written 
\begin{equation}
\label{avEgap}
\omega_{0,i+1}^{(\Lambda,N)}(H_0^{(\Lambda)}(\phi');\phi')
-\omega_{0,i}^{(\Lambda,N)}(H_0^{(\Lambda)}(\phi');\phi')\ge\gamma_i>0,
\quad i=1,2,\ldots,\ell-1.
\end{equation}

Consider changing the twisted phase $\phi$ from $\phi'$ to $\phi_1'\in[0,\phi')$. Then, 
the deviation of the energy from the averaged energy for the $i$-th sector is given by  
\begin{multline*}
\Delta E_{0,i}(\phi_1'):=\max_{1\le j\le q_i}
\big\{\big|E_{0,m_{i-1}+j}^{(N)}(\phi_1')-\omega_{0,i}^{(\Lambda,N)}(H_0^{(\Lambda)}(\phi_1');\phi_1')\big|\big\},\\
\quad i=1,2,\ldots,\ell,
\end{multline*}
where $m_0=0$ and 
$$
 m_{i-1}=\sum_{j=1}^{i-1} q_j \quad \mbox{for\ } i=2,3,\ldots,\ell.
$$
We can find $\phi_1'\in[0,\phi')$ so that the spectral gaps between two adjacent sectors 
remain open when the twisted phase $\phi$ continuously varies from $\phi'$ to $\phi_1'\in[0,\phi')$.  
Further, for a sufficiently small deviation $\phi'-\phi_1'$ of the twisted phase, 
we can assume that the deviation of the energy in the $i$-th sector satisfies  
\begin{equation}
\label{deviEboundgamma}
\Delta E_{0,i}(\phi_1')\le \frac{1}{4}\gamma,
\end{equation}
where $\gamma:=\min\{\gamma_1,\gamma_2,\ldots,\gamma_{\ell-1}\}$. 
Unless these spectral gaps close for varying the twisted phase, 
the deviation of the averaged energy is exponentially small in the linear size $L^{(1)}$ of the lattice as 
$$
\omega_{0,i}^{(\Lambda,N)}(H_0^{(\Lambda)}(\phi_1');\phi_1')
-\omega_{0,i}^{(\Lambda,N)}(H_0^{(\Lambda)}(\phi');\phi')=\mathcal{O}(\exp[-{\rm Const.}L^{(1)}])
$$
which is the result in the preceding Sec.~\ref{subSec:TPDavE}. 
Therefore, by combining this with the bounds (\ref{avEgap}) and (\ref{deviEboundgamma}), one notices the following: 
If the condition (\ref{deviEboundgamma}) is always retained during continuously varying the twisted phase 
from $\phi=\phi'$ to $\phi=0$ for any fixed lattice $\Lambda_n^{(1)}$, 
then the degeneracy of the low-energy sector must be lifted at $\phi=0$.
This contradicts the assumption that the energies of the low-energy sector are 
(quasi)degenerate at $\phi=0$. 
Thus, there must exists a phase $\phi_2'\in[0,\phi')$ such that the condition (\ref{deviEboundgamma}) 
does not hold at the phase $\phi=\phi_2'$.  

{From} these observations, we can assume that there exist 
a subsequence $\{\Lambda_n^{(1)}\}_n$ of the sequence $\{\Lambda_n^{(0)}\}_n$ of the lattices, 
and the minimum value of the phase $\phi_1'\in(0,\phi')$ such that 
the condition (\ref{deviEboundgamma}) holds for all the sectors and that, in particular, 
there exists a sector $i^{(1)}$ that the deviation of the energy in the $i^{(1)}$-th sector satisfies   
\begin{equation}
\label{deviEphipgamma}
\Delta E_{0,i^{(1)}}(\phi_1')=\frac{1}{4}\gamma
\end{equation}
for all the lattices $\Lambda_n^{(1)}$. Actually, since the eigenenergies are a continuous function of 
the twisted phase $\phi$ for a fixed lattice, there exists $\phi_1'\in(0,\phi')$ 
for each lattice $\Lambda_n^{(0)}$ such that 
$$
\max_{1\le i\le \ell}\{\Delta E_{0,i}(\phi_1')\}=\frac{1}{4}\gamma.
$$
Therefore, there exists a subsequence $\{\Lambda_n^{(1)}\}_n$ of the sequence $\{\Lambda_n^{(0)}\}_n$
that the condition (\ref{deviEphipgamma}) holds for a sector $i^{(1)}$. 

The condition (\ref{deviEphipgamma}) implies that the degeneracy of the $i^{(1)}$-th sector is lifted 
at the twisted phase $\phi=\phi_1'$.  Appropriately choosing the subsequence $\{\Lambda_n^{(1)}\}_n$ 
of the sequence $\{\Lambda_n^{(0)}\}_n$ of the lattices, 
the $i^{(1)}$-th sector splits into $\ell(i^{(1)})$ sectors each of which is $q_{i^{(1)},i}$-fold (quasi)degenerate 
for $i=1,2,\ldots,\ell(i^{(1)})$, and the spectral gaps $\gamma_{i^{(1)},m}$ between two adjacent sectors  
appear for $m=1,2,\ldots,\ell(i^{(1)})-1$. Therefore,  
we can define the spectral projection $P_{0,i^{(1)},i}^{(\Lambda,N)}(\phi)$ and the expectation value  
$\omega_{0,i^{(1)},i}^{(\Lambda,N)}(\cdots;\phi)$ for the $(i^{(1)},i)$-th sector for $i=1,2,\ldots,\ell(i^{(1)})$ 
in the same way as in the above. 
If another sector $j^{(1)}\ne i^{(1)}$ satisfies the same condition $\Delta E_{0,j^{(1)}}(\phi_1')=\gamma/4$ 
for an infinite number of the lattices in the sequence $\{\Lambda_n^{(1)}\}_n$, 
then we make the same procedure as the $i^{(1)}$-th sector.  

Further decrease the phase to $\phi=\phi_2'$ satisfying $0\le \phi_2'<\phi_1'<\phi'$. 
Then, we can define the deviation $\Delta E_{0,i^{(1)},i}(\phi_2')$ of the energy 
for the $(i^{(1)},i)$-th sector in the same way, and $\Delta E_{0,j^{(1)},j}(\phi_2')$ 
for the $(j^{(1)},j)$-th sector. 
For a sufficiently small deviation $\phi_1'-\phi_2'$, we can assume that 
\begin{equation}
\label{DeltaE0kphi2gamma}
\Delta E_{0,k}(\phi_2')\le \frac{1}{4}\gamma
\end{equation}
for the $k$-th sector satisfying 
$$
\Delta E_{0,k}(\phi_1')<\frac{1}{4}\gamma 
$$
and 
\begin{equation}
\label{DeltaE0j1jphi2gamma}
\Delta E_{0,j^{(1)},j}(\phi_2')\le \frac{1}{4}\gamma(j^{(1)})
\end{equation}
for the $(j^{(1)},j)$-th sector satisfying 
$$
\Delta E_{0,j^{(1)}}(\phi_1')=\frac{1}{4}\gamma,
$$
where 
$$
\gamma(j^{(1)}):=
\min\{\gamma_{j^{(1)},1},\gamma_{j^{(1)},2},\ldots,\gamma_{j^{(1)},\ell(j^{(1)})-1},\gamma/4\}
$$
If these conditions are always retained during continuously varying the twisted phase from $\phi=\phi_1'$ to 
$\phi=0$ for any fixed lattices $\Lambda_n^{(1)}$, then the degeneracy of the low-energy sector must be 
degenerate at $\phi=0$. This contradicts the assumption again.  

If there exist a subsequence of the lattices and $\phi_2'$ such that the condition 
$$
\Delta E_{0,k}(\phi_2')=\frac{1}{4}\gamma
$$
holds for a sector $k$ in the first case (\ref{DeltaE0kphi2gamma}), 
then we make the same procedure as in the above $i^{(1)}$-th sector.

In the second case (\ref{DeltaE0j1jphi2gamma}), if there exist a subsequence $\{\Lambda_n^{(2)}\}_n$ 
of the sequence $\{\Lambda_n^{(1)}\}_n$ of the lattices and the twisted phase $\phi_2'$ such that the condition
$$
\label{Delta E0j1jphi2gamma}
\Delta E_{0,j^{(1)},j}(\phi_2')=\frac{1}{4}\gamma(j^{(1)})
$$
holds for the $(j^{(1)},j)$-th sector, this condition implies that the degeneracy of the sector is lifted again.  
Therefore, we can make the same procedure. 

Since the initial degeneracy $q$ is finite, the process must stop in finite times. 
Besides, once a spectral gap appears in the spectrum, it never closes 
because the total deviation of energy is bounded by 
$$
\frac{1}{4}\gamma'+\frac{1}{4^2}\gamma'+\cdots=\frac{1}{3}\gamma'<\frac{1}{2}\gamma',    
$$
where $\gamma'$ is the initial gap when the degeneracy of the sector is lifted. 
Thus, the degeneracy of the spectrum at the twisted phase $\phi=0$ is not recovered  
when continuously varying the twisted phase from $\phi=\phi'$ to $\phi=0$.
This is a contradiction. Theorem~\ref{thm:twistDegeneGap} has been proved.

\section{Proof of the inequality (\ref{a0tdiff})}
\label{Proofa0tdiffinequal}

Consider a Hamiltonian $H^{(\Lambda)}$ on the lattice $\Lambda$. 
We denote by $H^{(\Omega)}$ the restriction of the Hamiltonian $H^{(\Lambda)}$ 
onto the subset $\Omega$ of the lattice $\Lambda$. Let $A$ be a local observable, 
and define the time evolutions as 
\[
A^{(\Lambda)}(t):=e^{itH^{(\Lambda)}}Ae^{-itH^{(\Lambda)}}
\]
and 
\[
A^{(\Omega)}(t):=e^{itH^{(\Omega)}}Ae^{-itH^{(\Omega)}}.
\]
Then, one has \cite{NOS} 
\begin{equation}
\label{Atdiff}
\bigl\Vert A^{(\Lambda)}(t)-A^{(\Omega)}(t)\bigr\Vert
\le {\rm sgn}\>t\int_0^t ds\; \bigl\Vert [(H^{(\Lambda)}-H^{(\Omega)}),A^{(\Omega)}(t-s)]\bigl\Vert.
\end{equation}
In order to prove this inequality, we use the following identity:  
\[
A^{(\Lambda)}(t)-A^{(\Omega)}(t)=\int_0^t ds\> \frac{d}{ds}
e^{isH^{(\Lambda)}}A^{(\Omega)}(t-s)e^{-isH^{(\Lambda)}}. 
\]
The derivative of the integrand in the right-hand side is computed as  
\begin{align*}
\frac{d}{ds}
e^{isH^{(\Lambda)}}A^{(\Omega)}(t-s)e^{-isH^{(\Lambda)}}&=
\frac{d}{ds}
e^{isH^{(\Lambda)}}e^{i(t-s)H^{(\Omega)}}Ae^{-i(t-s)H^{(\Omega)}}e^{-isH^{(\Lambda)}}\\
&= ie^{isH^{(\Lambda)}}[(H^{(\Lambda)}-H^{(\Omega)}),A^{(\Omega)}(t-s)]e^{-isH^{(\Lambda)}}. 
\end{align*}
Combining these, the desired inequality (\ref{Atdiff}) is obtained. 

\section{Proof of the inequality (\ref{intf(z)semicirclebound})}
\label{prooff(z)semicirclebound}

The function $f(z)$ in the integrand in the left-hand side of (\ref{intf(z)semicirclebound}) 
is written 
\begin{align*}
f(T_2e^{i\theta})&=\frac{1}{q}\sum_{m=1}^q \sum_{n\ge 1}
\langle\Phi_{0,m}^{(N)}(\phi),A(\phi)\Phi_n^{(N)}(\phi)\rangle
\langle\Phi_n^{(N)}(\phi),B(\phi)\Phi_{0,m}^{(N)}(\phi)\rangle\\
&\times \exp[i(E_n^{(N)}(\phi)-E_{0,m}^{(N)}(\phi))T_2e^{i\theta}],
\end{align*}
where we have written 
\[
A(\phi)=a^\ast P_0^{(\Lambda,N)}(\phi)a
\]
for short. Using $iT_2e^{i\theta}=iT_2\cos\theta-T_2\sin\theta$, the Schwarz inequality 
and the assumption on the spectral gap, one has 
\[
\left|f(T_2e^{i\theta})\right|\le \Vert a\Vert^2\Vert B(\phi)\Vert 
\exp[-T_2\Delta E^{\rm min}\sin\theta].
\]
Therefore, the integral is evaluated as 
\begin{align*}
\left|\int \frac{f(z)}{z-is}dz\right|&\le \int_0^\pi d\theta \left|f(T_2e^{i\theta})\right|
\frac{T_2}{\sqrt{T_2^2\cos^2\theta+(T_2\sin\theta-s)^2}}\\
&\le \Vert a\Vert^2\Vert B(\phi)\Vert \int_0^\pi d\theta 
\exp[-T_2\Delta E^{\rm min}\sin\theta]\frac{1}{\sqrt{1-2s/T_2}}\\
&\le \Vert a\Vert^2\Vert B(\phi)\Vert\frac{\pi}{\Delta E^{\rm min}}
\frac{1-e^{-T_2\Delta E^{\rm min}}}{T_2}
\frac{1}{\sqrt{1-2s/T_2}}.
\end{align*}
This is the desired bound (\ref{intf(z)semicirclebound}). 

\section{Proof of the inequality (\ref{I1bound})}
\label{estI1}

Note that 
\begin{align}
\label{omegaBphiAphicomp}
&\omega_0^{(\Lambda,N)}(\tilde{B}^{(\Lambda)}(t;\phi)A'(\phi);\phi)\\ \nonumber
&=\frac{1}{q}\sum_{m=1}^q \sum_{n\ge 1}\langle\Phi_{0,m}^{(N)}(\phi),B(\phi)\Phi_n^{(N)}(\phi)\rangle 
\langle\Phi_n^{(N)}(\phi),A'(\phi)\Phi_{0,m}^{(N)}(\phi)\rangle \\ \nonumber
&\times \exp[i(E_{0,m}^{(N)}(\phi)-E_n^{(N)}(\phi))t],
\end{align}
where we have written 
\[
A'(\phi)=a P_0^{(\Lambda,N)}(\phi) a^\ast 
\]
for short. Therefore, the integral with respect to time $t$ in (\ref{I1}) is written 
\begin{equation}
\frac{1}{2\pi i}\int_{-T_2}^{+T_2}dt\; \frac{1}{t-is}\exp[-i(E_n^{(N)}(\phi)-E_{0,m}^{(N)}(\phi))t]
e^{-\alpha t^2}. 
\end{equation}
In order to estimate this integral, we prepare a tool \cite{Hastings,HastingsKoma,NachtergaeleSims}. 

For $E\in\re$, one has 
\[
e^{iEt}e^{-\alpha t^2}=\frac{1}{2\sqrt{\pi \alpha}}\int_{-\infty}^{+\infty}dw\; e^{iwt}
\exp[-(w-E)^2/(4\alpha)].
\]
Therefore, 
\begin{align}
\label{inteiEtalphat2}
\frac{1}{2\pi i}\int_{-T_2}^{+T_2}\frac{e^{iEt}e^{-\alpha t^2}}{t-z}dt
&=\frac{1}{2\sqrt{\pi \alpha}}\int_{-\infty}^{+\infty}dw\; \exp[-(w-E)^2/(4\alpha)]\\ \nonumber
&\times\frac{1}{2\pi i}\int_{-T_2}^{+T_2}dt\; \frac{e^{iwt}}{t-z}
\end{align}
for $z\in\co$ satisfying ${\rm Im}\; z>0$. For the integral with respect to time $t$, we write 
\begin{equation}
\label{defFzT2}
F_{z,T_2}(w)=\frac{1}{2\pi i}\int_{-T_2}^{+T_2}dt\; \frac{e^{iwt}}{t-z}
\end{equation}
with ${\rm Im}\; z>0$. Then, the integral (\ref{inteiEtalphat2}) is written 
\begin{align}
\label{intexpiEtalphat2}
\frac{1}{2\pi i}\int_{-T_2}^{+T_2}\frac{e^{iEt}e^{-\alpha t^2}}{t-z}dt
&=\frac{1}{2\sqrt{\pi \alpha}}\int_{-\infty}^{+\infty}dw\; \exp[-(w-E)^2/(4\alpha)]F_{z,T_2}(w)\\ \nonumber
&=\frac{1}{2\sqrt{\pi \alpha}}\int_0^{+\infty}dw\; e^{iwz}\exp[-(w-E)^2/(4\alpha)]\\ \nonumber
&+R_1+R_2
\end{align}
with
\[
R_1=\frac{1}{2\sqrt{\pi \alpha}}\int_0^{+\infty}dw\; \exp[-(w-E)^2/(4\alpha)]
\left[F_{z,T_2}(w)-e^{iwz}\right]
\]
and
\[
R_2=\frac{1}{2\sqrt{\pi \alpha}}\int_{-\infty}^0dw\; \exp[-(w-E)^2/(4\alpha)]F_{z,T_2}(w).
\]
One can easily obtain 
\begin{equation}
\label{FzT2+bound}
\left|F_{z,T_2}(w)-e^{iwz}\right|\le \frac{1}{wT_2}\left[1-e^{-wT_2}\right],\quad \mbox{for\ $w>0$}
\end{equation}
and
\[
\left|F_{z,T_2}(w)\right|\le \frac{1}{|w|T_2}\left[1-e^{-|w|T_2}\right],\quad \mbox{for\ $w<0$},
\]
where we have assumed ${\rm Im}\; z>0$. Using these bounds, one has 
\[
|R_1|\le\frac{1}{2\sqrt{\pi \alpha}}\int_0^{+\infty}dw\; \exp[-(w-E)^2/(4\alpha)]
\frac{1}{wT_2}\left[1-e^{-wT_2}\right] 
\]
and
\[
|R_2|\le \frac{1}{2\sqrt{\pi \alpha}}\int_{-\infty}^0dw\; \exp[-(w-E)^2/(4\alpha)]
\frac{1}{|w|T_2}\left[1-e^{-|w|T_2}\right].
\]
Further, applying the Schwarz inequality yields  
\begin{align*}
|R_1|+|R_2|&\le \frac{1}{2\sqrt{\pi \alpha}}\int_{-\infty}^{+\infty}dw\; \exp[-(w-E)^2/(4\alpha)]
\frac{1}{|w|T_2}\left[1-e^{-|w|T_2}\right]\\
&\le {\rm Const.}\frac{1}{\alpha^{1/4}T_2^{1/2}}. 
\end{align*}
Thus, $R_1$ and $R_2$ are vanishing in the limit $T_2\rightarrow\infty$. 

{From} (\ref{omegaBphiAphicomp}) and (\ref{intexpiEtalphat2}), the integral $I_1$ is written  
\begin{align*}
I_1&=\frac{1}{2\pi i}\int_{-T_2}^{+T_2}dt\; \frac{1}{t-is}
\omega_0^{(\Lambda,N)}(\tilde{B}^{(\Lambda)}(t;\phi)A'(\phi);\phi)e^{-\alpha t^2}\\
&=\frac{1}{q}\sum_{m=1}^q\sum_{n\ge 1}\;
\langle\Phi_{0,m}^{(N)}(\phi),B(\phi)\Phi_n^{(N)}(\phi)\rangle 
\langle\Phi_n^{(N)}(\phi),A'(\phi)\Phi_{0,m}^{(N)}(\phi)\rangle \\
&\times\frac{1}{2\pi i}\int_{-T_2}^{+T_2}dt\; \frac{1}{t-is}\exp[-i(E_n^{(N)}(\phi)-E_{0,m}^{(N)}(\phi))]
e^{-\alpha t^2}\\
&=\frac{1}{q}\sum_{m=1}^q\sum_{n\ge 1}\; \langle\Phi_{0,m}^{(N)}(\phi),B(\phi)\Phi_n^{(N)}(\phi)\rangle 
\langle\Phi_n^{(N)}(\phi),A'(\phi)\Phi_{0,m}^{(N)}(\phi)\rangle \\
&\times\biggl[\frac{1}{2\sqrt{\pi \alpha}}\int_0^{\infty}dw\; e^{-ws}
\exp\bigl[-\{w+(E_n^{(N)}(\phi)-E_{0,m}^{(N)}(\phi))\}^2/(4\alpha)\bigr]\\
&\qquad\qquad\qquad\qquad +R_1+R_2\biggr]. 
\end{align*}
Therefore, we have 
\[
|I_1|\le\frac{1}{2}\Vert B(\phi)\Vert \Vert A'(\phi)\Vert e^{-(\Delta E^{\rm min})^2/(4\alpha)}
\le\frac{1}{2}\Vert B(\phi)\Vert \Vert a \Vert^2 e^{-(\Delta E^{\rm min})^2/(4\alpha)}
\]
in the limit $T_2\rightarrow\infty$. Here, we have used the assumption (\ref{assumpGapphi0})  
on the spectral gap.

\section{Proof of the inequality (\ref{I2bound})}
\label{estI2}

Let $S$ be a large positive number which satisfies $S<T_2$. We decompose the integral $I_2$ into three parts as  
\begin{align*}
I_2&=\frac{1}{2\pi i}\int_{-T_2}^{+T_2}dt\frac{1}{t-is}
\omega_0^{(\Lambda,N)}(a^\ast P_0^{(\Lambda,N)}(\phi)[a,B^{(\Lambda)}(t;\phi)];\phi)e^{-\alpha t^2}\\
&=I_2^{-}+I_2^{0}+I_2^{+}
\end{align*}
with
\[
I_2^{-}=\frac{1}{2\pi i}\int_{-T_2}^{-S}dt\frac{1}{t-is}
\omega_0^{(\Lambda,N)}(a^\ast P_0^{(\Lambda,N)}(\phi)[a,B^{(\Lambda)}(t;\phi)];\phi)e^{-\alpha t^2},
\]
\[
I_2^{0}=\frac{1}{2\pi i}\int_{-S}^{+S}dt\frac{1}{t-is}
\omega_0^{(\Lambda,N)}(a^\ast P_0^{(\Lambda,N)}(\phi)[a,B^{(\Lambda)}(t;\phi)];\phi)e^{-\alpha t^2}
\]
and
\[
I_2^{+}=\frac{1}{2\pi i}\int_{+S}^{+T_2}dt\frac{1}{t-is}
\omega_0^{(\Lambda,N)}(a^\ast P_0^{(\Lambda,N)}(\phi)[a,B^{(\Lambda)}(t;\phi)];\phi)e^{-\alpha t^2}.
\]
The first and third integral are estimated as 
\[
|I_2^{\pm}|\le \frac{1}{2\pi}\frac{1}{\sqrt{S^2+s^2}}\sqrt{\frac{\pi}{\alpha}}\; \Vert a\Vert^2 
\Vert B(\phi)\Vert e^{-\alpha S^2}
\le \frac{\Vert a\Vert^2 
\bigl\Vert B(\phi)\bigr\Vert}{2\sqrt{\pi\alpha}}\frac{e^{-\alpha S^2}}{S}.
\]

For the second integral $I_2^0$, we use the Lieb-Robinson bounds \cite{HastingsKoma,NachtergaeleSims}. 
Note that the operator $B(\phi)$ is a sum of the local operators $b_j(\phi)$ as 
\[
B(\phi)=\sum_j b_j(\phi). 
\] 
{From} this and the Lieb-Robinson bounds, one has    
\begin{align*}
\bigl\Vert [a,B^{(\Lambda)}(t;\phi)]\bigr\Vert 
&\le \sum_j\bigl\Vert [a,b_j^{(\Lambda)}(t;\phi)]\bigr\Vert\\ 
&\le{\rm Const.}|{\rm supp}\; a|\Vert a\Vert 
\sum_j|{\rm supp}\; b_j(\phi)|\Vert b_j(\phi)\Vert e^{-\mu r}(e^{v|t|}-1)\\
&\le {\rm Const.}|{\rm supp}\; a|\Vert a\Vert \bigl\Vert B(\phi)\bigr\Vert e^{-\mu r}(e^{v|t|}-1),
\end{align*}
where the operator $b_j^{(\Lambda)}(t;\phi)$ is the time evolution of $b_j(\phi)$. 
Using this bound, the integral $I_2^0$ can be estimated as 
\begin{align*}
|I_2^0|&\le{\rm Const.}|{\rm supp}\; a|\Vert a\Vert^2 \bigl\Vert B(\phi)\bigr\Vert e^{-\mu r}
\int_{-S}^{+S}dt\; \frac{e^{v|t|}-1}{|t|}e^{-\alpha t^2}\\ 
&\le {\rm Const.}|{\rm supp}\; a|\Vert a\Vert^2 \bigl\Vert B(\phi)\bigr\Vert e^{-\mu r}
\int_0^{+S}dt\; \frac{e^{vt}-1}{t}\\
&\le {\rm Const.}|{\rm supp}\; a|\Vert a\Vert^2 \bigl\Vert B(\phi)\bigr\Vert\exp[vS-\mu r]. 
\end{align*}
Combining this with the bounds for $I_2^{\pm}$, one obtains
\[
|I_2|\le \Vert a\Vert^2 
\bigl\Vert B(\phi)\bigr\Vert\left\{
\frac{1}{\sqrt{\pi\alpha}}\frac{e^{-\alpha S^2}}{S}
+{\rm Const.}|{\rm supp}\; a|\exp[vS-\mu r]\right\}. 
\]
\section{Proof of the inequality (\ref{I3bound})}
\label{estI3}

Note that 
\begin{align}
\label{f(t)inI3}
f(t)&=\omega_0^{(\Lambda,N)}(A(\phi)\tilde{B}^{(\Lambda)}(t;\phi);\phi)\\ \nonumber
&=\frac{1}{q}\sum_{m=1}^q\sum_{n\ge 1}\langle\Phi_{0,m}^{(N)}(\phi),A(\phi)\Phi_n^{(N)}(\phi)\rangle
\langle \Phi_n^{(N)}(\phi),B(\phi)\Phi_{0,m}^{(N)}(\phi)\rangle \\ \nonumber
&\qquad\qquad \times \exp[i(E_n^{(N)}(\phi)-E_{0,m}^{(N)}(\phi))t]. 
\end{align}
Therefore, the integral with respect to time $t$ in $I_3$ is 
\begin{equation}
\label{tintinI3}
\frac{1}{2\pi i}\int_{-T_2}^{+T_2}dt\; \frac{1}{t-is}\exp[i\Delta E_{n,m}t]
(e^{\alpha s^2}-e^{-\alpha t^2}), 
\end{equation}
where we have written $\Delta E_{n,m}=E_n^{(N)}(\phi)-E_{0,m}^{(N)}(\phi)$ for short. 
{From} (\ref{defFzT2}) and (\ref{FzT2+bound}), one has 
\begin{align*}
&\frac{1}{2\pi i}\int_{-T_2}^{+T_2}dt\; \frac{1}{t-is}\exp[i\Delta E_{n,m}t]
e^{\alpha s^2}\\
&=\exp[\alpha s^2-\Delta E_{n,m}s]+R_3,
\end{align*}
where the correction $R_3$ is vanishing in the limit $T_2\rightarrow\infty$ and 
we have used the assumption (\ref{assumpGapphi0}) on the spectral gap. 
Since 
\[
\frac{1}{2\sqrt{\pi\alpha}}\int_{-\infty}^{+\infty}e^{-sw}\exp[-(w-\Delta E_{n,m})^2/(4\alpha)]
=\exp[\alpha s^2-\Delta E_{n,m}s], 
\]
the above integral is written 
\begin{align*}
&\frac{1}{2\pi i}\int_{-T_2}^{+T_2}dt\; \frac{1}{t-is}\exp[i\Delta E_{n,m}t]
e^{\alpha s^2}\\
&=\frac{1}{2\sqrt{\pi\alpha}}\int_{-\infty}^{+\infty}e^{-sw}\exp[-(w-\Delta E_{n,m})^2/(4\alpha)]+R_3.
\end{align*}
Combining this with (\ref{intexpiEtalphat2}), one obtains
\begin{align*}
&\frac{1}{2\pi i}\int_{-T_2}^{+T_2}dt\; \frac{1}{t-is}\exp[i\Delta E_{n,m}t]
(e^{\alpha s^2}-e^{-\alpha t^2})\\
&=\frac{1}{2\sqrt{\pi\alpha}}\int_{-\infty}^{+\infty}e^{-sw}\exp[-(w-\Delta E_{n,m})^2/(4\alpha)]+R_3\\
&-\frac{1}{2\sqrt{\pi \alpha}}\int_0^{+\infty}dw\; e^{-sw}\exp[-(w-\Delta E_{n,m})^2/(4\alpha)] 
-R_1-R_2\\
&=\frac{1}{2\sqrt{\pi\alpha}}\int_{-\infty}^0 e^{-sw}\exp[-(w-\Delta E_{n,m})^2/(4\alpha)]
-R_1-R_2+R_3.
\end{align*}
Using $0<s\le cr$ and $\Delta E_{n,m}\ge \Delta E^{\rm min}$, and choosing $\alpha=\Delta E^{\rm min}/(2cr)$, 
the above integral in the last line can be evaluated as 
\[
\frac{1}{2\sqrt{\pi\alpha}}\int_{-\infty}^0 e^{-sw}\exp[-(w-\Delta E_{n,m})^2/(4\alpha)]
\le \frac{1}{2}e^{-\Delta E^{\rm min}cr/2}. 
\]
Therefore, from (\ref{f(t)inI3}) and (\ref{tintinI3}), the integral $I_3$ of (\ref{I3}) is 
estimated as  
\[
|I_3|\le \frac{1}{2}\Vert a\Vert^2\Vert B(\phi)\Vert e^{-\Delta E^{\rm min}cr/2}
\]
in the limit $T_2\rightarrow\infty$. 

\section{Differentiability of $\Theta^{(2)}(\phi_1)$}
\label{DiffTheta}

Fix $\phi_1^{(0)}\in[0,2\pi)$, and consider $\phi_1$ in the neighborhood of $\phi_1^{(0)}$. 
Then, the product $C^{(2)}(\phi_1)C^{(2)}(\phi_1^{(0)})^\dagger$ of the two unitary matrices 
is nearly equal to the identity matrix. Consider the Hermitian matrix $\Delta\Theta^{(2)}(\phi_1)$ which 
is defined by 
$$
\exp\big[i\Delta\Theta^{(2)}(\phi_1)\big]=C^{(2)}(\phi_1)C^{(2)}(\phi_1^{(0)})^\dagger. 
$$
For a sufficiently small $|\phi_1-\phi_1^{(0)}|$, the matrix $\Delta\Theta^{(2)}(\phi_1)$ is 
well defined and differentiable with respect to $\phi_1$. Clearly, one has 
\begin{equation}
\label{C(2)expDeltaTheta}
C^{(2)}(\phi_1)=\exp\big[i\Delta\Theta^{(2)}(\phi_1)\big]C^{(2)}(\phi_1^{(0)}).
\end{equation}
Further, one obtains 
\begin{align}
\label{intTrCC}
\int_{\phi_1^{(0)}}^{\phi_1^{(1)}}d\phi_1\; {\rm Tr}\; C^{(2)}(\phi_1)^\dagger
\frac{\partial}{\partial\phi_1}C^{(2)}(\phi_1)&=
i\int_{\phi_1^{(0)}}^{\phi_1^{(1)}}d\phi_1\; {\rm Tr}\;\frac{\partial}{\partial\phi_1}\Delta\Theta^{(2)}(\phi_1)\\
\nonumber &=i\big[{\rm Tr}\;\Delta\Theta^{(2)}(\phi_1^{(1)})-{\rm Tr}\Delta\Theta^{(2)}(\phi_1^{(0)})\;\big]\\
\nonumber &=i{\rm Tr}\;\Delta\Theta^{(2)}(\phi_1^{(1)}).  
\end{align} 

On the other hand, using the expression
$$
C^{(2)}(\phi_1)=\exp\big[i\Theta^{(2)}(\phi_1)\big]
$$
and taking the determinant for both sides of (\ref{C(2)expDeltaTheta}), we have 
$$
{\rm Tr}\; \Theta^{(2)}(\phi_1)\equiv {\rm Tr}\;\Delta\Theta^{(2)}(\phi_1)+{\rm Tr}\; \Theta^{(2)}(\phi_1^{(0)})
\quad (\mbox{mod \ }2\pi).
$$
This implies 
$$
{\rm Tr}\;\Delta\Theta^{(2)}(\phi_1^{(1)})\equiv {\rm Tr}\; \Theta^{(2)}(\phi_1^{(1)})
-{\rm Tr}\; \Theta^{(2)}(\phi_1^{(0)})\quad (\mbox{mod \ }2\pi).
$$
Substituting this into the right-hand side of (\ref{intTrCC}), we obtain 
\begin{multline*}
\int_{\phi_1^{(0)}}^{\phi_1^{(1)}}d\phi_1\; {\rm Tr}\; C^{(2)}(\phi_1)^\dagger
\frac{\partial}{\partial\phi_1}C^{(2)}(\phi_1)\\
\equiv i\big[{\rm Tr}\; \Theta^{(2)}(\phi_1^{(1)})
-{\rm Tr}\; \Theta^{(2)}(\phi_1^{(0)})\big]\quad (\mbox{mod \ }2\pi).
\end{multline*}
This result justifies the formula (\ref{I(2)TrTheta}). 

\section{Proof of Lemma~\ref{lem:TopoInvUnch}}
\label{ProofTopoInvUnch}

To begin with, we recall that the relation between the two projections is given by  
$$
P_0^{(\Lambda,N)}(\tilde{\boldsymbol{\phi}},\tilde{\alpha})=\exp[-i\alpha\phi_1n_x]
P_0^{(\Lambda,N)}(\tilde{\boldsymbol{\phi}})
\exp[i\alpha\phi_1n_x]. 
$$
By differentiating with respect to $\phi_1$, one obtains 
\begin{align*}
\frac{\partial}{\partial\phi_1}P_0^{(\Lambda,N)}(\tilde{\boldsymbol{\phi}},\tilde{\alpha})&=
\exp[-i\alpha\phi_1n_x]
P_{0,1}^{(\Lambda,N)}(\tilde{\boldsymbol{\phi}})
\exp[i\alpha\phi_1n_x]\\
&-(i\alpha n_x)\exp[-i\alpha\phi_1n_x]
P_0^{(\Lambda,N)}(\tilde{\boldsymbol{\phi}})
\exp[i\alpha\phi_1n_x]\\
&+\exp[-i\alpha\phi_1n_x]
P_0^{(\Lambda,N)}(\tilde{\boldsymbol{\phi}})
\exp[i\alpha\phi_1n_x](i\alpha n_x). 
\end{align*}
Using these relations, we have 
\begin{align*}
&{\rm Tr}\; P_0^{(\Lambda,N)}(\tilde{\boldsymbol{\phi}},\tilde{\alpha})
[P_{0,1}^{(\Lambda,N)}(\tilde{\boldsymbol{\phi}},\tilde{\alpha}),
P_{0,2}^{(\Lambda,N)}(\tilde{\boldsymbol{\phi}},\tilde{\alpha})]\\ 
&={\rm Tr}\; P_0^{(\Lambda,N)}(\tilde{\boldsymbol{\phi}})
[P_{0,1}^{(\Lambda,N)}(\tilde{\boldsymbol{\phi}}),
P_{0,2}^{(\Lambda,N)}(\tilde{\boldsymbol{\phi}})]\\
&-i\alpha\; {\rm Tr}\; P_0^{(\Lambda,N)}(\tilde{\boldsymbol{\phi}})
[n_x P_0^{(\Lambda,N)}(\tilde{\boldsymbol{\phi}}),
P_{0,2}^{(\Lambda,N)}(\tilde{\boldsymbol{\phi}})]\\
&+i\alpha\; {\rm Tr}\; P_0^{(\Lambda,N)}(\tilde{\boldsymbol{\phi}})
[P_0^{(\Lambda,N)}(\tilde{\boldsymbol{\phi}})n_x,
P_{0,2}^{(\Lambda,N)}(\tilde{\boldsymbol{\phi}})].
\end{align*}
Therefore, it is enough to handle the second and the third terms in the right-hand side.  
Except for the factor $i\alpha$, these two terms are written 
\begin{align*}
&{\rm Tr}\; P_0^{(\Lambda,N)}(\tilde{\boldsymbol{\phi}})
P_{0,2}^{(\Lambda,N)}(\tilde{\boldsymbol{\phi}})n_x 
-{\rm Tr}\; P_0^{(\Lambda,N)}(\tilde{\boldsymbol{\phi}})
n_x P_0^{(\Lambda,N)}(\tilde{\boldsymbol{\phi}})P_{0,2}^{(\Lambda,N)}(\tilde{\boldsymbol{\phi}})\\
&+{\rm Tr}\; P_0^{(\Lambda,N)}(\tilde{\boldsymbol{\phi}})n_x
P_{0,2}^{(\Lambda,N)}(\tilde{\boldsymbol{\phi}})
-{\rm Tr}\; P_0^{(\Lambda,N)}(\tilde{\boldsymbol{\phi}})
P_{0,2}^{(\Lambda,N)}(\tilde{\boldsymbol{\phi}})P_0^{(\Lambda,N)}(\tilde{\boldsymbol{\phi}})n_x\\
&={\rm Tr}\; P_0^{(\Lambda,N)}(\tilde{\boldsymbol{\phi}})
P_{0,2}^{(\Lambda,N)}(\tilde{\boldsymbol{\phi}})
\big[1-P_0^{(\Lambda,N)}(\tilde{\boldsymbol{\phi}})\big]n_x\\
&+{\rm Tr}\; P_0^{(\Lambda,N)}(\tilde{\boldsymbol{\phi}})
n_x\big[1-P_0^{(\Lambda,N)}(\tilde{\boldsymbol{\phi}})\big]P_{0,2}^{(\Lambda,N)}(\tilde{\boldsymbol{\phi}})\\
&=\sum_{m=1}^q \Big[\big\langle\hat{\Phi}_{0,m}^{(N)}(\tilde{\boldsymbol{\phi}}),
P_{0,2}^{(\Lambda,N)}(\tilde{\boldsymbol{\phi}})
\big[1-P_0^{(\Lambda,N)}(\tilde{\boldsymbol{\phi}})\big]n_x
\hat{\Phi}_{0,m}^{(N)}(\tilde{\boldsymbol{\phi}})\big\rangle\\
&+\big\langle\hat{\Phi}_{0,m}^{(N)}(\tilde{\boldsymbol{\phi}}), 
n_x\big[1-P_0^{(\Lambda,N)}(\tilde{\boldsymbol{\phi}})\big]P_{0,2}^{(\Lambda,N)}(\tilde{\boldsymbol{\phi}})
\hat{\Phi}_{0,m}^{(N)}(\tilde{\boldsymbol{\phi}})\big\rangle\Big],
\end{align*}
where we have used Proposition~\ref{prop:Kato}. Further, by using (\ref{diffPhiid}), this is written 
\begin{align*}
&\sum_{m=1}^q \Big[\Big\langle\frac{\partial}{\partial\phi_2}\hat{\Phi}_{0,m}^{(N)}(\tilde{\boldsymbol{\phi}}),
\big[1-P_0^{(\Lambda,N)}(\tilde{\boldsymbol{\phi}})\big]n_x
\hat{\Phi}_{0,m}^{(N)}(\tilde{\boldsymbol{\phi}})\Big\rangle\\
&+\Big\langle\hat{\Phi}_{0,m}^{(N)}(\tilde{\boldsymbol{\phi}}), 
n_x\big[1-P_0^{(\Lambda,N)}(\tilde{\boldsymbol{\phi}})\big]
\frac{\partial}{\partial\phi_2}\hat{\Phi}_{0,m}^{(N)}(\tilde{\boldsymbol{\phi}})\Big\rangle\Big]\\
&=\sum_{m=1}^q \Big[\Big\langle\frac{\partial}{\partial\phi_2}\hat{\Phi}_{0,m}^{(N)}(\tilde{\boldsymbol{\phi}}),
n_x
\hat{\Phi}_{0,m}^{(N)}(\tilde{\boldsymbol{\phi}})\Big\rangle
+\Big\langle\hat{\Phi}_{0,m}^{(N)}(\tilde{\boldsymbol{\phi}}),n_x
\frac{\partial}{\partial\phi_2}\hat{\Phi}_{0,m}^{(N)}(\tilde{\boldsymbol{\phi}})\Big\rangle\Big]\\
&-\sum_{m=1}^q\sum_{m'=1}^q 
\Big[\Big\langle\frac{\partial}{\partial\phi_2}\hat{\Phi}_{0,m}^{(N)}(\tilde{\boldsymbol{\phi}}),
\hat{\Phi}_{0,m'}^{(N)}(\tilde{\boldsymbol{\phi}})\Big\rangle
\Big\langle\hat{\Phi}_{0,m'}^{(N)}(\tilde{\boldsymbol{\phi}}),n_x
\hat{\Phi}_{0,m}^{(N)}(\tilde{\boldsymbol{\phi}})\Big\rangle\\
&+\Big\langle\hat{\Phi}_{0,m}^{(N)}(\tilde{\boldsymbol{\phi}}),n_x
\hat{\Phi}_{0,m'}^{(N)}(\tilde{\boldsymbol{\phi}})\Big\rangle
\Big\langle\hat{\Phi}_{0,m'}^{(N)}(\tilde{\boldsymbol{\phi}}),
\frac{\partial}{\partial\phi_2}\hat{\Phi}_{0,m}^{(N)}(\tilde{\boldsymbol{\phi}})\Big\rangle\Big]\\
&=\sum_{m=1}^q \frac{\partial}{\partial\phi_2}\Big\langle\hat{\Phi}_{0,m}^{(N)}(\tilde{\boldsymbol{\phi}}),
n_x\hat{\Phi}_{0,m}^{(N)}(\tilde{\boldsymbol{\phi}})\Big\rangle\\
&-\sum_{m=1}^q\sum_{m'=1}^q
\Big\langle\hat{\Phi}_{0,m}^{(N)}(\tilde{\boldsymbol{\phi}}),
n_x\hat{\Phi}_{0,m'}^{(N)}(\tilde{\boldsymbol{\phi}})\Big\rangle
\frac{\partial}{\partial\phi_2}\Big\langle\hat{\Phi}_{0,m'}^{(N)}(\tilde{\boldsymbol{\phi}}),
\hat{\Phi}_{0,m}^{(N)}(\tilde{\boldsymbol{\phi}})\Big\rangle.
\end{align*}
Clearly, the double sum in the right-hand side is vanishing. Integrating the first sum with respect to 
$\phi_2$ yields  
\begin{align*}
&\int_0^{2\pi}d\phi_2\; \sum_{m=1}^q \frac{\partial}{\partial\phi_2}
\Big\langle\hat{\Phi}_{0,m}^{(N)}(\tilde{\boldsymbol{\phi}}),
n_x\hat{\Phi}_{0,m}^{(N)}(\tilde{\boldsymbol{\phi}})\Big\rangle\\
&=\sum_{m=1}^q\Big[\Big\langle\hat{\Phi}_{0,m}^{(N)}(\phi_1,2\pi),
n_x\hat{\Phi}_{0,m}^{(N)}(\phi_1,2\pi)\Big\rangle
-\Big\langle\hat{\Phi}_{0,m}^{(N)}(\phi_1,0),
n_x\hat{\Phi}_{0,m}^{(N)}(\phi_1,0)\Big\rangle\Big].
\end{align*}
By using the relation (\ref{PhiC(2)Phi}), we can show that the right-hand side is vanishing. 
Actually, we obtain 
\begin{align*}
&\sum_{m=1}^q\Big\langle\hat{\Phi}_{0,m}^{(N)}(\phi_1,2\pi),
n_x\hat{\Phi}_{0,m}^{(N)}(\phi_1,2\pi)\Big\rangle\\
&=\sum_{m=1}^q\sum_{m'=1}^q\sum_{m''=1}^q C_{m,m'}^{(2)}(\phi_1)^\ast C_{m,m''}^{(2)}(\phi_1)
\Big\langle\hat{\Phi}_{0,m'}^{(N)}(\phi_1,0),
n_x\hat{\Phi}_{0,m''}^{(N)}(\phi_1,0)\Big\rangle\\
&=\sum_{m'=1}^q \Big\langle\hat{\Phi}_{0,m'}^{(N)}(\phi_1,0),
n_x\hat{\Phi}_{0,m'}^{(N)}(\phi_1,0)\Big\rangle. 
\end{align*}
Thus, the corresponding contribution is vanishing.

\bigskip\bigskip
\thanks{\textbf{Acknowledgement:} It is a pleasure to thank the following people for 
discussions and correspondence: Hosho Katsura and Mahito Kohmoto.}



\begin{thebibliography}{99}
\bibitem{AG} Aizenman, M., Graf, G. M.: Localization Bounds for an 
Electron Gas. J. Phys. A{\bf 31}, 6783--6806 (1998) 

\bibitem{ASS}  Avron, J. E., Seiler, R., Simon, B.:
Charge Deficiency, Charge Transport and Comparison of Dimensions. 
Commun. Math. Phys. {\bf 159}, 399--422 (1994)

\bibitem{BVS} Bellissard, J., Van Elst, A., Schulz-Baldes, H.: 
The Noncommutative Geometry of the Quantum Hall Effect. 
J. Math. Phys. {\bf 35}, 5373--5451 (1994) 

\bibitem{Connes} Connes, A.: {\it Noncommutative Geometry}, Academic Press, 
San Diego, 1994

\bibitem{ElSch} Elgart, A., Schlein, B.: 
Adiabatic Charge Transport and the Kubo Formula for Landau Type  
Hamiltonians. Comm. Pure Appl. Math. {\bf 57}, 590--615 (2004)  

\bibitem{HMNOSV} Haegeman, J., Michalakis, S., Nachtergaele, B., Osborne, T. J., Schuch, N., Verstraete, F.: 
Elementary excitations in gapped quantum spin systems. 
Phys. Rev. Lett. {\bf 111} (2013), 080401. 

\bibitem{Hastings} Hastings, M. B.: Lieb-Schultz-Mattis in higher dimensions. 
Phys. Rev. B{\bf 69} (2004), 104431. 

\bibitem{HastingsKoma} Hastings, M. B., Koma, T.: Spactral gap and exponential decay of correlations. 
Commun. Math. Phys. {\bf 265} (2006), 781--804. 

\bibitem{HastingsMichalakis} Hastings, M. B., Michalakis, S.: 
Quantization of Hall Conductance For Interacting Electrons on a Torus.
Commun. Math. Phys. {\bf 334} (2015), 433--471.
 
\bibitem{Hatsugai} Hatsugai, Y.: Quantized Berry phases as a local order parameters of 
a quantum liquid. 
J. Phys. Soc. Jpn. {\bf 75} (2006), 123601. 

\bibitem{HKH1} Hirano, T., Katsura, H., Hatsugai, Y.: 
Topological classification of gapped spin chains: Quantized Berry phase as 
a local order parameter. 
Phys. Rev. B {\bf 77} (2008), 094431.  

\bibitem{HKH} Hirano, T., Katsura, H., Hatsugai, Y.: 
Degeneracy and consistency condition for Berry phases: Gap closing under a local gauge twist. 
Phys. Rev. B {\bf 78} (2008), 054431. 

\bibitem{Kato} Kato, T.: Perturbation Theory for Linear Operators, 2nd ed. Springer, 
Berlin, Heidelberg, New York, 1980. 

\bibitem{KawajiWakabayashi} Kawaji, S., Wakabayashi, J.: 
Temperature Dependence of Transverse and Hall Conductivities of Silicon 
MOS Inversion Layers under Strong Magnetic Fields. In: 
{\it Physics in High Magnetic Fields.} Chikazumi, S., Miura, N. (eds), 
pp.~284--287, Springer, Berlin, Heidelberg, New York, 1981 

\bibitem{KDP} von Klitzing, K., Dorda, G., Pepper, M.: 
New Method for High Accuracy Determination of the Fine Structure Constant 
Based on Quantized Hall Resistance. 
Phys. Rev. Lett. {\bf 45}, 494--497 (1980) 

\bibitem{Kohmoto} Kohmoto, M.: Topological Invariant and the Quantization 
of the Hall Conductance. Ann. Phys. {\bf 160}, 343--354 (1985) 

\bibitem{Koma4} Koma, T.: Spectral Gaps of Quantum Hall Systems with 
Interactions. J. Stat. Phys. {\bf 99}, 313--381 (2000) 

\bibitem{Koma3} Koma, T.: Insensitivity of Quantized Hall Conductance to 
Disorder and Interactions. J. Stat. Phys. {\bf 99}, 383--459 (2000)
  
\bibitem{Koma1} Koma, T.:
Revisiting the charge transport in quantum Hall systems.
Rev.\ Math.\ Phys.\ {\bf 16} (2004), 1115--1189.

\bibitem{Koma2} Koma, T.:
Widths of the Hall Conductance Plateaus. 
{\it J. Stat. Phys.} {\bf 130} (2008), 843-934. 

\bibitem{LiebRobinson} Lieb, E. H., Robinson, D. W.: The finite group velocity of quantum spin systems. 
Commun. Math. Phys. {\bf 28} (1972), 251--257. 

\bibitem{NOS} Nachtergaele, B., Ogata, Y., Sims, R.: Propagation of correlations 
in quantum lattice systems. J. Stat. Phys. {\bf 124} (2006), 
1--13. 

\bibitem{NachtergaeleSims} Nachtergaele, B., Sims, R.: Lieb-Robinson bounds and the exponential clustering 
theorem. 
Commun. Math. Phys. {\bf 265} (2006), 119--130. 

\bibitem{NSCM} Neupert, T., Santos, L., Chamon, C., Mudry, C.:
Fractional quantum Hall states at zero magnetic field. 
Phys.\ Rev.\ Lett. {\bf 106} (2011), 236804. 

\bibitem{NTW} Niu, Q., Thouless, D. J., Wu, Y. S.: Quantized Hall conductance 
as a topological invariant. Phys. Rev. B{\bf 31} (1985), 3372--3377. 

\bibitem{RB} Regnault, N., Bernevig, B. A.: Fractional Chern insulator. 
Phys. Rev. X {\bf 1} (2011), 021014.  

\bibitem{SGSS} Sheng, D. N., Gu, Z.-C., Sun, K., Sheng, L.: 
Fractional quantum Hall effect in the absence of Landau levels. 
Nature Commun. {\bf 2} (2011), 389

\bibitem{TKNN} Thouless, D. J., Kohmoto, M., Nightingale, M. P., den 
Nijs, M.: Quantized Hall Conductance in a Two-Dimensional Periodic Potential. 
Phys. Rev. Lett. {\bf 49}, 405--408 (1982) 

\bibitem{WBR} Wu, Y.-L., Bernevig, B. A., Regnault, N.: 
Zoology of fractional Chern insulators. 
Phys. Rev. B {\bf 85} (2012), 075116. 
\end{thebibliography}
\end{document}